%% file: characterization_arxiv.tex
\documentclass[a4paper,11pt]{article}
\usepackage{amsmath,amsfonts,amssymb,amsthm}
\usepackage{latexsym}
\usepackage{color}

\usepackage{graphicx} 
\usepackage{caption}
\usepackage{subcaption}

\usepackage{hyperref}
\usepackage{enumerate}

\usepackage[margin=2.5cm]{geometry}

\input{macros}

\numberwithin{equation}{section}
\newcommand{\cmpqed}{}
\newcommand{\ofwhat}[1]{Proof of {#1}}
\newcommand{\videoref}{\cite{suppl:video}}
\newcommand{\sidecaption}{\centering}
\newenvironment{theopargself}{}{}

\newtheorem{definition}{Definition}[section]
\newtheorem{lemma}[definition]{Lemma}
\newtheorem{proposition}[definition]{Proposition}
\newtheorem{theorem}[definition]{Theorem}

\title{Characterization of local observables in integrable quantum field theories}
\author{Henning Bostelmann\thanks{%
University of York, Department of Mathematics, York YO10 5DD, United Kindom. 
E-mail: \href{mailto:henning.bostelmann@york.ac.uk}{\nolinkurl{henning.bostelmann@york.ac.uk}}
}
 \and Daniela Cadamuro\thanks{%
University of Bristol, School of Mathematics, University Walk, Bristol BS8 1TW, United Kingdom.
E-mail: \href{href:dc13950@bristol.ac.uk}{\nolinkurl{dc13950@bristol.ac.uk}}
}
}

\date{March 5, 2015}

\begin{document}

\maketitle

\begin{abstract}
  \input{abstract}
\end{abstract}


\input{introduction}

\input{generaldefinitions}

\input{tools}
  \input{graphs}

  \input{boundaryvalues}

\input{wedges}

\input{doublecones}

\input{conclusions}

\section*{Acknowledgements}
\input{acknowledgements}


\bibliographystyle{alpha} 
\bibliography{../../integrable,suppl_arxiv}

\end{document}

%% file: macros.tex
\definecolor{detailsgray}{gray}{0.3}

\newcounter{todo}

\def \cbb{\mathbb{C}}

\def \nbb{\mathbb{N}}
\def \rbb{\mathbb{R}}

\def \zbb{\mathbb{Z}}

\def \bcal {\mathcal{B}}
\def \ccal {\mathcal{C}}
\def \dcal {\mathcal{D}}

\def \gcal {\mathcal{G}}
\def \hcal {\mathcal{H}}
\def \ical {\mathcal{I}}

\def \kcal {\mathcal{K}}

\def \ncal {\mathcal{N}}
\def \ocal {\mathcal{O}}

\def \qcal {\mathcal{Q}}
\def \rcal {\mathcal{R}}
\def \scal {\mathcal{S}}

\def \ucal {\mathcal{U}}

\def \wcal {\mathcal{W}}

\def \bfk  {\mathfrak{B}}

\newcommand{\thetav}{\boldsymbol{\theta}}
\newcommand{\etav}{\boldsymbol{\eta}}
\newcommand{\lambdav}{\boldsymbol{\lambda}}
\newcommand{\nuv}{\boldsymbol{\nu}}
\newcommand{\zetav}{\boldsymbol{\zeta}}
\newcommand{\piv}{\boldsymbol{\pi}}
\newcommand{\xiv}{\boldsymbol{\xi}}

\newcommand{\av}{\boldsymbol{a}}
\newcommand{\bv}{\boldsymbol{b}}
\newcommand{\cv}{\boldsymbol{c}}
\newcommand{\ev}{\boldsymbol{e}}
\newcommand{\nv}{\boldsymbol{n}}

\newcommand{\xv}{\boldsymbol{x}}
\newcommand{\yv}{\boldsymbol{y}}
\newcommand{\zv}{\boldsymbol{z}}
\newcommand{\zerov}{\boldsymbol{0}}

\def \. { \,\! }

\def\clap#1{\hbox to 0pt{\hss#1\hss}}

\def\mathclap{\mathpalette\mathclapinternal}

\def\mathclapinternal#1#2{\clap{$\mathsurround=0pt#1{#2}$}}

\def \cdotarg { \, \cdot \, }

\DeclareMathOperator{\im}{Im}
\DeclareMathOperator{\re}{Re}

\newcommand{\idop}{\boldsymbol{1}}

\def \st {^\ast}

\DeclareMathOperator{\supp}{supp}

\newcommand{\vn}[1]{^{(\smash{#1})}}

\newcommand{\rhs}[1]{\mathrm{r.h.s.}\eqref{#1}}

\newcommand{\Hil}{\mathcal{H}}

\newcommand{\fpn}{\Hil^{\mathrm{f}}}
\newcommand{\fpno}{\Hil^{\omega,\mathrm{f}}}
\newcommand{\fpnp}{P^\mathrm{f}}

\newcommand{\boundedops}{\bfk(\Hil)}

\newcommand{\qf}{\mathcal{Q}}

\newcommand{\A}{\mathcal{A}}
\newcommand{\M}{\mathcal{M}}

\newcommand{\hscalar}[2]{\langle #1 , #2 \rangle }

\newcommand{\gnorm}[2]{\lVert #1 \rVert_{#2}}

\newcommand{\onorm}[2]{\lVert #1 \rVert_{#2}^\omega}

\newcommand{\leftwedge}{\mathcal{W}'}
\newcommand{\rightwedge}{\mathcal{W}}

\newcommand{\zd}{z^\dagger}

\newcommand{\lvector}[2]{\boldsymbol{\ell}_{#1}(#2)}
\newcommand{\rvector}[2]{\boldsymbol{r}_{#1}(#2)}

\newcommand{\cmeA}[2]{f_{#1}^{\lbrack #2 \rbrack}}
\newcommand{\cmeB}[2]{f_{#1}{\lbrack #2 \rbrack}}
\newcommand{\cme}[2]{\mathchoice{\cmeA{#1}{#2}}{\cmeB{#1}{#2}}{\cmeB{#1}{#2}}{\cmeB{#1}{#2}}}

\newcommand{\oa}{\varpi}

\DeclareMathOperator{\strip}{S}

\newcommand{\tube}{\mathcal{T}}

\newcommand{\perms}[1]{\mathfrak{S}_{#1}}

\DeclareMathOperator*{\res}{res}

\DeclareMathOperator{\ich}{ich}
\DeclareMathOperator{\ach}{ach}


%% file: abstract.tex

Integrable quantum field theories in 1+1 dimensions have recently become amenable to a rigorous construction, but many questions about the structure of their local observables remain open. 
Our goal is to characterize these local observables in terms of their expansion coefficients in a series expansion by interacting annihilators and creators, similar to form factors. 
We establish a rigorous one-to-one characterization, where locality of an observable is reflected in analyticity properties of its expansion coefficients; this includes detailed information about the high-energy behaviour of the observable and the growth properties of the analytic functions.
Our results hold for generic observables, not only smeared pointlike fields, and the characterizing conditions depend only on the localization region -- we consider wedges and double cones -- and on the permissible high energy behaviour. 

%% file: introduction.tex
\section{Introduction}\label{sec:intro}

\subsection{Background}

The structure of local observables in relativistic quantum theory is a longstanding open problem. While the notion of locality as such is quite straightforward to formulate, e.g., in terms of the Wightman \cite{StrWig:PCT} or Haag-Kastler \cite{Haa:LQP} axioms, already the very existence of models with local observables is a hard mathematical question. In fact, beyond interaction-free models, rigorous existence results are known only in simplified situations in low-dimensional spacetime. Here a recent focus has been on so-called \emph{integrable quantum field theories} in 1+1 dimensions; see e.g.~\cite{Lechner:2008,BostelmannLechnerMorsella:2011,BischoffTanimoto:2013}. 

Integrable quantum field theories are often defined by means of a classical Lagrangian, and their local observables -- in the form of pointlike localized fields -- are traditionally constructed in terms of their form factors, i.e., specific matrix elements of the field in asymptotic scattering states. While these form factors have been computed in several classes of models \cite{FringMussardoSimonetti:1993,BabujianFringKarowskiZapletal:1993,BabujianFoersterKarowski:2013}, the local fields or their $n$-point functions are infinite series in the form factors, and it remains open whether these series converge in a meaningful way.

A very different approach, advocated by Schroer and Wiesbrock \cite{SchroerWiesbrock:2000-1}, is based on the notion of fields intrinsically localized in infinitely extended spacelike wedges, which have a much simpler structure. As there, let us restrict ourselves to a theory describing one species of scalar bosons of mass $\mu > 0$ without bound states. One considers the fields
\begin{equation}\label{eq:introfields}
 \phi(x) = \int d\theta \, \Big(e^{i p (\theta)\cdot x} \zd(\theta) + e^{-ip(\theta)\cdot x} z(\theta) \Big), 
\quad
\phi'(x) = J \phi(-x)  J,
\end{equation}
where $\zd(\theta)$, $z(\theta)$ are ``interacting'' creators and annihilators, depending on rapidity $\theta$, that fulfill the $S$-dependent Zamolodchikov-Faddeev relations rather than the CCR (see Sec.~\ref{sec:zgen}); $J$ is the PCT operator. One notices that 
\begin{equation}
  [\phi(x), \phi'(y)] = 0 \quad \text{if $(x-y)^2< 0$ and $x_1 < y_1$ },
\end{equation}
that is, if $x$ is spacelike \emph{to the left} of $y$. This allows us to interpret the field $\phi'(y)$ as localized in a wedge $\rightwedge_y$ with tip at $y$ extending to the right, and $\phi(x)$ in a wedge $\wcal_x'$ extending to the left.\footnote{For consistency with the literature, we stick with this slightly unintuitive usage of primed vs.~unprimed quantities.} Local observables in bounded regions can then be defined as those relatively local to $\phi$ and $\phi'$. (See Fig.~\ref{fig:wedges} for the concept, and Sec.~\ref{sec:locality} for a mathematical definition.) 

\begin{figure}\sidecaption
 \includegraphics[width=0.65\textwidth]{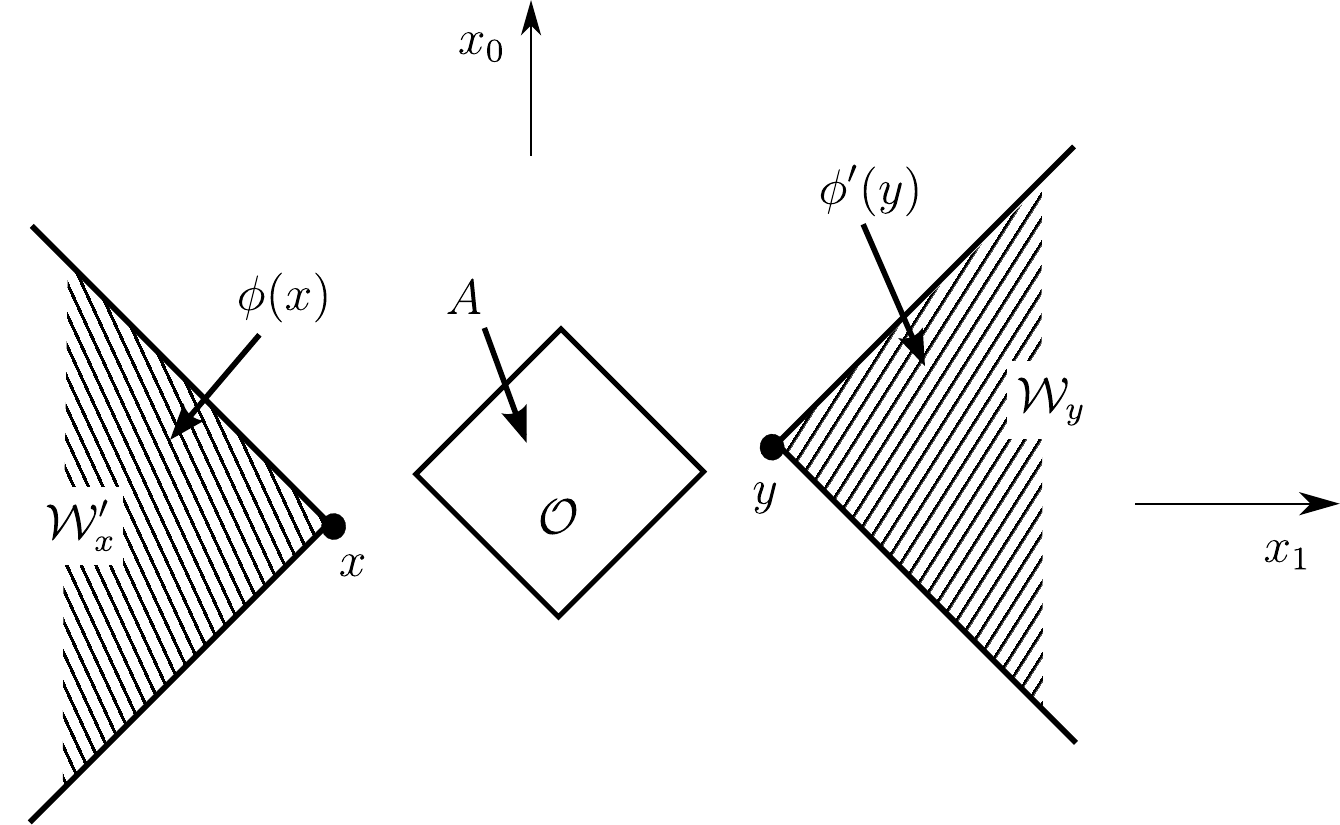}
 \caption{Localization regions of the left field $\phi(x)$, the right field $\phi'(y)$, and of a generic local observable $A$. The bounded region $\ocal$ is spacelike separated from both wedges $\wcal_x'$ and $\wcal_y$, corresponding to the fact that $A$ commutes with both $\phi(x)$ and $\phi'(y)$.}\label{fig:wedges}
\end{figure}

Lechner \cite{Lechner:2008} used this approach to prove the existence of local observables as bounded operators in a class of models, including the Ising and sinh-Gordon models. The proof uses rather abstract methods, and does not give an explicit construction of examples of local operators. Nonetheless, the abstract existence result is sufficient to show asymptotic completeness and to fully determine the scattering theory of the model. Note that in this context, the models are no longer defined in terms of a Lagrangian, but in the spirit of an inverse scattering problem: The particle spectrum and the two-particle scattering matrix $S$ are seen as inputs to the construction, entering the definition of $\zd$ and $z$. 

We will extract more information on the properties of these observables $A$ by expanding them into a series of normal-ordered monomials in $\zd,z$,
\begin{equation}\label{eq:expansionintro}
A = \sum_{m,n=0}^\infty \int \frac{d\thetav \, d\etav}{m!n!} \cme{m,n}{A}(\thetav,\etav) z^{\dagger}(\theta_1)\cdots z^\dagger(\theta_m) z(\eta_1)\cdots z(\eta_n),
\end{equation}
where we denote $\thetav=(\theta_1,\ldots,\theta_m)$, $\etav=(\eta_1,\ldots,\eta_n)$. This expansion is possible for every operator $A$ in a certain regularity class, independent of its localization properties  \cite{BostelmannCadamuro:expansion}. It is similar but not identical to the well-known form factor expansion, inasmuch as the form factor expansion is based on asymptotic ``free'' creators and annihilators $a^\dagger_\mathrm{in}$, $a_\mathrm{in}$, whereas \eqref{eq:expansionintro} uses the ``interacting'' objects $\zd$, $z$ instead. Correspondingly, the expansion coefficients $\cme{m,n}{A}(\thetav,\etav)$ agree with the form factors for certain regions of the arguments. There is an explicit (if intricate) expression for the  $\cme{m,n}{A}(\thetav,\etav)$ as linear functionals of $A$, and their behaviour under space-time symmetry transformations is known \cite{BostelmannCadamuro:expansion}.

\subsection{Aims and results}

Our aim in the present paper is to gain more insight into the structure of the local operators, complementing the abstract construction. Specifically, we will characterize local operators $A$ in terms of analyticity properties of their expansion coefficients $\cme{m,n}{A}$ in the series expansion \eqref{eq:expansionintro}. This is somewhat similar to the well-known analysis of analyticity of form factors \cite{Smirnov:1992}. However, our results differ from previous work in two essential aspects.

First, they are valid for any local observable $A$, and not restricted to quantum fields that are localized at space-time points. Rather, our analysis is based on the abstract localization region of $A$ only, keeping track of its size and shape. Specifically, we will derive results for left wedges and for double cones.

Second, our characterization accounts for the functional analytic properties of the operators (or quadratic forms) $A$. In particular, the high-energy behavior of $A$ and its influence on the asymptotic growth of the expansion coefficients $\cme{m,n}{A}$ is described in detail.

To illustrate the results, let us first consider the case that $A$ is localized in a left wedge $\leftwedge_r$ with tip at $(0,r)$ on the time-0 axis. 
The left wedge, as opposed to the right wedge, is a natural choice here: one notes that if $A$ is the left field $\phi(x)$, cf.~\eqref{eq:introfields}, then the expansion \eqref{eq:expansionintro} is rather simple and consists of only of two terms; but in the case of the right field, $A=\phi'(x)$, the expansion coefficients have a much more complicated structure, cf.~\cite[Prop.~3.11]{BostelmannCadamuro:expansion}.

Now if a general observable $A$ is localized in $\wcal'_{r}$, one will expect from \cite{Lechner:2008,BostelmannCadamuro:expansion} that the $\cme{m,n}{A}$ are boundary values of a common analytic function, i.e.,
one has
\begin{equation}
   \cme{m,n}{A}(\thetav,\etav)=F_{m+n}(\thetav+i\zerov,\etav+i\piv-i\zerov).
\end{equation}
These $F_k$ fulfill the following properties, which we write somewhat informally for the moment.
\begin{enumerate}[(1)]
 \item They are analytic in the tube region $0 < \im\zeta_1 < \ldots < \im \zeta_k < \pi$.
 \item As a consequence of the Zamolodchikov-Faddeev relations, they are $S$-symmetric, that is, their boundary values at real arguments fulfill
    \begin{equation}\label{eq:ssymmwedge}
       F_k(\theta_1,\ldots,\theta_{j+1},\theta_j,\ldots,\theta_k )  
       =  S(\theta_{j}-\theta_{j+1}) F_k(\theta_1,\ldots,    \theta_{j},\theta_{j+1},\ldots,\theta_k ). 
    \end{equation}
    This is also known as \emph{Watson's equation}.
  \item \label{it:wedgeb1}
          Their boundary values at real arguments, $\thetav \mapsto F_k(\thetav+i \zerov)$, are square integrable.
        (This follows if $A\Omega$ has finite norm, which we shall assume here.)
  \item \label{it:wedgeb2}
     Their growth behaviour at real infinity is essentially given by
    \begin{equation}
         \lvert F(\thetav+i\lambdav) \rvert \sim \prod_{j=1}^k e^{\mu r \cosh \theta_j \sin \lambda_j}. 
    \end{equation}
\end{enumerate}

Our aim is to formulate a full characterization of wedge-local observables, that is, to prove that a suitable variant of  (1)-(4) holds for $F_k$ \emph{if and only if} $A$ is localized in the wedge $\wcal_{r}'$. 
Evidently, the conditions need some mathematical refinement; in particular, the bounds (\ref{it:wedgeb1}) and (\ref{it:wedgeb2}) need to take care of the high-energy behaviour of the observable. A precise version is given as condition (FW) in Def.~\ref{def:conditionFW}. Given these, we can indeed find a full characterization (Thm.~\ref{theorem:wedgeequiv}).

The situation becomes more intricate when localization is restricted to a bounded region, say, a double cone $\ocal_r$ of radius $r>0$ around the origin. From \cite{Lashkevich:1994,SchroerWiesbrock:2000-1}, one expects that the functions $F_k$ behave as follows.

\begin{enumerate}[(1)]
 \item They are meromorphic on all of $\cbb^k$, and analytic on the tube $\im \zeta_1 < \ldots < \im \zeta_k < \im \zeta_1 + 2 \pi$, except for possible first-order poles at $\zeta_n-\zeta_m=i \pi$ (the so-called \emph{kinematic poles}).
 \item In generalization of \eqref{eq:ssymmwedge}, one has for all complex arguments $\zeta_1,\ldots,\zeta_k$,
      \begin{equation}\label{eq:ssymmdc}
          F_k(\zeta_1,\ldots,\zeta_{j+1},\zeta_j,\ldots,\zeta_k )  =  S(\zeta_{j}-\zeta_{j+1}) F_k(\zeta_1,\ldots,\zeta_{j},\zeta_{j+1},\ldots,\zeta_k ). 
      \end{equation}
  \item They are \emph{$S$-periodic}, i.e.,  
       \begin{equation}\label{eq:introperio}
         F_k(\zeta_1,\ldots,\zeta_{k-1},\zeta_k+2 i \pi )  =  F_k(\zeta_k,\zeta_1,\ldots,\zeta_{k-1}),
       \end{equation}
        of which periodicity-like properties in the other variables follow from \eqref{eq:ssymmdc}. 
  \item \label{it:recursion0}
        Their residue at $\zeta_k-\zeta_1=i\pi$ is given by
        \begin{equation}\label{eq:residue1k}
            \res_{\zeta_k-\zeta_1 = i \pi} F_{k}(\zetav)
                =  \frac{1}{2\pi i} \Big(1-\prod_{p=1}^{k} S(\zeta_1-\zeta_p) \Big) F_{k-2}(\zeta_2,\ldots,\zeta_{k-1} );
          \end{equation}
        the residues at the other kinematic poles can again be inferred from \eqref{eq:ssymmdc}.
  \item Again, $\thetav \mapsto F_k(\thetav+i \zerov)$ is square integrable.
  \item Their growth behaviour at real infinity is essentially given by
    \begin{equation}
         \lvert F(\thetav+i\lambdav) \rvert \sim \prod_{j=1}^k e^{\mu r \cosh \theta_j \lvert\sin \lambda_j\rvert}. 
    \end{equation}
\end{enumerate}

The essential new feature are the \emph{recursion relations} (\ref{it:recursion0}), linking $F_k$ to $F_{k-2}$. In particular, they enforce that the sequence $F_k$ cannot terminate, except in very specific cases of constant $S$. For completeness, we note that the $F_k$ can have further singularities on hyperplanes in $\cbb^k$, stemming from poles in the scattering function $S$ ``outside the physical strip''.

As in the wedge-local case, we will formulate a precise variant of these properties as condition (FD) in Def.~\ref{def:conditionFD}, and we can indeed prove a full equivalence of these conditions with locality of $A$ in $\ocal_r$ in  Thm.~\ref{theorem:doubleconeequiv}. An essential ingredient of the proof is that a double cone is the intersection of two wedges, $\ocal_r = \wcal_{-r} \cap \wcal_{r}'$; and $A$ is local in $\ocal_r$ if and only if it is local in both $\wcal_{-r}$ and $\wcal_{r}'$. This makes it useful to analyze the wedge-local case first.

\subsection{Methods}

While the statement of locality conditions above was straightforward but somewhat heuristic, a rigorous formulation  needs careful analysis of the topological properties of all objects involved.

First of all, regarding the observables $A$, one needs to clarify which operator-theoretic class these objects should belong to. In fact, choosing the class of bounded operators (as in \cite{Lechner:2008}) is not useful in our context, since the expansion \eqref{eq:expansionintro} is based on unbounded objects $z(\theta),\zd(\theta)$. We try to be as general as possible, and work within a class of quadratic forms $A$, general enough to include (but not restricted to) smeared Wightman fields. While these quadratic forms can be unbounded, it is important that they cannot be arbitrarily divergent either, in particular with respect to their high-energy behaviour. With applications in mind \cite[Sec.~9]{Cadamuro:2012}, we will aim here for the most singular high-energy behaviour that is still compatible with localization, namely the generalized $H$-bounds proposed by Jaffe \cite{Jaffe:1967}; see Sec.~\ref{sec:jaffe} for details. Along with the choice 
of a class of quadratic forms, we also need to review our notion of locality, since a commutator in the usual sense is not meaningful between quadratic forms. Instead, we will define it by relative locality to the wedge-local fields $\phi$ and $\phi'$ (Sec.~\ref{sec:locality}).

The high energy behavior of the observables $A$ will also be reflected in the growth of the functions $F_k$ at large $|\thetav|$. This requires us to introduce some further norms on their boundary distributions, generalizing the $L^2$ norm (Sec.~\ref{sec:zgen}), before we can formulate the precise locality conditions.  

As an intermediate step between distributions $\cme{m,n}{A}$ and analytic functions $F_k$, we also work with the boundary distributions of $F_k(\zetav)$ at the manifold 
\begin{equation}
 0 = \im\zeta_1 = \ldots = \im\zeta_{j-1}< \im \zeta_j < \im \zeta_{j+1}= \ldots = \im \zeta_k = \pi, \quad j \in\{ 1,\ldots, k\}.
\end{equation}
These distributions will be denoted $T_k(\zetav)$ below. Technically, the $T_k$ are distributions on compactly supported test functions but analytic in one variable $\zeta_j$. We will formalize this in terms of \emph{CR distributions} on certain graph domains, a concept that we explain in Sec.~\ref{sec:graphs}. Locality conditions (TW) and (TD) will be formulated for these $T_k$ as well (Def.~\ref{def:conditionTW} and \ref{def:conditionTD} respectively), and shown to be equivalent to conditions (FW) and (FD) for the meromorphic functions $F_k$.

The remaining paper is organized as follows. We start by introducing our general mathematical setting in Sec.~\ref{sec:preliminaries}, specifying the details of our quantum field theoretical models and recalling the details of the expansion \eqref{eq:expansionintro}. We also develop the necessary technical tools for analytic functions and CR distributions in Sec.~\ref{sec:tools}. The characterization theorem for  local operators is then, for the case of wedge localization, formulated and proved in Sec.~\ref{sec:wedges}, and in extension of those results, for double cones in Sec.~\ref{sec:doublecones}. An animation \videoref{} complements Sec.~\ref{sec:doublecones}. We end with conclusions in Sec.~\ref{sec:conclusion}.

The present article is based in parts on the Ph.D.~thesis of one of the authors \cite{Cadamuro:2012}.

%% file: generaldefinitions.tex
\section{Mathematical setting for integrable QFTs}\label{sec:preliminaries}

We will first fix our mathematical setting and define the models of quantum field theory in question. We mainly follow \cite{Lechner:2008,BostelmannCadamuro:expansion}, and will recall the main results of those papers. 

We consider quantum field theory on 1+1 dimensional Minkowski space, with the indefinite scalar product written as $x \cdot y = x_0 y_0-x_1 y_1$. The models in question are integrable quantum field theories, specified by their particle spectrum and two-particle $S$ matrix. As in \cite{Lechner:2008}, we consider only one species of massive scalar particle, so that the two-particle scattering matrix is just a complex valued function $S$, which enters our construction as a parameter. We will specify the required properties of $S$ and the construction of the associated Hilbert space in Sec.~\ref{sec:scatter}.

Within these models, we are going to deal with (unbounded) quadratic forms with localization properties. We allow these to fulfill  generalized $H$ bounds in the sense of Jaffe \cite{Jaffe:1967}; this concept will be explained in Sec.~\ref{sec:jaffe}. We then define the associated spaces of quadratic forms in Sec.~\ref{sec:zgen}, and more importantly, recall from \cite{BostelmannCadamuro:expansion} that they can be expanded into a series of generalized annihilation and creation operators. Section~\ref{sec:locality} deals with locality properties in the models at hand, recalls Lechner's existence result for local operators \cite{Lechner:2008}, and extends the notion of locality to the level of quadratic forms.
 
\subsection{Scattering function and Hilbert space}\label{sec:scatter}

We first explain the properties of the two-particle scattering function $S$. Let $\strip(a,b)$ denote the strip $a < \im \zeta < b$ in the complex plane. We take $S$ to be an analytic function $S: \strip(0,\pi) \to \cbb$ which has a continuous and bounded extension to the closed strip $\overline{\strip(0,\pi)}$, and which fulfills the symmetry relations
\begin{equation}\label{eq:srelat}
  \forall \theta\in \rbb: \quad S(\theta)^{-1}=S(-\theta)=\overline{S(\theta)}=S(\theta+i \pi).
\end{equation}
Evidently, $S$ is of unit modulus on the lines $\rbb$ and $\rbb+i\pi$. Setting $S(\zeta):=S(\zeta+i\pi)^{-1}$ for $\zeta \in \strip(-\pi,0)$, and using \eqref{eq:srelat} and continuity at the boundary of the strip, we can extend $S$ to a $2 \pi i$-periodic meromorphic function on all of $\cbb$.

Since $S$ has no poles on the real line, its restriction to $\rbb$ is in particular smooth, so that the assumptions of \cite[Sec.~2.1]{BostelmannCadamuro:expansion} are fulfilled. However, compared with \cite{BostelmannCadamuro:expansion} we have added analyticity properties of $S$; these are crucial for describing local observables, as will become clear in Sec.~\ref{sec:locality}. We note that we will not need the additional regularity condition imposed in \cite[Def.~3.3]{Lechner:2008}. On the other hand, we stick to the assumption of \cite{Lechner:2008} that $S$ has no poles in the ``physical strip'' $\strip(0,\pi)$. Examples for scattering functions in our class include
\begin{enumerate}[(i)]
 \item $S(\theta)=1$ (the free field),
 \item $S(\theta)=-1$ (the massive Ising model),
 \item $S(\theta)=\dfrac{\sinh \theta - i a}{\sinh \theta + i a}$ with some $a \in (0,1)$ (the sinh-Gordon model),
 \item $S(\theta)=\exp( i a \sinh \theta)$ with some $a > 0$ (an ``exotic'' $S$-matrix used in \cite{GrosseLechner:2007}).
\end{enumerate}

Associated with $S$ and a permutation $\sigma \in \perms{n}$, we introduce the following function of $n$ variables:
\begin{equation}\label{eq:Sperm}
S^\sigma (\thetav) := \prod_{\substack{i<j \\ \sigma(i)>\sigma(j)}} S(\theta_{\sigma(i)}-\theta_{\sigma(j)}).
\end{equation}
Using these $S^\sigma$, one can define a representation of $\perms{n}$ on $L^2(\rbb^n)$ and on other function spaces \cite[Eq.~(3.5)]{Lechner:2008}. We are particularly interested in functions (or distributions) invariant under this representation; that is, functions $f$ of $n$ variables fulfilling
\begin{equation}\label{eq:ssymm}
\forall \sigma \in \perms{n}: \quad f(\thetav)=S^{\sigma}(\thetav)f(\thetav^{\sigma}),
\end{equation}
where $\thetav^{\sigma}=(\theta_{\sigma(1)},\ldots,\theta_{\sigma(n)})$.
We call these functions \emph{$S$-symmetric}. Since the $S^\sigma$ fulfill a simple composition law \cite[Eq.~(2.3)]{BostelmannCadamuro:expansion}, one knows that $f$ is $S$-symmetric when \eqref{eq:ssymm} is verified only for transpositions $\sigma$.

$S$-symmetric functions are relevant for defining the Hilbert space $\hcal$ of our model. Since we consider models with one species of scalar particle with mass $\mu >0$, the single particle space is given by $\Hil_{1}=L^{2}(\mathbb{R},d\theta)$, as in the real scalar free field in rapidity representation. Defining the $n$-particle space $\Hil_n$ as the subspace of $S$-symmetric functions in $L^2(\rbb^n)$, and setting $\Hil_0=\cbb\Omega$, we define the Hilbert space $\hcal$ as $\Hil:=\bigoplus_{n=0}^{\infty}\Hil_{n}$.
The orthogonal projection onto $\Hil_n\subset \Hil$ will be denoted as $P_n$, and we also consider the projectors $\fpnp_n:=\sum_{j=0}^{n}P_j$. The space of finite particle number states, $\fpn := \bigcup_{n}\fpnp_{n}\Hil$, is dense in $\Hil$.

We further recall the representation of the proper Poincar\'e group acting by (anti)unitary operators on $\hcal$. 
Translations and boosts act on $\psi = \oplus_{n=0}^\infty \psi_n\in\Hil$ as
\begin{equation}\label{eq:translboostdef}
(U(x,\lambda)\psi)_{n}(\thetav):= e^{ i p(\thetav) \cdot x } \psi_{n}(\thetav-\lambdav),
\end{equation}
where $\lambdav = (\lambda,\ldots,\lambda)$ and
\begin{equation}
 p(\thetav) := \sum_{k=1}^{n}p(\theta_{k}), \quad p(\theta):=\mu \begin{pmatrix} \cosh\theta \\ \sinh\theta \end{pmatrix} \quad (\thetav\in\rbb^n,\;\theta\in\rbb).
\end{equation}
The positive generator of time translations will be denoted $H$.
The space-time reflection acts by an antiunitary operator $U(j)=:J$ as
\begin{equation}
(U(j)\psi)_{n}(\thetav):= \overline{\psi_{n}(\theta_{n},\ldots,\theta_{1})}.
\end{equation}

For later reference, we fix the conventions for the Fourier transform of functions $g \in \scal(\rbb^2)$ in momentum and rapidity variables:
\begin{equation}\label{eq:ftg}
\tilde g(p) :=\frac{1}{2\pi}\int dx \, g(x)e^{ip\cdot x}, \quad
 g^{\pm}(\theta) := \tilde g (\pm p(\theta)) .
\end{equation}
Then $g^\pm \in \Hil_1$. We will also define the Zamolodchikov-Faddeev operators $\zd(\theta)$, $z(\theta)$ as ladder operators on $\Hil$ in Sec.~\ref{sec:zgen}, but we first need some functional analytic preparations.

\subsection{High energy behavior} \label{sec:jaffe}

The local observables that we will consider are not necessarily bounded operators, rather we will allow quadratic forms that are unbounded in states of high energy. This situation is common in Wightman quantum field theory, where local quantum fields are necessarily unbounded objects \cite{Wightman:1964}. One often confines attention to fields with polynomial energy bounds, i.e., such that $(1+H)^{-\ell} \phi(x) (1+H)^{-\ell}$ is bounded for some $\ell>0$ \cite{FreHer:pointlike_fields}. This choice is subtly related to the choice of test function space for the quantum fields, which is normally taken to be Schwartz space \cite{StrWig:PCT}. It was pointed out by Jaffe \cite{Jaffe:1967} that this choice is possibly too restrictive, and that there is room for generalization: instead of polynomial growth in energy, one can allow ``almost exponential'' growth like $\exp \omega(E)$ with a function $\omega$ that is almost, but not quite, growing linearly in $E$. (See also \cite{ConstantinescuThaleimer:1974}.) A 
mathematical treatment of the associated classes of test functions and distributions (due to Beurling) is given in \cite{Bjoerck:1965}.

In the present paper, with a view to constructing examples of local operators, we wish to treat as general a class of operators as possible, and will therefore adopt Jaffe's framework with some slight variations. We list the properties that we expect the function $\omega$ (the \emph{indicatrix}) to fulfill.

\begin{definition}\label{def:indicatrix}
An \emph{indicatrix}  is a smooth function $\omega:[0,\infty) \to [0,\infty)$ with the following properties.
\begin{enumerate}
\renewcommand{\theenumi}{($\omega$\arabic{enumi})}
\renewcommand{\labelenumi}{\theenumi}
\item \label{it:omegamonoton}
$\omega$ is monotonously increasing;
\item \label{it:sublinear}
$\omega(p+q)\leq \omega(p)+\omega(q)$ for all $p,q \geq 0$ \emph{(sublinearity)};
\item \label{it:omegagrowth}
$\displaystyle{\int_{0}^{\infty}\frac{\omega(p)}{1+p^{2}}\;dp < \infty}$
\emph{(Carleman's criterion)}.
\end{enumerate}
We call $\omega$ an \emph{analytic indicatrix} if, in addition, there exists a function $\oa$ on the upper half plane $\rbb + i [0,\infty)$, analytic in the interior and smooth at the boundary, such that
\begin{enumerate}
\setcounter{enumi}{3}
\renewcommand{\theenumi}{($\omega$\arabic{enumi})}
\renewcommand{\labelenumi}{\theenumi}

\item \label{it:omegaeven}
$\re\oa(p)=\re\oa(-p)$ for all $p \geq 0$;

\item \label{it:omegaestimate}
There exist $a_\omega,b_\omega>0$ such that $\omega(|z|) \leq \re\oa(z) \leq a_\omega \omega(|z|)+b_\omega$ for all $z \in \rbb+i[0,\infty)$.

\end{enumerate}
\end{definition}

These conditions are stronger than in \cite{BostelmannCadamuro:expansion}, where we required only \ref{it:omegamonoton} and \ref{it:sublinear}; the extra conditions are added for the purpose of describing locality, as will become clear below. Still, the conditions allow for a wide range of examples. One of them, corresponding to the usual polynomial growth behavior in energy, is the following for some $\beta>0$:
\begin{equation}\label{eq:omegalog}
\omega(p) = \beta\log(1+p), \quad
\oa(z)=2\beta\big(\operatorname{Log}(i+z)+1 \big).
\end{equation}
A second class of examples with stronger growth in $p$ is, with $0 < \alpha < 1$,
\begin{equation}\label{eq:omegapower}
\omega(p)= p^{\alpha}\cos \frac{\alpha \pi}{2},\quad
\oa(z)= i^{-\alpha}(z+i)^{\alpha}.
\end{equation}

We will now discuss spaces of functions with support in fixed regions of spacetime and with high-energy behavior controlled by a given indicatrix $\omega$. Let $\ocal$ be an open set in Minkowski space. We set $\dcal(\ocal) := C^{\infty}_{0}(\ocal)$ and
\begin{equation}
 \dcal^\omega(\ocal) := \{ f \in \dcal(\ocal) : \theta \mapsto e^{\omega(\cosh \theta)} f^\pm(\theta) \text{ is bounded and square integrable} \}.
\end{equation}
We don't equip $\dcal^\omega(\ocal)$ with a topology; see however \cite{Bjoerck:1965,ConstantinescuThaleimer:1974} on how to topologize similarly defined spaces. The interesting question for us is the size of $\dcal^\omega(\ocal)$. If $\omega$ is of the form \eqref{eq:omegalog}, or bounded by this, then $e^{\omega(p)}$ is bounded by a power of $p$, and evidently $\dcal^\omega(\ocal) =  \dcal(\ocal)$. For faster growing $\omega$, it is not clear a priori that $\dcal^\omega(\ocal)$ contains any non-zero element. But in fact, it is condition \ref{it:omegagrowth} that guarantees nontriviality: One even finds ``local units'' in $\dcal^\omega(\ocal)$, i.e., functions $f$ with $0 \leq f \leq 1$, with $f=1$ on any given compact set $\kcal\subset \ocal$, and $f=0$ outside any given neighborhood of $\kcal$ (\cite[Theorem~1.3.7]{Bjoerck:1965} holds analogously). By using convolutions with such functions, one finds that $\dcal^\omega(\ocal)$ is actually dense in $\dcal(\ocal)$, in the $\dcal(\ocal)$ topology.

For functions in $\dcal^\omega(\ocal)$, one can derive Paley-Wiener type estimates on their Fourier transform \cite[Sec~1.4]{Bjoerck:1965}. We use the following variant in our context, where the Fourier transform is defined as in \eqref{eq:ftg}, and where the region $\ocal$ is specifically the standard right wedge, 
\begin{equation}\label{eq:rightwedge}
  \rightwedge = \{ x \in \rbb^2: x_1>|x_0| \}.
\end{equation}

\begin{proposition}\label{prop:omegapw}
Let $\omega$ be an analytic indicatrix and $f \in \dcal^\omega(\wcal)$. Then $ f^-$ extends to an analytic function on the strip $\strip(0,\pi)$, continuous on its closure, and one has $f^-(\theta+i\pi)=f^+(\theta)$. For fixed $\ell\in\nbb_0$, there exists $c>0$ such that
\begin{equation}\label{eq:fmstrip}
  \Big\lvert \frac{d^\ell f^-}{d\zeta^\ell} (\theta+ i \lambda) \Big\rvert \leq c (\cosh \theta)^\ell e^{-\omega(\cosh \theta)/a_\omega}
\quad \text{for all }\theta \in \rbb, \; \lambda \in [0,\pi].
\end{equation}
\end{proposition}

\begin{proof}
 Since $f$ has compact support, $\tilde f$ and $f^\pm$ are actually entire, and the relation $f^-(\zeta \pm i\pi)=f^+(\zeta)$ follows by direct computation. We first prove the bound \eqref{eq:fmstrip} in the case $\ell=0$. To that end, we consider the function $g$ on $\strip(0,\pi)$ defined by
\begin{equation}
   g(\zeta) := f^-(\zeta) e^{\oa(\sinh \zeta)/a_\omega}.
\end{equation}
(Note that $\sinh(\cdotarg)$ maps the strip into the upper half plane.)
For $\zeta=\theta+i\lambda$ in the closed strip, one has
\begin{equation}\label{eq:omegacd}
 \re \oa(\sinh \zeta)/a_\omega \leq \omega( \lvert\sinh \zeta\rvert)+b_\omega/a_\omega \leq \omega(\cosh \theta)+b_\omega/a_\omega,
\end{equation}
where \ref{it:omegaestimate} and \ref{it:omegamonoton} have been used.
Since $f \in \dcal^\omega(\rightwedge)$, it follows that
\begin{equation}
   \sup_{\theta\in\rbb} |g(\theta)| \leq e^{b_\omega/a_\omega} \sup_{\theta\in\rbb} |e^{\omega(\cosh\theta)}f^{-}(\theta)|  < \infty.
\end{equation}
That is, $g$ is bounded on $\rbb$, and by a similar computation involving $f^+$ and \ref{it:omegaeven}, it is bounded on the line $\rbb+i\pi$ as well. In the interior of the strip, we know that $f^-(\zeta)$ is bounded since $\supp f \subset \rightwedge$; therefore,
\begin{equation}\label{eq:gstripest}
  \lvert g(\theta+i\lambda) \rvert  \leq e^{\omega(\cosh \theta)+b_\omega/a_\omega} \sup_{\zeta'\in\strip(0,\pi)} |f^- (\zeta') |,
\end{equation}
where \eqref{eq:omegacd} has been employed. We note that $\omega(p)=o(p)$ as $p \to \infty$ due to \ref{it:omegagrowth}, \ref{it:omegamonoton}. Using this in \eqref{eq:gstripest}, we can employ a Phragm\'en-Lindel\"of argument to show that $g$ is actually bounded on the strip, and takes its maximum at the boundary. (We can apply \cite[Theorem~3]{HardyRogosinski:1946} to the subharmonic function $\log |g|$.) In other words,
\begin{equation}\label{eq:fmbound}
  \lvert f^-(\zeta) \rvert \leq c \, \lvert e^{-\oa(\sinh \zeta)/a_\omega} \rvert\quad \text{for all } \zeta \in \strip(0,\pi)
\end{equation}
with some $c>0$. We estimate
\begin{equation}
  \re \oa(\sinh \zeta) \geq \omega(|\sinh \zeta|) \geq \omega (\cosh \theta - 1) \geq \omega(\cosh\theta)-\omega(1),
\end{equation}
where \ref{it:omegaestimate}, \ref{it:omegamonoton}, \ref{it:sublinear} have been used. Inserted into \eqref{eq:fmbound}, this gives \eqref{eq:fmstrip} for $\ell=0$.

For $\ell > 0$, we proceed as follows. Fix some $\zeta =\theta+i\lambda \in \overline{\strip(0,\pi)}$. We estimate the derivative of $f$ at $\zeta$ using Cauchy's formula: For any $0 < t < \pi$,
\begin{equation}\label{eq:cauchyest}
   \Big\lvert \frac{d^\ell f^-}{d\zeta^\ell} (\zeta) \Big\rvert \leq \ell! \,t^{-\ell} \sup_{|\zeta-\zeta'|=t} |f^-(\zeta')|.
\end{equation}
Note that here $\zeta'$ lies in $\strip(-\pi,2\pi)$ but not necessarily in $\strip(0,\pi)$, so that we will need to obtain estimates for $f^-$ on this extended strip. To that end, choose $s>0$ such that $ \supp f \subset \wcal_s'  = \wcal'+(0,s)$, and set $h(x):=f((0,s) -x)$, noting that
\begin{equation}\label{eq:fhexp}
    f^-(\zeta' \pm i \pi) = e^{i \mu s \sinh \zeta'} h^-(\zeta') .
\end{equation}
Since $h \in \dcal^\omega(\rightwedge)$, the result \eqref{eq:fmstrip} for $\ell=0$ applies to both $h$ and $f$. From \eqref{eq:fhexp} it then follows that
\begin{equation}
 \forall \zeta' =\theta' + i \lambda' \in \strip(-\pi,2\pi):\quad \lvert f^-(\zeta')\rvert \leq
c  \,e^{-\omega(\cosh \theta')/a_\omega} \,e^{\mu s \cosh \theta' \lvert \sin \lambda' \rvert}
\end{equation}
with some $c>0$. Inserting into \eqref{eq:cauchyest}, this yields for large $\theta$,
\begin{equation}
   \Big\lvert \frac{d^\ell f^-}{d\zeta^\ell} (\zeta) \Big\rvert
  \leq c\, \ell! \, t^{-\ell} \,e^{\mu s t \cosh (\theta+t) } e^{-\omega(\cosh(\theta-t))/a_\omega}.
\end{equation}
We now choose $t = 1/\cosh\theta$. With this $t$, we can find $c'>0$ such that $\cosh(\theta-t) \geq \cosh \theta - c'$, $\cosh (\theta+t) \leq \cosh \theta + c'$ for all $\theta$. Employing \ref{it:sublinear}, we obtain a constant $c''>0$ such that
\begin{equation}
   \Big\lvert \frac{d^\ell f^-}{d\zeta^\ell} (\zeta) \Big\rvert
  \leq c'' (\cosh\theta)^{\ell} e^{-\omega(\cosh\theta)/a_\omega}.
\end{equation}
For large $-\theta$, the computation is analogous. This yields \eqref{eq:fmstrip}.
\cmpqed\end{proof}

\subsection{Quadratic forms and operator expansions}\label{sec:zgen}

We will now define a space of quadratic forms on (dense subsets of) $\Hil$ and other structures related to a fixed indicatrix $\omega$, and we will recall how these quadratic forms can be expanded in Zamolodchikov-Faddeev operators as in \eqref{eq:expansionintro}. In most parts, this summarizes the setting and main results of \cite{BostelmannCadamuro:expansion}.

We first introduce some norms on distributions,\footnote{%
As in \cite{BostelmannCadamuro:expansion}, we will usually write all distributions as integrals of formal kernels. This is convenient since many of them will actually be boundary values of analytic functions. We will also often denote these kernels like maps, e.g., $\thetav \mapsto f(\thetav)$, although they are of course not maps in the strict sense; the purpose is to indicate the ``formal integration variables'' where needed.}
relating to the indicatrix $\omega$. 
For $f \in \dcal(\rbb^{m+n})'$, we set
\begin{align}\label{eq:crossnorm}\notag
\gnorm{f}{m \times n} &:= \sup \Big\{ \big\lvert \! \int f(\thetav,\etav) g(\thetav) h(\etav)  d\thetav d\etav \,\big\rvert : \\ 
&\qquad \qquad \qquad g \in \dcal(\rbb^m), \, h \in \dcal(\rbb^n), \,  \gnorm{g}{2} \leq 1, \, \gnorm{h}{2} \leq 1\Big\},
\\ \notag
\onorm{f}{m \times n} &:= \frac{1}{2} \gnorm{(\thetav,\etav) \mapsto e^{-\omega(E(\thetav))}f(\thetav,\etav)}{m \times n} 
\\ \label{eq:omegacrossnorm} &\qquad \qquad \qquad + \frac{1}{2} \gnorm{(\thetav,\etav) \mapsto f(\thetav,\etav)e^{-\omega(E(\etav))}}{m \times n}
\end{align}
where these expressions are finite; here  $E(\thetav) := p_0(\thetav)/\mu$.
We also consider the related, but $\omega$-independent norm for $f \in \dcal(\rbb^k)$,
\begin{equation}\label{eq:fullcrossnorm}
  \gnorm{f}{\times} := \sup\big\{ \Big\lvert  \int d\thetav f(\thetav) g_1(\theta_1)\cdots g_k(\theta_k) \Big\rvert : g_1,\ldots,g_k \in \dcal(\rbb), \, \gnorm{g_j}{2} \leq 1 \big\}.
\end{equation}
We note the computational rule \cite[Lemma~2.8]{Cadamuro:2012}
\begin{equation}
 \label{eq:crossnormfactor}
  \gnorm{\thetav \mapsto f(\thetav) \prod_j f_j(\theta_j)}{\times} \leq \gnorm{f}{\times} \prod_j \gnorm{f_j}{\infty}.
\end{equation}
For test functions $g \in \dcal(\rbb^k)$, we also set
\begin{equation}
  \onorm{g}{2} := \gnorm{\thetav \mapsto e^{\omega(E(\thetav))}g(\thetav)}{2}. 
\end{equation}
Further, we repeat the relevant $\omega$-related structures on the Hilbert space level. We denote $\Hil^\omega := \{\psi \in \Hil : \|e^{\omega(H/\mu)} \psi\| < \infty\}$. For $n\in\nbb_0$, we write $\Hil^\omega_n := \Hil^\omega \cap \Hil_n$, and $\fpno := \cup_n \Hil^{\omega}_n $; the latter space is dense in $\Hil$. 
By $\qf^\omega$, we denote the space of sesquilinear forms $A : \fpno\times\fpno \to \cbb$ such that for any $n \in \nbb_0$,
\begin{equation}\label{eq:aomeganorm}
 \gnorm{A}{n}^{\omega} := \frac{1}{2} \gnorm{\fpnp_n A e^{-\omega(H/\mu)}\fpnp_n}{} + \frac{1}{2} \gnorm{\fpnp_n e^{-\omega(H/\mu)} A \fpnp_n}{} < \infty.
\end{equation}

We also recall the representation of the Zamolodchikov-Faddeev algebra that underlies our models \cite{Lechner:2008}.  We define modified creation and annihilation operators $\zd(\theta)$, $z(\eta)$ by their action on $\psi \in \Hil_n \cap \dcal(\rbb^n)$ as
\begin{equation}
\begin{aligned}
(\zd(\theta)\psi)(\lambdav) &= \frac{\sqrt{n+1}}{(n+1)!} \,\sum_{\sigma\in\perms{n+1}} S^\sigma(\lambdav) \delta(\theta-\lambda_{\sigma(1)})\psi(\lambda_{\sigma(2)},\ldots,\lambda_{\sigma(n+1)}),\\
(z(\eta) \psi)(\lambdav) &= \sqrt{n} \,\psi(\eta,\lambdav).
\end{aligned}
\end{equation}
They fulfill the Zamolodchikov relations
\begin{equation}
\begin{aligned}
\zd(\theta)\zd(\theta') &= S(\theta-\theta')\zd(\theta')\zd(\theta),\\
z(\eta)z(\eta') &= S(\eta-\eta')z(\eta')z(\eta),\\
z(\eta)\zd(\theta) &= S(\theta-\eta)\zd(\theta)z(\eta) + \delta(\theta-\eta) \idop.
\end{aligned}
\end{equation}
More precisely, the ``smeared'' operators $z^\#(f) = \int d\theta f(\theta) z^\#(\theta)$ are operator-valued distributions. For fixed $f$, the $\zd(f)$, $z(f)$ are unbounded operators on $\fpn$, but they fulfill the bounds
\begin{equation}\label{omegaz}
\begin{aligned}
 \gnorm{ e^{\omega(H/\mu)} \zd(f) e^{-\omega(H/\mu)} \fpnp_n }{}
  &\leq \sqrt{n+1} \onorm{f}{2},
\\
 \gnorm{ e^{\omega(H/\mu)} z(f) e^{-\omega(H/\mu)} \fpnp_n }{}
  &\leq \sqrt{n} \onorm{f}{2}.
\end{aligned}
\end{equation}
We showed in \cite[Prop.~2.1]{BostelmannCadamuro:expansion} that one can define multilinear extensions of ``normal ordered monomials'' of these annihilators and creators, $\zd(\theta_1)\ldots \zd(\theta_m)z(\eta_1)\ldots(\eta_n)$. Namely, if $f \in \dcal(\rbb^{m+n})'$ with $\onorm{f}{m \times n} < \infty$, then
\begin{equation}
  z^{\dagger m}z^n(f) =  \int d\thetav \, d\etav\, f(\thetav,\etav) \, \zd(\theta_1)\ldots \zd(\theta_m)z(\eta_1)\ldots z(\eta_n)
\end{equation}
is a well-defined quadratic form in $\qf^\omega$.
The importance of these expressions lies in the fact that \emph{any} quadratic form can be expanded into a series of such monomials $z^{\dagger m}z^n(f)$. We summarize this result, see \cite[Sec.~3]{BostelmannCadamuro:expansion} for details.

\begin{theorem}\label{theorem:expansion}
 For any $m,n\in\nbb_0$, let $f_{m,n} \in \dcal(\rbb^{m+n})'$ with $\onorm{f_{m,n}}{m \times n} < \infty$. Then, there is a unique $A \in \qf^\omega$ such that
\begin{equation}\label{eq:shortexp}
   A = \sum_{m,n} \frac{1}{m!n!} z^{\dagger m}z^n(f_{m,n}).
\end{equation}
Conversely, given $A \in \qf^\omega$, there are unique $f_{m,n} \in \dcal(\rbb^{m+n})'$  that are $S$-symmetric (in the first $m$ and last $n$ variables separately) and fulfill $\onorm{f_{m,n}}{m \times n} < \infty$ such that \eqref{eq:shortexp} holds.
\end{theorem}

The ``expansion coefficients'' $f_{m,n}$ depend linearly on $A$; we will denote them as $\cme{m,n}{A}$ in the following. They can be expressed as matrix elements of $A$ in a very explicit way, which we will now recall.

As in \cite{BostelmannCadamuro:expansion}, a \emph{contraction} $C$ is defined as a triple $C= (m,n,\{(l_1,r_1),\ldots,(l_{|C|},r_{|C|})\})$, where $m,n\in\nbb_0$, $1 \leq l_j \leq m$, $m+1 \leq r_j \leq m+n$. We denote $\ccal_{m,n}$ the set of contractions with fixed $m$ and $n$. Associated to a fixed contraction $C$, we consider the quantities

\begin{align}
\label{eq:lcvector}
   \lvector{C}{\thetav} &:= \zd(\theta_1) \cdots \widehat{\zd(\theta_{l_1})} \cdots \widehat{\zd(\theta_{l_{|C|}})} \cdots \zd(\theta_m) \Omega,
\\
   \rvector{C}{\etav} &:= \zd(\eta_n) \cdots \widehat{\zd(\eta_{r_1-m})} \cdots \widehat{\zd(\eta_{r_{|C|}-m})} \cdots \zd(\eta_1) \Omega,
\\
 \delta_{C}(\thetav,\etav) &:= \prod_{j=1}^{|C|}\delta(\theta_{l_j}-\eta_{r_j-m}), \label{eq:deltac}
\\
 S_C(\thetav,\etav) &:= \Big(\prod_{j=1}^{|C|}\prod_{p_{j}=l_{j}+1}^{r_{j}-1}S^{(m)}_{p_{j},l_{j}}\Big) \prod_{\substack{r_{i}<r_{j} \\ l_{i}<l_{j}}} S^{(m)}_{l_{j},r_{i}},\label{eq:sc}
\\
 R_C(\thetav,\etav) &:= \prod_{j=1}^{|C|} \Big(1-\prod_{p_j=1}^{m+n} S^{(m)}_{l_j,p_j}(\thetav,\etav) \Big),\label{eq:rc}
\end{align}
where the hat marks elements that are left out of the sequence, and where $S^{(m)}_{p,q}(\xiv)=S(\xi_p-\xi_q)$ if $p,q \leq m$ or $p,q>m$, and $S^{(m)}_{p,q}(\xiv)=S(\xi_q-\xi_p)$ otherwise.
With these notions, we have \cite[Eq.~(3.16)]{BostelmannCadamuro:expansion}
\begin{equation}\label{eq:fmndef}
\cme{m,n}{A}(  \thetav,  \etav) =
\sum_{C\in \ccal_{m,n}} (-1)^{|C|} \delta_C \, S_C (\thetav,\etav) \,
\hscalar{ \lvector{C}{\thetav} }{ A \,\rvector{C}{\etav} }.
\end{equation}
Using this explicit expression, one can prove directly that the $\cme{m,n}{A}$ are $S$-symmetric in $\thetav$ and in $\etav$, that $\onorm{\cme{m,n}{A}}{m \times n}<\infty$ if $A \in \qf^\omega$, and one can compute the behavior of $\cme{m,n}{A}$ if $A$ is subjected to translations, boosts and space-time reflections. Details can be found in \cite[Sec.~3]{BostelmannCadamuro:expansion}, and we will refer to there directly when we use those results in the present paper.

\subsection{Locality}\label{sec:locality}

We now come to the description of \emph{locality} in our models, i.e., to observables associated with certain regions of spacetime. The regions of interest are, first of all, wedges: the standard right wedge 
$\rightwedge$ as defined in \eqref{eq:rightwedge}; its causal complement, the standard left wedge $\leftwedge$; and their translates, $\wcal_x := \rightwedge + x$  and $\wcal_y' := \leftwedge + y = (\wcal_y)'$ ($x,y \in \rbb^2$). 
Further, we will consider the double cone $\mathcal{O}_{x,y}=\mathcal{W}_{x}\cap \mathcal{W}_{y}'$, with $x,y \in \mathbb{R}^{2}$, $y-x\in \mathcal{W}$.
In particular, we are interested in the double cone of radius $r>0$ centered at the origin, defined as
$\mathcal{O}_{r}=\mathcal{W}_{-r}\cap \mathcal{W}_{r}'$, where $\wcal_{r}:=\wcal_{(0,r)}$ etc.

As suggested by Schroer \cite{Schroer:1997-1}, we take our basic observables to be localized in wedge regions (say, in $\rightwedge$) rather than in bounded regions, since wedge-local observables are easier to describe. More precisely, we define a quantum field $\phi$ as
\begin{equation}
\phi(f):= \zd(f^{+}) + z(f^{-}), \quad f\in \mathcal{S}(\mathbb{R}^{2}).
\end{equation}
This is very similar to the free scalar field, and in fact, it reduces to the free field if $S=1$. In the general case, the field $\phi$ still retains a number of properties that are known from the free field \cite{Lechner:2003}: It is defined on $\fpn$, and essentially selfadjoint for real-valued $f$. Moreover, $\phi$ has the Reeh-Schlieder property, transforms covariantly under the representation $U(x,\lambda)$ of the proper orthochronous Poincar\'e group, and it solves the Klein-Gordon equation.

However, for generic $S$, the field $\phi(x)$ is not localized at the space-time point $x$ in the usual sense. Rather, we can understand $\phi(x)$ as localized in the infinitely extended wedge $\wcal_x'$ in the following way. Let us introduce the ``reflected'' Zamolodchikov operators, 
\begin{equation}
z(\theta)':=Jz(\theta)J, \quad \zd(\theta)':= J \zd(\theta)J,
\end{equation}
and define another field $\phi'$ as, $f\in \mathcal{S}(\mathbb{R}^{2})$,
\begin{equation}
\phi'(f):=J\phi(f^{j})J,\quad f^{j}(x):=\overline{f(-x)}.
\end{equation}
The two fields $\phi, \phi'$ are now relatively \emph{wedge-local} \cite[Prop.~2]{Lechner:2003}: For real-valued test functions $f,g$ with $\supp f \subset \mathcal{W}'$ and $\supp g \subset \mathcal{W}$, one finds that the closures $\phi(f)^-$ and $\phi'(g)^-$ spectrally commute. (This is obtained by computing the commutation relations of $z,z^\dagger$ with $z', z^{\dagger \prime}$; we will recall these later in Sec.~\ref{sec:fw-to-aw}.) Hence, we can understand $\phi(x)$ and $\phi'(y)$ as being localized in the shifted left wedge $\mathcal{W}'_{x}$ and in the shifted right wedge $\mathcal{W}_{y}$, respectively.

Instead of working with unbounded (closed) operators, it is often convenient to work with associated algebras of bounded operators. We consider the following von Neumann algebra:
\begin{equation}
\M = \{e^{i\phi(f)^-} \,|\, f \in \scal_\rbb(\rbb^2), \, \supp f \subset \wcal' \}'
= \{e^{i\phi(f)^-} \,|\, f \in \dcal_\rbb^\omega(\wcal') \}'.
\end{equation}
(The subscript $_\rbb$ means the restriction to real-valued functions, and the second equality follows due to continuity of the map $\scal(\rbb^2) \to \boundedops$, $f \mapsto \exp i \phi(f)^-$ in the strong operator topology, cf.~\cite[Prop.~5.2.4]{BraRob:qsm2}.)
As shown in \cite[Thm.~3.2]{Lechner:2008}, we can consistently interpret the algebra $\M$ as localized in the right wedge $\rightwedge$, so that the canonically defined algebras for other wedges,
\begin{equation}
  \A(\wcal_x) := U(x,0) \M U(x,0)\st, \quad  \A(\wcal_y') := J \A(\wcal_{-y}) J,
\end{equation}
fulfill causal commutation relations. We can then define the algebra for a double cone $\ocal_{x,y}= \wcal_x \cap \wcal_y'$ by intersection,
\begin{equation}\label{eq:aodef}
\A(\mathcal{O}_{x,y}):= \A(\wcal_x) \cap  \A(\wcal_y'),
\end{equation}
and for other bounded regions $\ocal$ by additivity. The resulting map $\ocal \mapsto \A(\ocal)$ then fulfills all standard axioms of the algebraic approach to quantum field theory \cite[Sec.~2]{Lechner:2008}.

It is not \emph{a priori} clear that the algebras $\A(\ocal)$ are nontrivial for bounded regions $\ocal$, i.e., that they contain any operator except for multiples of the identity. However, Lechner proved \cite{Lechner:2008} that at least for $\ocal$ of a certain minimum size, and under a mild regularity condition for the scattering function $S$, the vacuum vector $\Omega$ is indeed cyclic for $\A(\ocal)$, of which it follows that the algebras are type $\mathrm{III}_1$ factors \cite{BuchholzLechner:2004}. This is important in our context, since it is (to the best of our knowledge) the only existence result so far for compactly localized observables in the models at hand, as long as $S$ is not constant.

In this paper, we will discuss locality properties in terms of the decomposition in Thm.~\ref{theorem:expansion}, and therefore, on the level of quadratic forms. We need a notion of locality that is applicable to such quadratic forms $A \in \qf^\omega$. This will be defined relative to the wedge-local fields $\phi,\phi'$, that is, based on commutators of $A$ with these fields. 

To that end, we first need to clarify in which sense these commutators are defined.
If $f \in \dcal^\omega(\rbb^2)$, then $\zd(f^+)$ maps $\fpno$ into $\fpno$, so that $A\zd(f^+)$ is well-defined as a quadratic form; indeed, one finds from \eqref{eq:aomeganorm}, \eqref{omegaz} that $\onorm{A\zd(f^+)}{k} \leq \sqrt{k+1}\onorm{f^+}{2}\onorm{A}{k+1}$. 
Analogously, the product $\zd(f^+)A$, and products of $A$ with $z(f^-)$, $\phi(f)$, $\phi'(f)$ from the left or the right are well-defined within $\qf^\omega$, and hence we can also define the commutator $[A,\phi(f)]:=A\phi(f)-\phi(f)A \in \qf^\omega$. This enables us to introduce our notion of locality.

\begin{definition}\label{definition:omegalocal}
  Let $A \in \qf^\omega$. We say that $A$ is \emph{$\omega$-local} in $\wcal_x$  if
\begin{equation}\label{eq:acommute}
     [A,\phi(f)] = 0
\quad
\text{for all }
    f \in \dcal^\omega(\wcal_x'), \text{ as a relation in $\qf^\omega$}.
\end{equation}
$A$ is called $\omega$-local in $\wcal_x'$ if $J A^\ast J$ is $\omega$-local in $\wcal_{-x}$. $A$ is called $\omega$-local in the double cone $\ocal_{x,y} = \wcal_x \cap \wcal'_y$ if it is $\omega$-local in both $\wcal_x$ and $\wcal'_y$.
\end{definition}

$\omega$-locality (say, in $\rightwedge$) is a priori a weaker notion than locality in the usual sense. For example, it does not tell us whether $A$ commutes with unitary operators $\exp i\phi(f)^-$ with $\supp f \subset \leftwedge$, or with a general element $B \in \M'$. In fact, we would not be able to write down such commutators in a meaningful way if $A$ is just a quadratic form. For bounded operators, $\omega$-locality reduces to the usual notion: If $A \in \boundedops$, then $A$ is $\omega$-local in $\wcal$ if and only if $A \in \M$, and it is $\omega$-local in $\ocal_{x,y}$ if and only if $A \in \A(\ocal_{x,y})$ as defined in \eqref{eq:aodef}. (This is easy to see by using closability of $\phi(f),\phi'(f)$ and density arguments.) In view of applications, it would be favorable to have a similar statement for unbounded, but closable operators $A$. This can in fact be achieved \cite[Prop.~4.4]{Cadamuro:2012}, and will be presented elsewhere \cite{BostelmannCadamuro:examples-wip}.
For our present purposes, we give the following characterization of $\omega$-locality.

\begin{lemma}\label{lemma:localitychar}
 Let $\omega$ be an indicatrix, and $A \in \qf^\omega$. The following conditions are equivalent:
\begin{enumerate}
\renewcommand{\theenumi}{(\roman{enumi})}
\renewcommand{\labelenumi}{\theenumi}

 \item \label{it:charlocal}
    $A$ is $\omega$-local in $\wcal$.
 \item \label{it:charcommutator}
    $[A,\phi(f)]= 0$ for all $f \in \dcal^\omega(\wcal')$, as a relation in $\qf^\omega$.
 \item \label{it:charomegavar}
    For every $\psi,\chi\in\fpno$, there exists an indicatrix $\omega'$ such that
   $\hscalar{\psi}{[A,\phi(f)]\chi} = 0$ for all $f \in \dcal^{\omega'}(\wcal')$.
 \item \label{it:chartempered}
    For every $\psi,\chi\in\fpno$, it holds that
    $\hscalar{\phi(x)\psi}{A \chi} = \hscalar{\psi}{A \phi(x) \chi}$ for $x \in\wcal'$,
    in the sense of tempered distributions.
\end{enumerate}
\end{lemma}

\begin{proof} 
We first note that in \ref{it:chartempered}, the expression $\hscalar{\psi}{A \phi(x) \chi}$ can indeed be understood as a tempered distribution (and similar arguments then apply to $\hscalar{\phi(x)\psi}{A \chi}$). Namely, since $\gnorm{A}{k}^\omega < \infty$, the matrix element $\hscalar{\psi}{A \chi}$ is well defined (by continuous extension) if $\psi,\chi \in \fpn$ and at least \emph{one} of $\psi,\chi$ is in $\Hil^{\omega}$. 
Noting that $\phi(f) \fpn \subset \fpn$, and that the map $\scal(\rbb^2) \to \Hil$, $f \mapsto \phi(f)\chi$ is continuous, it then follows that $f \mapsto \hscalar{\psi}{A \phi(f) \chi}$ is continuous in the Schwartz topology.---%
Now \ref{it:charlocal}$\Leftrightarrow$\ref{it:charcommutator} is true by definition; \ref{it:charcommutator}$\Rightarrow$\ref{it:charomegavar} is trivial with $\omega'=\omega$; \ref{it:charomegavar}$\Rightarrow$\ref{it:chartempered} follows due to density of $\dcal^{\omega'}(\wcal')$ in $\dcal(\wcal')$; and \ref{it:chartempered}$\Rightarrow$\ref{it:charcommutator} holds since $\dcal^\omega(\wcal')\subset\scal(\rbb^2)$.
\cmpqed\end{proof}

%% file: tools.tex
\section{Meromorphic functions on tube domains}\label{sec:tools}

As may be evident from the introduction, our analysis will rely in large parts on the theory of analytic and meromorphic functions in several variables. Specifically, most of these functions will be defined on a tube domain $\tube(\bcal) = \rbb^k + i \bcal$ with some set $\bcal \subset \rbb^k$. We will now establish some tools which are helpful in this context. First, this concerns the case where $\bcal$ is a priori not an open set, but a collection of certain lines (or a graph). In Sec.~\ref{sec:graphs}, we discuss how analytic functions on such domains can be understood, and how well-known results -- Bochner's tube theorem and the maximum modulus principle -- extend to this situation. Second, we need to understand the structure of first-order poles and the residues of meromorphic functions on tube domains, and control the directional dependence of the boundary distributions near such poles; this will be done in Sec.~\ref{sec:bvlemma}.




%% file: graphs.tex
\subsection{CR functions on graphs} \label{sec:graphs}

By a \emph{graph} $\gcal$ in $\rbb^k$, we mean a collection of points in $\rbb^k$ (the nodes), together with a set of straight lines connecting some of these nodes (the edges). For our purposes, the nodes will always lie on the grid $\pi \zbb^k$, and the edges will always be axis-parallel lines between next neighbors; that is, the lines have the parametrized form $\lambdav(s) = \boldsymbol{\nu} + s \ev^{(j)}$, where $\ev^{(j)}$ is a standard basis vector of $\rbb^k$, where $0 < s < \pi$, and $\boldsymbol{\nu}$ as well as $\boldsymbol{\nu} + \pi \ev^{(j)}$ are nodes of $\gcal$.
The \emph{tube over $\gcal$}, denoted $\tube(\gcal)$, is the set of all $\zetav = \thetav + i \lambdav$ with $\thetav \in \rbb^k$ and $\lambdav$ on an edge of $\gcal$.

A \emph{CR function $F$ on $\tube(\gcal)$} is a smooth function on $\tube(\gcal)$ which is analytic along the edges; that is, with an edge $\lambdav(s)$ parametrized as above, $F$ is analytic in $\zeta_j$ in the specified domain, while being smooth in all (real) variables.
We moreover demand that the boundary values of $F$ and of all its derivatives exist at the nodes, $s \searrow 0$ and $s \nearrow \pi$,  and that where several edges end in a common node, the different limits of $F$ agree.  We can therefore just refer to \emph{the} boundary value at a node, without indicating the direction of the limit. We will however sometimes write $F(\cdotarg + i \nuv)\vert_\gcal$ for the boundary value at node $\nuv$ obtained within $\gcal$, in case that several graphs play a role.

A \emph{CR distribution $F$ on $\tube(\gcal)$} is, correspondingly, a function analytic along the edges while being a $\dcal(\rbb^{k-1})'$ distribution in the remaining real variables.\footnote{%
See~\cite[Ch.~I, Appendix 2, \S{}3]{GelfandShilov:1964vol1} for a discussion of distributions depending analytically on a parameter.} Similar to the above, we demand that all boundary values at nodes exist in the sense of distributions, and agree where several edges meet in a common node.

We will derive some general properties of CR functions on $\tube(\gcal)$, which are mostly extensions of standard results adapted to our framework. 

First of all, we remark that CR distributions can be ``regularized'' by convolution with test functions: Let $F$ be a CR distribution on $\gcal$, and let $g = (g_1,\ldots,g_k) \in \dcal(\rbb)^k$. We define
\begin{equation}\label{eq:fgconvolute}
     (F \ast g)(\zetav) := \int F(\zetav - \xiv) g_1(\xi_1) \ldots g_k(\xi_k) \, d^k\xiv.
\end{equation}
Then $F \ast g$ is a CR \emph{function} on $\tube(\gcal)$, as follows from continuity in $g_j$ in the $\dcal(\rbb)$ topology. 

Further, CR distributions obey a version of the tube theorem: they can be extended analytically to the convex hull of the graph. 
To formulate that, let us denote with $\bar \gcal \subset \rbb^k$ the closure of the edges of $\gcal$ (or equivalently, the edges together with the nodes, as a subset of $\rbb^k$). 
Following \cite{Kazlow:1979}, we define 
\begin{equation}
\begin{aligned}
 \ich \gcal &:= \operatorname{conv}(\bar \gcal)^\circ, & \text{the \emph{interior} (of the convex hull) of $\gcal$},
\\
 \ach \gcal &:= (\ich\gcal) \cup \bar\gcal, \quad &\text{the \emph{almost convex hull} of $\gcal$.}
\end{aligned}
\end{equation}

\begin{lemma} \label{lem:graphtube}
  Let $\gcal$ be a connected graph and $F$ a CR distribution on $\tube(\gcal)$. Then, $F$ extends to an analytic function on $\ich{\gcal}$ with distributional boundary values on $\ach \gcal$.
\end{lemma}

\begin{proof}
 We apply results from \cite{Kazlow:1979}. In terms used there, $\bar \gcal$ is a connected, locally closed, locally starlike set. For any $g \in \dcal(\rbb)^k$, the convolution $F \ast g$ is a CR function on $\tube(\gcal)$. Indeed, at nodes where two edges along the same axis meet, $F \ast g$ continues analytically in the respective variable across the node (by Morera's theorem). Thus, at any node, $F \ast g$ is a smooth function (possibly with single-sided derivatives) defined on lines in at most $k$ independent directions in $\im\zetav$; it is then easy to see that $F \ast g$ is smooth in the sense of Whitney. Hence $F \ast g$ is a CR' function on $\tube(\bar \gcal)$ in the sense of \cite[Def.~2.12]{Kazlow:1979}. Applying \cite[Theorem~6.1]{Kazlow:1979} yields an extension $G$ of $F \ast g$ to $\tube(\ach \gcal)$, analytic in $\tube(\ich\gcal)$.

 It remains to show that $G = F \ast g$ with some function $F$ analytic in $\tube(\ich\gcal)$. We sketch this argument briefly.\footnote{See, e.g., \cite[p.~530]{Eps:edge_of_wedge} for a more detailed exposition in a similar situation.} One first observes that at each fixed $\lambdav\in\ich\gcal$, the map $g \mapsto G(i\lambdav)$ is continuous in the $\dcal$-topology (cf.~the remark in \cite[Sec.~12]{Kazlow:1979}). Smearing also in $\lambdav$ within $\ich\gcal$, we obtain a distribution in $2k$ variables which -- due to analyticity of $G$ -- fulfills the Cauchy-Riemann equations in the weak sense. But that implies that this distribution has an analytic kernel \cite[p.~72]{Schwartz:1959b}, which yields the desired function $F$.
\cmpqed\end{proof}

As a next point, due to \cite[Sec.~11, Corollary]{Kazlow:1979}, the maximum modulus principle holds for CR functions on $\tube(\gcal)$ if $\gcal$ is connected:\footnote{%
Given some mild conditions on the growth of $F$, one can use a Phragm\'en-Lindel\"of argument to show that the supremum in \eqref{eq:maxmodulusedge} can actually be restricted to the tube over the \emph{nodes} of $\gcal$ \cite[Eq.~(C.6)]{Cadamuro:2012}.}
\begin{equation}\label{eq:maxmodulusedge}
   \sup_{\zetav \in \tube (\ach \gcal) } |F(\zetav)|
 =   \sup_{\zetav \in \tube(\bar \gcal) } |F(\zetav)|.
\end{equation}
We want to obtain a similar maximum modulus principle for $L^2$-like bounds of a CR \emph{distribution} $F$, more precisely for the norm $\gnorm{\cdotarg}{\times}$ as defined in Eq.~\eqref{eq:fullcrossnorm}. 
This can be achieved with standard techniques.

\begin{lemma}\label{lemma:maxmodcross}
Let $\gcal$ be a connected graph and $F$ a CR distribution on $\tube(\gcal)$. For the extension of $F$ to $\tube(\ach \gcal)$, it holds that
\begin{equation}\label{eq:maxmodcrossedge}
    \sup_{\lambdav \in \ach \gcal } \gnorm{ F(\cdotarg + i \lambdav ) }{\times}
 =  \sup_{\lambdav \in \bar\gcal }\gnorm{ F(\cdotarg + i \lambdav ) }{\times}.
\end{equation}
\end{lemma}

\begin{proof}
 For $g=(g_1, \ldots, g_k)\in\dcal(\rbb)^k$ with $\gnorm{g_j}{2}\leq 1$, we define $F \ast g(\zetav)$ as in \eqref{eq:fgconvolute}; this function is analytic on $\tube(\ich \gcal )$ and a CR function on $\tube(\gcal)$. We note
\begin{equation}\label{eq:fgsup}
   \gnorm{F(\cdotarg + i \lambdav)}{\times} = \sup_g |F \ast g( i \lambdav)|
   = \sup_{g,\thetav} |F \ast g( \thetav + i \lambdav)|.
\end{equation}
 Applying the maximum modulus principle \eqref{eq:maxmodulusedge} to $F \ast g$ then immediately yields \eqref{eq:maxmodcrossedge}.
\cmpqed\end{proof}

We further prove a result on pointwise bounds on analytic functions, estimated by the supremum of their norm $\gnorm{\cdotarg}{\times}$. This is not restricted to CR functions on graphs, but is useful in conjunction with Lemma~\ref{lemma:maxmodcross}. It follows by use of the mean value property; see e.g.~\cite[Prop.~4.4]{Lechner:2008} for a  similar application of this technique.

\begin{proposition}\label{proposition:pointwise}
Let $\ical\subset \rbb^k$ be open and $F$ analytic on $\tube(\ical)$. Then, for all $\zetav\in \tube(\ical)$,
\begin{equation}
|F(\zetav)|\leq \frac{(4/\pi)^{k} \; k^{k/4}}{\operatorname{dist}(\im\zetav,\partial \ical)^{k/2}}\sup_{\lambdav \in \ical} \gnorm{F(\cdotarg +i\lambdav)}{\times} .
\end{equation}
\end{proposition}
\begin{proof}
\sloppy
For fixed $\zetav\in\tube(\ical)$, let $D_{t}\subset \mathbb{C}$ be the disc around the origin with radius $t:=\frac{1}{2}k^{-1/2}\operatorname{dist}(\im\zetav,\partial \mathcal{I})$. The polydisc $(D_{t}\times\cdots \times D_{t})+\zetav$ is then contained in $\tube(\ical)$. The mean value property for analytic functions yields
\begin{equation}\label{eq:meanval}
\begin{aligned}
F(\zetav) &=
(\pi t^{2})^{-k}\int_{D_{t}}d\theta_{1}d\lambda_{1}\ldots \int_{D_{t}}d\theta_{k}d\lambda_{k}\;F(\zetav+\thetav+i\lambdav)\\
&= (\pi t^{2})^{-k}\int_{[-t,t]^{\times k}}d\lambdav\; \big( F \ast \chi_{\lambdav} \big) (\zetav+i\lambdav),
\end{aligned}
\end{equation}
where $\chi_{\lambdav}(\thetav) = \prod_{j=1}^k \chi_j(\theta_j)$, and where $\chi_j$ is the characteristic function of the interval $[-(t^2\!-\!\lambda_j^2)^{1/2}, +(t^2\!-\!\lambda_j^2)^{1/2}]$.
Since $\zetav+i\lambdav \in \tube(\ical)$ by construction, we can estimate
\begin{equation}
| (F \ast \chi_{\lambdav})(\zetav+i\lambdav) |
\leq \sup_{\lambdav' \in \ical} \gnorm{F(\cdotarg +i\lambdav')}{\times} \cdot \prod_{j=1}^{k} \gnorm{\chi_j}{2}.
\end{equation}
(Cf.~Eq.~\eqref{eq:fgsup}; that relation can be continued to $L^2$ functions $g_j$ by continuity.)
Taking into account $\gnorm{\chi_j}{2} \leq \sqrt{2t}$, we find from \eqref{eq:meanval},
\begin{equation}
|F(\zetav)| \leq (\pi t^{2})^{-k}  (2t)^{k} (2t)^{k/2} \sup_{\lambdav' \in \ical} \gnorm{F(\cdotarg +i\lambdav')}{\times} ,
\end{equation}
which implies the desired result after inserting the definition of $t$.
\cmpqed\end{proof}

%% file: boundaryvalues.tex
\subsection{Residues and boundary distributions in several variables} \label{sec:bvlemma}

For our analysis, we make use of meromorphic functions in several variables and of their residues. Let us fix the corresponding notations and conventions. In this paper, all poles of meromorphic functions will be located on hyperplanes, $\zv \cdot \av = c$ with $\av \in\rbb^k$, $c \in \cbb$. Specifically, if $\av^{(1)},\ldots,\av^{(p)}\in\rbb^k$ are pairwise linear independent and $c_1,\ldots,c_p\in\cbb$, then we say that $F$ has \emph{first-order poles} at the hyperplanes $\zv \cdot \av^{(j)} = c_j$ if
\begin{equation}\label{eq:firstorderpole}
   F(\zv) \prod_{j=1}^p (\zv \cdot \av^{(j)} - c_j)
\end{equation}
is analytic in a neighborhood of these hyperplanes. Regarding residues of $F$, we choose our notational convention as follows: 
\begin{equation} \label{eq:resnotation}
   \res_{\zv \cdot \av^{(j)} = c_j} F =  \big((\zv \cdot \av^{(j)} - c_j) F(\zv)\big) \Big\vert_{\zv \cdot \av^{(j)} = c_j}
\end{equation}
The residues of $F$ on one of the hyperplanes are then again meromorphic on a lower-dimensional complex manifold, which we can identify with $\cbb^{k-1}$.
Note that the notation \eqref{eq:resnotation}, while convenient for us, needs to be taken with some care: For $\alpha \in \rbb \backslash\{0\}$, we have
\begin{equation} 
   \res_{\zv \cdot (\alpha\av) = \alpha c} F = \alpha \res_{\zv \cdot \av = c} F,
\end{equation}
although $\zv \cdot \av = c$ and $\zv \cdot (\alpha \av) = \alpha c$ describe the same geometric set. (We accept this slightly unintuitive feature for simplicity; the alternative would be to work with oriented manifolds, and with differential forms rather than functions.)

In this section, we will investigate the boundary values of meromorphic functions (in the sense of distributions). Suppose that $F$ is meromorphic around the real hyperplane $\rbb^k$ and fulfills bounds of the type $|F(\xv+ i\epsilon \bv)| \leq c \epsilon^{-\ell}$ with some $\ell>0$,  locally uniformly in $\xv\in\rbb^k$ and in $\bv$ within some open convex cone. (This is evidently fulfilled in the situation \eqref{eq:firstorderpole}.) Then the boundary value
\begin{equation}
  F(\xv + i 0 \bv) := \lim_{\epsilon \searrow 0} F(\xv+ i \epsilon \bv)
\end{equation}
exists in $\dcal(\rbb^k)'$ and is independent of $\bv$ inside the chosen cone \cite[Sec.~125~Z]{Ito:1993vol1}. As shorthand, we will often write the boundary distribution just as $F(\xv+ i \zerov)$ without specifying the cone, just noting that the boundary is approached from a certain part of the analyticity region. The notation $F(\xv+ i \yv + i \zerov)$ with fixed $\yv\in\rbb^k$ should be understood accordingly by translating the argument.

Evidently, if $\bv$ is varied \emph{across} connected components of the analyticity region -- e.g., if $\bv$ is taken across a pole hyperplane -- then the boundary value can change with $\bv$. We aim to describe this in more detail.
To that end, recall that if $F$ is a function of one complex variable, analytic near the real axis except for a first-order pole at $z=0$, then we have the relation between boundary distributions,
\begin{equation}
   F(x-i0) = F(x+i0) + 2 \pi i \delta(x) \res_{z=0} F(z).
\end{equation}
A first multi-dimensional generalization is formulated as follows.

\begin{lemma}\label{lemma:onepole}
  Let $\ucal \subset \rbb^k$ be a neighborhood of zero, $\ccal \subset \rbb^k$ an open convex cone, and $\av \in \rbb^k$.
  Let $F$ be meromorphic on $\tube(\ucal)$ and $(\zv \cdot \av)F(\zv)$ analytic on $\tube(\ccal \cap \ucal)$. 
  Let $\bv^+,\bv^-,\bv^\bot \in \ccal$ so that $\pm \av \cdot \bv^\pm  > 0$, $\av\cdot\bv^\bot  = 0$. Then it holds that
  \begin{equation}\label{eq:onepole}
     F(\xv + i 0 \bv^-)  = F(\xv + i 0 \bv^+) + 2 \pi i \delta(\xv \cdot \av) \res_{\zv \cdot \av = 0} F(\xv + i 0 \bv^\bot).
  \end{equation}
\end{lemma}

\begin{proof}
We assume without loss of generality that $\ucal$ is a ball around the origin, that $\bv^+,\bv^-,\bv^\bot\in\ucal$, and that $\av=\ev^{(1)}$.
We prove the distributional equation \eqref{eq:onepole} when evaluated on test functions $g \in \dcal(\kcal)$, where $\kcal$ is a fixed convex compact set. Let $G(\zv):=(\zv\cdot\av)F(\zv)$; by hypothesis, this function is analytic on $\tube(\ccal \cap \ucal)$, and we have $G(\zv)=\res_{\zv\cdot\av=0}F(\zv)$ if $\zv\cdot\av=0$. 

We will first prove \eqref{eq:onepole} under the additional assumption that $G$ and its gradient, $\nabla G$, have continuous extensions to $\kcal + i (\bar\ccal \cap \bar\ucal)$. We compute,
\begin{multline}\label{eq:fdiff}
 \int \Big( F(\xv+i0 \bv^-)-F(\xv+i0 \bv^+) \Big) g(\xv)d\xv
 \\ = \lim_{\epsilon\searrow 0}
 \int \Big( \frac{G(\xv+i \epsilon \bv^-)}{x_1 + i \epsilon b_1^-} - \frac{G(\xv + i \epsilon \bv^+)}{x_1 + i \epsilon b_1^+} \Big) g(\xv)d\xv.
\end{multline}
With the notation $\xv=(x_1,\hat\xv)$ and the substitution $y = x_1/\epsilon$, we can rewrite this as
\begin{equation}\label{eq:gdiff}
 \eqref{eq:fdiff} 
 = \lim_{\epsilon\searrow 0}
 \int dy\,d\hat\xv\,  \frac{i b_1^+ G(\zv_\epsilon^-) + i b_1^- G(\zv_\epsilon^+) 
+ y \big(G(\zv_\epsilon^-) - G(\zv_\epsilon^+) \big)}{(y+ib_1^-)(y+ib_1^+)}
g(\epsilon y,\hat\xv),
\end{equation}
where $\zv_\epsilon^\pm := (\epsilon y + i \epsilon b_1^\pm , \hat\xv + i \epsilon \hat{\bv} {\vphantom{\bv}}^\pm)$. In the numerator of \eqref{eq:gdiff}, we have due to the support properties of $g$ that $|y|<c/\epsilon$ with some $c>0$, and for $\epsilon\leq 1$,
\begin{equation}
   |G(\zv_\epsilon^-) - G(\zv_\epsilon^+) |  \leq \epsilon \gnorm{\bv^+-\bv^-}{} \;\sup \big\{ \gnorm{\nabla G(\zv)}{} : \zv \in \kcal + i (\bar\ccal \cap\bar\ucal ) \big\}.
\end{equation}
By our continuity assumption, the supremum is finite and $|G(\zv_\epsilon^\pm)|$ is bounded as well; thus we obtain an integrable majorant in \eqref{eq:gdiff}. Applying the dominated convergence theorem, we arrive at
\begin{equation}
 \eqref{eq:fdiff} 
 =  \int dy \, \frac{i b_1^+ + i b_1^-}{(y+ib_1^-)(y+ib_1^+)}
\int d\hat\xv\, G(0,\hat\xv) g(0,\hat\xv).
\end{equation}
Using $\pm b_1^\pm > 0$, the integral in $y$ can be solved to give $2 \pi i$, which proves \eqref{eq:onepole} under our additional continuity hypothesis.

In the general case, after passing to slightly smaller $\ccal$ and $\ucal$, we can achieve that $G$ and $\nabla G$ are continuous on $\kcal + i (\bar\ccal \cap \bar\ucal)$ except possibly at $\im \zv = 0$. However, since $G$ is meromorphic in $\tube(\ucal)$, it is locally given as a quotient of two analytic functions, where the denominator has zeros of finite order; therefore, we can find $c>0$, $\ell>0$ such that
\begin{equation}\label{eq:gpoly}
   |G(\zv)| + \gnorm{\nabla G(\zv)}{} \leq c \gnorm{\im \zv}{}^{-\ell} \quad \text{for all }
  \zv \in \kcal + i (\bar\ccal \cap \bar\ucal), \; \im \zv \neq 0.
\end{equation}
Now let $\partial_\bot = \bv^\bot\cdot\nabla$ be the partial derivative in direction of $\bv^\bot$, and $G^{(-m)}$ an $m$th-order antiderivative of $G$ with respect to that direction. (The antiderivative is constructed by repeated integration; note that convexity of $\ccal$ enters here.) Due to \eqref{eq:gpoly}, we know that for sufficiently large $m$, both $G^{(-m)}$ and $\nabla G^{(-m)}$ are continuous on $\kcal + i (\bar\ccal \cap \bar\ucal)$, including the points where $\im \zv = 0$. (Cf.~\cite[Thm.~IX.16]{ReedSimon:1975-2} for this technique.) Since $\bv^\bot \cdot\av=0$, we have
\begin{equation}
   \int F(\xv + i \epsilon \bv^\pm) g(\xv) d\xv =
   (-1)^m \int \frac{G^{(-m)}(\xv + i \epsilon \bv^\pm)}{x_1 + i \epsilon b^\pm_1 }  \partial_\bot^m g (\xv) \, d\xv.
\end{equation}
We can now apply our previous analysis to $G^{(-m)}$, $ \partial_\bot^m g$ in place of $G,g$, yielding
\begin{equation}
 \int \Big( F(\xv+i0 \bv^-)-F(\xv+i0 \bv^+) \Big) g(\xv)d\xv = (-1)^m 2 \pi i \int G^{(-m)}(0,\hat\xv) \partial_\bot^m g(0,\hat\xv)\, d\hat\xv. 
\end{equation}
Observing that $G^{(-m)}(0,\hat\xv) = \lim_{\epsilon \to 0} G^{(-m)}(0,\hat\xv + i \epsilon \hat\bv\vphantom{\bv}^\bot)$, exchanging the limit with the integration sign, and then integrating by parts gives the result \eqref{eq:onepole}.
\cmpqed\end{proof}

Using the above lemma, we can derive a similar formula for a function that has first-order poles at several distinct hyperplanes.

\begin{proposition}\label{proposition:multivarres}
  Let $\ucal \subset \rbb^k$ be a neighborhood of zero, $\ccal \subset \rbb^k$ an open convex cone, and $\av\vn{1},\ldots,\av\vn{p} \in \rbb^k$ pairwise linear independent.
  Let $F$ be meromorphic on $\tube(\ucal)$ and $(\zv \cdot \av\vn{1})\cdots (\zv \cdot \av\vn{p})F(\zv)$ analytic on $\tube(\ccal \cap \ucal)$. 
  For any $M\subset \{1,\ldots,p\}$, let $\bv^M\in\ccal$ such that $\av\vn{j}\cdot\bv^M=0$ if $j\in M$, $\av\vn{j}\cdot\bv^M>0$ if $j \not\in M$.
  Let $\cv \in \ccal$ such that $\av\vn{j}\cdot\cv <0$ for all $j$. Then it holds that
\begin{equation}\label{sevdimres}
F(\xv+i0\, \cv)=\;\sum_{\mathclap{M\subset\left\{ 1,\ldots,p\right\}}}\;(2i\pi)^{|M|}\Big( \prod_{m\in M}\delta(\xv\cdot \av\vn{m})\Big) \res_{\zv\cdot \av\vn{m_{1}}=0} \!\! \ldots \!\! \res_{\zv\cdot \av\vn{m_{|M|}}=0}F(\xv+i0\,\bv^{M})
\end{equation}
with the notation $M=\{m_1,\ldots,m_{|M|}\}$.
\end{proposition}

\begin{proof}

We use induction on $p$. For $p=1$, the claim follows directly from Lemma~\ref{lemma:onepole} with $\bv^+ = \bv^{\emptyset}$, $\bv^-=\cv$, $\bv^\bot=\bv^{\{1\}}$.

Suppose now that the statement holds for $p-1$ in place of $p$. After possibly renumbering  the vectors $\av\vn{j}$, we can choose $\cv'\in\ccal$ such that $\av\vn{1}\cdot \cv' >0$, but $\av\vn{j}\cdot \cv' <0$ for $j \geq 2$. Within the tube over the cone $\ccal^- := \{ \yv \in \ccal : \yv\cdot\av\vn{j} < 0 \text{ for } j \geq 2 \}$, the function $(\zv\cdot\av\vn{1}) F(\zv) $ is analytic. Applying Lemma~\ref{lemma:onepole} with $\bv^+ = \cv'$, $\bv^-=\cv$, yields
\begin{equation}\label{eq:ccp}
 F(\xv + i  0 \cv) = F(\xv + i 0 \cv') + 2 \pi i \delta(\xv\cdot\av\vn{1}) \res_{\zv \cdot \av\vn{1}=0} F(\xv + i 0 \cv'')
\end{equation}
where $\cv''\in\ccal^-$ is chosen such that $\av\vn{1}\cdot \cv''=0$. To the term $F(\xv + i 0 \cv')$ we can apply the induction hypothesis with respect to the cone $\ccal^+:=\{ \yv \in \ccal : \yv\cdot\av\vn{1} > 0 \}$, noting that $(\zv \cdot \av\vn{2})\cdots (\zv \cdot \av\vn{p})F(\zv)$ is analytic on $\tube(\ccal^+ \cap \ucal)$. This yields
\begin{equation}\label{eq:fcpinduction}
F(\xv+i0 \cv')=\;\sum_{\mathclap{M\subset\left\{ 2,\ldots,p\right\}}}(2i\pi)^{|M|}\Big( \prod_{m\in M}\delta(\xv\cdot \av\vn{m})\Big) \res_{\zv\cdot \av\vn{m_{1}}=0}\!\!\ldots\!\! \res_{\zv\cdot \av\vn{m_{|M|}}=0}F(\xv+i0\,\bv^{M}).
\end{equation}
Further, the residue of $F$ in \eqref{eq:ccp} is a meromorphic function on the hyperplane $\zv\cdot\av\vn{1}=0$, which we can identify with $\cbb^{k-1}$; the function is analytic when multiplied with $(\zv\cdot\av\vn{2})\cdots(\zv\cdot\av\vn{p})$. Applying the induction hypothesis with respect to the cone $\ccal^0:=\{ \xv \in \ccal : \xv\cdot\av\vn{1} = 0 \}$ yields
\begin{equation}\label{eq:fresinduction}
\begin{aligned}
\res_{\zv\cdot\av\vn{1}=0} F(\xv+i0 \cv'')
=&\sum_{M\subset\left\{ 2,\ldots,p\right\}}(2i\pi)^{|M|}\Big( \prod_{m\in M}\delta(\xv\cdot \av\vn{m})\Big) 
\\
&\times \res_{\zv\cdot \av\vn{m_{1}}=0}\ldots \res_{\zv\cdot \av\vn{m_{|M|}}=0} \;\res_{\zv\cdot \av\vn{1}=0}F(\xv+i0\,\bv^{M\cup\{1\}}).
\end{aligned}
\end{equation}
Inserting \eqref{eq:fcpinduction} and \eqref{eq:fresinduction} into \eqref{eq:ccp}, and relabeling the summation index $M$ in \eqref{eq:fresinduction} as $M \cup \{1\}$, we arrive at the proposed result.
\cmpqed\end{proof}

%% file: wedges.tex
\section{Locality in the left wedge}\label{sec:wedges}

As announced in the introduction, we will now proceed to characterize the locality of observables in terms of analyticity properties of their expansion coefficients $\cme{m,n}{\cdotarg}$. In this section, we will consider observables localized in the left wedge $\wcal_r'$ with tip at $(0,r)$ on the time-0 axis. 

This characterization is formulated as the equivalence of three conditions. We formalize locality of a quadratic form $A$ in $\wcal_r'$ as a condition (AW). Given such $A$, we will be able to extend its expansion coefficients $\cme{m,n}{A}$ to CR distributions $T_{m+n}$ on a certain graph, fulfilling a set of conditions (TW). These distributions $T_{m+n}$ can further be extended to analytic functions $F_{m+n}$ on the convex hull of the graph, fulfilling conditions (FW). From functions $F_k$ fulfilling (FW), we can in turn construct a quadratic form $A$ which fulfills (AW). All three conditions will depend on the parameter $r \in \rbb$ and on the choice of an analytic indicatrix $\omega$; we do not explicitly denote this dependence.

Let us now define the three conditions in detail. This is simple to do for condition (AW) on the level of quadratic forms, since we already introduced a corresponding notion of locality in Def.~\ref{definition:omegalocal}.

\begin{definition}\label{def:conditionAW}
$A \in \qf^\omega$ fulfills condition (AW) if it is $\omega$-local in $\wcal_r'$.
\end{definition}


The next locality condition is formulated in the language of CR distributions on graphs in $\rbb^k$, as introduced in Sec.~\ref{sec:graphs}. 
The graph in question, which we denote $\gcal^k_+$, is given as follows: Its nodes are $\lambdav^{(k,j)} = (0,\ldots,0,\pi\ldots,\pi)\in \rbb^k$, with $j$ entries of $\pi$, and $0 \leq j \leq k$. Its edges are those between $\lambdav^{(k,j)}$ and $\lambdav^{(k,j+1)}$. Thus $\gcal_+^k$ has $k+1$ nodes and $k$ edges; see Fig.~\ref{fig:gkplus} for the cases $k=2$ and $k=3$. The locality condition reads as follows.

\begin{definition} \label{def:conditionTW}
A collection $T=(T_{k})_{k=0}^\infty$ of distributions on $\tube(\gcal^k_+)$ fulfills condition (TW) if the following holds for any fixed $k$, and with $\thetav \in \rbb^k$ arbitrary:
\begin{enumerate}
\renewcommand{\theenumi}{(TW\arabic{enumi})}
\renewcommand{\labelenumi}{\theenumi}

\item \label{it:twmero} \emph{Analyticity:} $T_k$ is a CR distribution on $\tube(\gcal^k_+)$.

\item \label{it:twsymm} \emph{$S$-symmetry}:
For any $\sigma \in \perms{k}$, we have
$\displaystyle{
T_k(\thetav)= S^\sigma(\thetav) T_k(\thetav^\sigma) .
}$

\item \label{it:twboundsreal}
\emph{Bounds at nodes:}
For any $j \in \{0,\ldots,k\}$,
\begin{equation*}
\onorm{ T_k( \cdotarg + i \lambdav^{(k,j)} )}{(k-j) \times j} < \infty.
\end{equation*}

\item \label{it:twboundsimag}
\emph{Bounds at edges:}
There exists $c>0$ such that for any $\lambdav\in \bar \gcal_+^k$,
\begin{equation*}
\gnorm{ e^{i \mu r \sum_j \sinh \zeta_j}e^{-\sum_{j}\oa( \sinh \zeta_{j})}
T_k( \pmb{\zeta})\big\vert_{\zetav=\cdotarg +i\lambdav} }{\times} \leq c.
\end{equation*}

\end{enumerate}

\end{definition}

\begin{figure}
        \centering
        \begin{subfigure}[b]{0.3\textwidth}
                \centering
                \includegraphics[width=\textwidth]{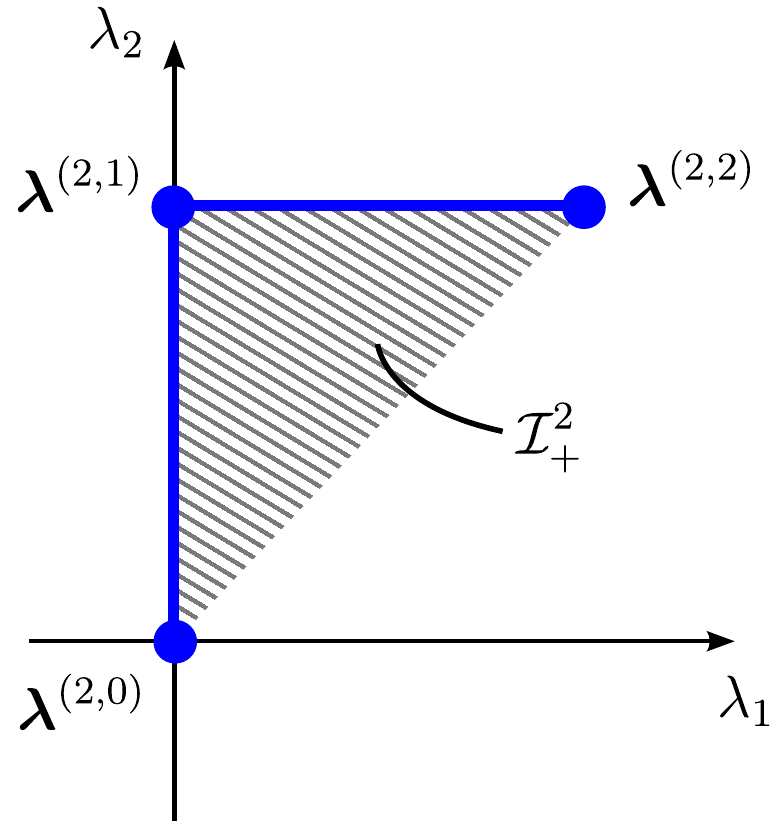}
                \caption{$k=2$}
                \label{fig:g2plus}
        \end{subfigure}%
        ~ 
\qquad
        \begin{subfigure}[b]{0.3\textwidth}
                \centering
                \includegraphics[width=\textwidth]{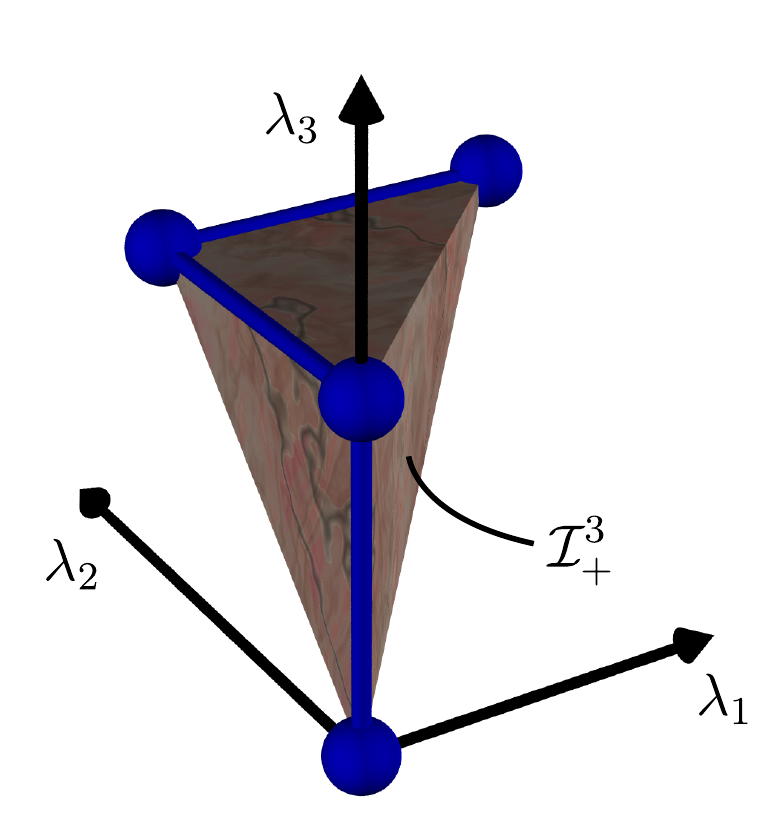}
                \caption{$k=3$}
                \label{fig:g3plus}
        \end{subfigure}
        \caption{The graph $\gcal^k_+$ and its interior $\ical^k_+$}\label{fig:gkplus}
\end{figure}

Note that \ref{it:twsymm} is a rewritten version of Eq.~\eqref{eq:ssymmwedge} in the introduction. Conditions \ref{it:twboundsreal} and \ref{it:twboundsimag} encode the high-energy behavior of the observable and, in the case of \ref{it:twboundsimag}, the geometric position of the localization region $\wcal_r'$.


Finally, we consider analytic functions on the tube domain $\tube(\ical^k_+)$, where 
\begin{equation}\label{eq:ikplus}
 \ical^k_+ := \ich \gcal_+^k = \{\lambdav : 0 < \lambda_1 < \ldots < \lambda_k < \pi \}
\end{equation}
is the interior of (the convex hull of) the graph above. This corresponds to the domain $\rbb^k-i\Lambda_k$ in the notation of \cite{Lechner:2008}. We formulate a locality condition in these terms.

\begin{definition}\label{def:conditionFW}
A collection $F=(F_{k})_{k=0}^\infty$ of functions $\tube(\ical_+^k) \to \cbb$ fulfills condition (FW) if the following holds for any fixed $k$, and with $\thetav\in\rbb^k$, $\zetav \in \cbb^k$ arbitrary:

\begin{enumerate}
\renewcommand{\theenumi}{(FW\arabic{enumi})}
\renewcommand{\labelenumi}{\theenumi}

\item \label{it:fwmero}
\emph{Analyticity:}
$F_k$ is analytic on $\tube(\ical_+^k)$.

\item \label{it:fwsymm} \emph{$S$-symmetry:}
For any $\sigma \in \perms{k}$, we have
$
\displaystyle{
F_k(\thetav+i\zerov)
= S^\sigma(\thetav) F_k(\thetav^\sigma +i\zerov) .
}
$

\item \label{it:fwboundsreal}
\emph{Bounds at nodes:}
For each $j \in \{0,\ldots,k\}$, we have
\begin{equation*}
\| F_k\big( \cdotarg + i \lambdav^{(k,j)} + i \zerov\big) \|_{(k-j) \times j}^{\omega} < \infty.
\end{equation*}

\item \label{it:fwboundsimag}
\emph{Pointwise bounds:}
There exist $c,c'>0$ such that for all $\zetav\in\tube(\ical^k_+)$,
\begin{equation*}
  |F_k(\zetav)| \leq c \operatorname{dist}(\im \zetav,\partial \ical_+^k)^{-k/2} \prod_{j=1}^k \exp \big(\mu r  \im \sinh \zeta_j+ c' \omega(\cosh \re \zeta_j)\big).
\end{equation*}
\end{enumerate}
\noindent
Here $+i\zerov$ denotes the boundary distribution when approached from within $\ical_+^k$.
\end{definition}

Again, \ref{it:fwsymm} corresponds to Eq.~\eqref{eq:ssymmwedge}, and \ref{it:fwboundsreal}, \ref{it:fwboundsimag} encode high-energy behavior of the observable and the position of the wedge.
%
%
The conditions (AW), (TW), (FW) are now equivalent in the following precise sense.

\begin{theorem} \label{theorem:wedgeequiv}
 Let $r\in\rbb$ and an analytic indicatrix $\omega$ be fixed.
\begin{enumerate}
\renewcommand{\theenumi}{(\roman{enumi})}
\renewcommand{\labelenumi}{\theenumi}
 \item \label{it:wedgeequiv-at}
 If $A \in \qf^\omega$ fulfills (AW), then there are distributions $T_k$ fulfilling (TW) such that
\begin{equation}\label{eq:tmnboundary}
\cme{m,n}{A}(\thetav,\etav) = T_{m+n}(\thetav,\etav+i\piv).
\end{equation}
 \item \label{it:wedgeequiv-tf}
If $T_k$ fulfill (TW), then there are functions $F_k$ fulfilling (FW) such that for $0 \leq j \leq k$,
\begin{equation}\label{eq:tkfk}
T_{k}\big(\thetav + i \lambdav^{(k,j)}\big)
= F_{k}\big(\thetav + i \lambdav^{(k,j)} + i\zerov \big).
\end{equation}
 \item \label{it:wedgeequiv-fa} 
  If $F_k$ fulfill (FW), then there is a quadratic form $A$ fulfilling (AW) such that
\begin{equation}\label{eq:fmnfromF}
  \cme{m,n}{A}(\thetav,\etav)=F_{m+n}(\thetav+i\zerov,\etav+i\piv-i\zerov).
\end{equation}
\end{enumerate}
\noindent
Again, $\pm i\zerov$ denotes approach from within $\ical_+^k$.
\end{theorem}

Before proceeding to a proof, let us first note that the conditions (AW), (TW), (FW) as well as the statement of the theorem behave ``covariantly'' under translations along the time-0 axis. Namely, if $A$ fulfills (AW) for some $r$, and if $s \in \rbb$, then $A':=U((0,s),0) A U((0,s),0)^\ast$ fulfills (AW) with $r+s$ instead of $r$. Also, one finds \cite[Prop.~3.9]{BostelmannCadamuro:expansion} that $\cme{m,n}{A'}(\thetav,\etav) = \exp(isp_1(\thetav,\etav+i\piv)) \cme{m,n}{A}(\thetav,\etav)$. Similarly, if $T_k$ fulfill (TW) for some $r$, then $T'_k(\zetav):=\exp(i s p_1(\zetav)) T_k(\zetav)$ fulfills (TW) for $r+s$, and likewise for $F'_k(\zetav):=\exp(i s p_1(\zetav)) F_k(\zetav)$. (With respect to \ref{it:twboundsreal} and \ref{it:fwboundsreal}, one notes here that $\exp(i s p_1(\zetav))$ is a factorizing phase factor on the nodes of the graph.) Therefore, if Thm.~\ref{theorem:wedgeequiv} holds for some $r$, then it holds for $r+s$ as well.

Thus, it suffices to prove Thm.~\ref{theorem:wedgeequiv} in the case $r=0$. We will do this for the three parts of the theorem individually in the following three subsections.


\subsection{(AW) \texorpdfstring{$\Rightarrow$}{=>} (TW)}\label{sec:aw-to-tw}

In order to show the first part of Thm.~\ref{theorem:wedgeequiv}, we start from a quadratic form $A\in\qcal^\omega$ which is $\omega$-local in the wedge $\wcal'$, and study the analyticity properties of its expansion coefficients $\cme{m,n}{A}$ along the graph $\gcal^k_+$. As a first step, we prove the following lemma which is very similar to \cite[Lemma~4.1]{Lechner:2008}, but generalized to our class of observables. Note that this is formulated as a statement about matrix elements of observables in the \emph{right} wedge $\wcal$.

\begin{lemma}\label{lemma:K}
Let $A\in \qf^{\omega}$ be $\omega$-local in $\rightwedge$,  and $\psi\in\Hil_m^\omega$, $\chi\in\Hil_n^\omega$. There exists an analytic function $K:\strip(0,\pi) \to \mathbb{C}$ whose boundary values satisfy
\begin{equation}\label{kboundary}
K(\theta)= \hscalar{ \psi}{ [\zd(\theta),A]\chi },
\quad
K(\theta +i\pi)= \hscalar{ \psi}{[A,z(\theta)]\chi}
\end{equation}
in the sense of distributions. Moreover, there is a constant $c_{mn}$ such that for any $0 \leq \lambda \leq \pi$, 
\begin{equation}\label{eq:kbounds}
  \gnorm{K(\cdotarg + i \lambda)}{2} \leq c_{mn} \, \onorm{\psi}{2} \; \onorm{\chi}{2} \; \onorm{A}{m+n+1} \,.
\end{equation}
\end{lemma}

\begin{proof}
The proof mainly follows \cite{Lechner:2008}, so we confine ourselves to a sketch; see \cite[Lemma~6.2]{Cadamuro:2012} for details in our case.
One considers the time zero fields $\varphi,\pi$ of $\phi$ \cite[Eq.~(3.18)]{Lechner:2008}, and the corresponding expectation values
\begin{equation}
k_{-}(f):=\langle \psi,[\varphi(f),A]\chi\rangle,\quad k_{+}(f):=\langle \psi,[\pi(f),A]\chi\rangle.
\end{equation}
As in Lemma~\ref{lemma:localitychar}\ref{it:chartempered}, these $k_{\pm}$ are Schwartz distributions with support in the right half-line. Hence their Fourier-Laplace transforms $\tilde k_\pm$ are analytic on the lower half plane and fulfill Paley-Wiener type bounds \cite[Thm.~IX.16]{ReedSimon:1975-2}.
Our function $K$ is then given by
\begin{equation}
 K(\zeta) := \frac{1}{2}\Big( \mu \cosh(\zeta) \,\tilde k_-(-\mu\sinh\zeta) - i \tilde k_+(-\mu\sinh\zeta) \Big),
\end{equation}
which works out to have the proposed boundary values \eqref{kboundary}. 

The bound \eqref{eq:kbounds} follows for $\lambda = 0$ and $\lambda = \pi$ by direct estimates from \eqref{kboundary}, using \eqref{eq:aomeganorm}. For general $\lambda$, it is then a consequence of the three-lines theorem, as in \cite[Eq.~(4.15)]{Lechner:2008}.
\cmpqed\end{proof}

This allows us to find analytic continuations of the distributions $\cme{m,n}{A}$ if $A$ is localized in the \emph{left} wedge, again very similar to \cite{Lechner:2008}.\footnote{%
Note that our coefficients $\cme{m,n}{A}$ are equal to $\langle J A^\ast J\rangle_{m+n,m}^{\mathrm{con}}$ in the notation of \cite{Lechner:2008}; correspondingly, \cite{Lechner:2008} obtains the analogue of our Prop.~\ref{proposition:analpositivesimplex} for localization in the \emph{right} wedge.}

\begin{proposition}\label{proposition:analpositivesimplex}

If $A\in \mathcal{Q}^{\omega}$ is $\omega$-local in $\leftwedge$, then there are CR distributions $T_k$ on $\tube(\gcal^k_+)$ such that
\begin{equation}\label{eq:tkboundaryvalue}
T_k(\thetav+i\lambdav^{(k,j)})=\cme{k-j,j}{A}(\thetav), \quad
0\leq j \leq k.
\end{equation}
Further, there exists a constant $c_k$ such that for all $\lambdav\in\bar\gcal^k_+$,
\begin{equation}\label{eq:tkobounds}
\gnorm{\thetav \mapsto T_k(\thetav + i \lambdav) e^{-\sum_j \omega(\cosh \theta_j)}}{\times} \leq c_k \onorm{A}{k}. \end{equation}
\end{proposition}

\begin{proof}
Again, we confine ourselves to a sketch, and refer the reader to \cite[Lemma~6.3]{Cadamuro:2012} for details. One shows that for any $m \geq 1$, $n \geq 0$, the distribution $\cme{m,n}{A}(\thetav,\etav)$ has an analytic continuation in the variable $\theta_{m}$ to the strip $\strip(0,\pi)$, with its distributional boundary value at $\im\theta_{m}=\pi$ given by
\begin{equation}\label{eq:analbouf}
\cme{m,n}{A}(\theta_1,\ldots,\theta_{m}+i\pi,\ldots,\theta_{m+n})
=
\cme{m-1,n+1}{A}(\theta_1,\ldots,\theta_{m},\ldots,\theta_{m+n}).
\end{equation}
This is obtained by rewriting the definition of $\cme{m,n}{A}$ as a sum of matrix elements of commutators $[\zd(\theta), J A\st J]$, as in \cite[Lemma 4.2]{Lechner:2008}. One then applies Lemma~\ref{lemma:K} to obtain the desired analytic continuation. 

Using \eqref{eq:analbouf} repeatedly, we can construct $T_k$ as the analytic continuation of $\cme{k,0}{A}$ along $\tube(\gcal_+^k)$ with the boundary values \eqref{eq:tkboundaryvalue}. Since these are independent of direction, the $T_k$ are CR distributions. By direct computation from the representation of $\cme{m,n}{A}$ in \cite[Lemma 4.2]{Lechner:2008}, using the estimate \eqref{eq:kbounds}, we can also obtain
\begin{multline}
  \Big\lvert \int d\thetav \,\cme{m,n}{A} (\theta_1,\ldots,\theta_{m}+i\lambda,\ldots,\theta_{m+n}) g_1(\theta_1)\cdots g_{m+n}(\theta_{m+n})\Big\rvert
 \\ \leq c_{m+n} \onorm{A}{m+n} \gnorm{g_m}{2} \prod_{j \neq m} \onorm{g_j}{2}
\end{multline}
for any $g_1,\ldots,g_{m+n}\in\dcal(\rbb)$, with some $c_{m+n}>0$ depending only on $m+n$. The proposed bounds \eqref{eq:tkobounds} then follow immediately.
\cmpqed\end{proof}

We are now in the position to prove that the $T_k$ above fulfill all conditions (TW), summarizing the result of the subsection.

\begin{theopargself}
\begin{proof}[\ofwhat{Thm.~\ref{theorem:wedgeequiv}\ref{it:wedgeequiv-at}}]
Let $r=0$. Given $A$ which is $\omega$-local in $\leftwedge$, Prop.~\ref{proposition:analpositivesimplex} gives us distributions $T_k$ with the boundary values \eqref{eq:tmnboundary} and the analyticity property \ref{it:twmero}.
$S$-symmetry \ref{it:twsymm} follows since $\cme{k,0}{A}$ has this property \cite[Prop.~3.4]{BostelmannCadamuro:expansion}. The bounds \ref{it:twboundsreal} are immediate, since 
\begin{equation}\label{eq:fkpnodebounds}
  \gnorm{ T_k (\cdotarg + i \lambdav^{(k,j)})}{(k-j) \times j}^{\omega}
 = \gnorm{ \cme{k-j,j}{A} }{(k-j) \times j}^{\omega} < \infty;
\end{equation}
see \cite[Prop.~3.3]{BostelmannCadamuro:expansion}. 
For \ref{it:twboundsimag}, consider the distribution
\begin{equation}\label{eq:fkhatdef}
   T_k' (\zetav) := T_k (\zetav) \exp(-\sum_{j}\oa(\sinh\zeta_{j})),
\end{equation}
which is CR on $\tube(\gcal_+^k)$. By \ref{it:omegaestimate} and \ref{it:sublinear}, the exponential factor fulfills the estimate
\begin{equation}
|e^{-\oa(\sinh\zeta_{j})}|\leq e^{-\omega(\cosh\re\zeta_{j})}e^{\omega(1)}.
\end{equation}
Therefore, from \eqref{eq:tkobounds} and \eqref{eq:crossnormfactor}, $T_k'$ fulfills the bound
\begin{equation}
  \gnorm{ T_k' (\cdotarg + i \lambdav)}{ \times }
  \leq e^{k\cdot \omega(1)} c_k \onorm{A}{k},
\end{equation}
which implies \ref{it:twboundsimag}.
\cmpqed\end{proof}
\end{theopargself}

\subsection{(TW) \texorpdfstring{$\Rightarrow$}{=>} (FW)}\label{sec:tw-to-fw}

For the second part of Thm.~\ref{theorem:wedgeequiv}, we start from CR distributions $T_k$ on $\tube(\gcal_+^k)$ and extend them as analytic functions to the interior of the graph. The techniques for this have already been introduced in Sec.~\ref{sec:graphs}.

\begin{theopargself}
\begin{proof}[\ofwhat{Thm.~\ref{theorem:wedgeequiv}\ref{it:wedgeequiv-tf}}]
Let $T_k$ fulfill (TW) with $r=0$. Using Lemma~\ref{lem:graphtube}, we can find analytic functions $F_k$ on $\tube(\ical_+^k)$ which have $T_k$ as boundary distributions, as required for \eqref{eq:tkfk}. The $F_k$ evidently fulfill \ref{it:fwmero}. Also, \ref{it:fwsymm} and \ref{it:fwboundsreal} are immediate from  \ref{it:twsymm} and \ref{it:twboundsreal}, respectively. For \ref{it:fwboundsimag}, we consider the function 
\begin{equation}
F_k'(\zetav):= F_k( \zetav ) \, \exp\big(-\sum_{j}\oa(\sinh \zeta_{j})\big).  
\end{equation}
From condition \ref{it:twboundsimag}, we know that
\begin{equation}
   \gnorm{F'_k(\cdotarg + i \lambdav)}{\times} \leq c \quad \text{for all $\lambdav\in\bar\gcal_+^k$. }
\end{equation}
Due to the maximum modulus principle, Lemma~\ref{lemma:maxmodcross}, the same bound holds for all $\lambdav\in\ical_+^k$. Prop.~\ref{proposition:pointwise} then yields
\begin{equation}
   \lvert F'_k(\thetav + i \lambdav) \rvert \leq c'  \operatorname{dist}(\lambdav,\partial \ical_+^k)^{-k/2}.
\end{equation}
Since $\re\oa(\sinh (\theta_j+i\lambda_j)) \leq a_\omega \omega(\cosh \theta_j) + b_\omega$ by \ref{it:omegaestimate}, this gives the  bounds \ref{it:fwboundsimag} for $F_k$.
\cmpqed\end{proof}
\end{theopargself}

\subsection{(FW) \texorpdfstring{$\Rightarrow$}{=>} (AW)}\label{sec:fw-to-aw}

For the last part of Thm.~\ref{theorem:wedgeequiv}, we set out from analytic functions $F_k$ fulfilling (FW) with $r=0$,  and construct a quadratic form $A$ which is $\omega$-local in $\wcal'$. In fact, we define $A$ by its series expansion,
\begin{equation}\label{eq:afromf-w}
A := \sum_{m,n=0}^\infty \int \frac{d\thetav\,d\etav }{m!n!}
F_{m+n}(\thetav + i \zerov, \etav+i\piv-i\zerov)  z^{\dagger m}(\thetav) z^n(\etav).
\end{equation}
We first remark that \eqref{eq:afromf-w} is well-defined. Setting $g_{mn}(\thetav,\etav):= F_{m+n}(\thetav+i\zerov,\etav+i\piv-i\zerov)$, it follows from \ref{it:fwboundsreal} that $\onorm{g_{mn}}{m\times n}<\infty$.
Thus, as a consequence of Thm.~\ref{theorem:expansion}, the series in \eqref{eq:afromf-w} gives a well-defined quadratic form $A\in\qf^\omega$.

To prove that $A$ is $\omega$-local in $\leftwedge$, we need to establish that its commutator with the wedge-local field $\phi'(x)$ vanishes if $x>0$; we do this on the level of the expansion terms in \eqref{eq:afromf-w}. 
To that end, we first recall the commutation relations of $z,z^\dagger$ with $z', z^{\dagger \prime}$ \cite[Lemma 4]{Lechner:2003}. For $g\in \mathcal{H}_{1}$, the following holds in the sense of operator-valued distributions on $\fpn$:
\begin{align}
 [z(\overline{g})',z^{\dagger}(\theta)]&=B^{g,\theta},
 &
[z^{\dagger}(\overline{g})',z(\theta)]&=-(B^{\bar g,\theta})^{*} \label{eq:comzpz}
\\
[z(g)',z(\theta)]&=0, & [z^{\dagger}(g)',z^{\dagger}(\theta)]&=0,\label{eq:zzp}
\end{align}
where $B^{g,\theta}=\oplus_{n=0}^{\infty}B_{n}^{g,\theta}$, and where $B_{n}^{g,\theta}$ acts on $\Hil_n$ as a multiplication operator,
\begin{equation}\label{multop}
B_{n}^{g,\theta}(\theta_{1},\ldots,\theta_{n})=g(\theta)\prod_{j=1}^{n}S(\theta-\theta_{j}).
\end{equation}
We show a generalization of \eqref{eq:comzpz} to normal-ordered products of annihilators and creators.

\begin{lemma}
Let $g\in\mathcal{H}_{1}$. The following commutation relations hold in the sense of operator-valued distributions on $\fpn$:
\begin{align}
[z(\overline{g})',\zd(\theta_{1})\ldots \zd(\theta_{m})]&=\sum_{j=1}^{m} \Big(\prod_{l=j+1}^{m}S(\theta_{j}-\theta_{l})\Big)\zd(\theta_{1})\ldots\widehat{\zd(\theta_{j})}\ldots \zd(\theta_{m})B^{g,\theta_{j}},\label{comm2}
\\
[\zd(\overline{g})',z(\theta_{1})\ldots z(\theta_{m})]&=-\sum_{j=1}^{m}\Big(\prod_{l=1}^{j-1}S(\theta_{l}-\theta_{j})\Big)(B^{\bar g,\theta_{j}})^{*}z(\theta_{1})\ldots\widehat{z(\theta_{j})}\ldots z(\theta_{m}).\label{comm4}
\end{align}
\end{lemma}

\begin{proof}
Our proof of Eq.~\eqref{comm2} is based on induction on $m$.
For $m=1$, Eq.~\eqref{comm2} reduces to \eqref{eq:comzpz}, and is proved as in \cite[Lemma 4]{Lechner:2003}.

Assume now that Eq.~\eqref{comm2} holds for $m-1$ in place of $m$. Using the induction hypothesis and Eq.~\eqref{eq:comzpz}, we have
\begin{equation}\label{eq:computecommutator}
\begin{aligned}
\lbrack z(\overline{g})',& \zd(\theta_{1})\ldots  \zd(\theta_{m-1})\zd(\theta_{m}) \rbrack\\
&=[z(\overline{g})',\zd(\theta_{1})\ldots \zd(\theta_{m-1})]\zd(\theta_{m})
+\zd(\theta_{1})\ldots \zd(\theta_{m-1})[z(\overline{g})',\zd(\theta_{m})]\\
&=\sum_{j=1}^{m-1}\Big(\prod_{l=j+1}^{m-1}S(\theta_{j}-\theta_{l})\Big) \zd(\theta_{1})\ldots \widehat{\zd(\theta_{j})}\ldots \zd(\theta_{m-1})B^{g,\theta_{j}}\zd(\theta_{m})\\
&\qquad\qquad+\zd(\theta_{1})\ldots \zd(\theta_{m-1})B^{g,\theta_{m}}.
\end{aligned}
\end{equation}
With the help of the exchange relation
\begin{equation}\label{combz}
B^{g,\theta_j}\zd(\theta_m)=S(\theta_j-\theta_m)\zd(\theta_m)B^{g,\theta_j},
\end{equation}
which can be computed directly from the definitions,
we obtain \eqref{comm2} from \eqref{eq:computecommutator}.
Now \eqref{comm4} follows from \eqref{comm2} by taking adjoints.
\cmpqed\end{proof}

Equipped with these tools, we can compute $[A,\phi'(x)]$ for a generic $A\in\qf^\omega$ in terms of its expansion coefficients.

\begin{proposition}\label{proposition:aphicomm}
Let $A\in \qf^{\omega}$, $g\in \dcal(\rbb^2)$. In the sense of matrix elements on $\fpno \times \fpno$, it holds that
\begin{multline}\label{maincommutator}
[A,\phi'(g)]=\sum_{m,n=0}^\infty \int \frac{d\thetav d\etav}{m!n!}\int d\xi\;
 \Big( \cme{m,n+1}{A} (\thetav, \xi , \etav) z^{\dagger m}(\thetav) (B^{\overline{g^{+}},\xi})^{*}  z^{n}(\etav) \\
- \cme{m+1,n}{A} (\thetav,\xi,\etav) z^{\dagger m}(\thetav) B^{g^{-},\xi}  z^{n}(\etav)\Big).
\end{multline}
\end{proposition}

Note that the infinite sums above are actually finite in matrix elements; convergence questions do not arise.

\begin{proof}
We need to compute the commutator
\begin{equation}
[A,\phi'(g)]=\Big\lbrack \sum_{m,n=0}^{\infty}\int \frac{d\thetav d\etav}{m!n!}\;\cme{m,n}{A}(\thetav,\etav) z^{\dagger m}(\thetav)z^{n}(\etav),\;z^{\dagger}(\overline{g^{+}})'+z(\overline{g^{-}})'\Big\rbrack.
\end{equation}
Using \eqref{eq:zzp}, \eqref{comm2} and \eqref{comm4}, this expands to
\begin{equation}
 \begin{aligned}
\lbrack A,\phi'(g) \rbrack =& \hphantom{-}
\sum_{m\geq 0, n\geq 1}\int \frac{d\thetav d\etav}{m!n!}\;\cme{m,n}{A}(\thetav,\etav) \sum_{j=1}^{n}\ \Big(\prod_{l=1}^{j-1}S(\eta_{l}-\eta_{j})\Big)
\\
&\qquad\qquad\qquad \times z^{\dagger m}(\thetav)(B^{\overline{g^{+}},\eta_{j}})^{*}z(\eta_{1})\ldots\widehat{z(\eta_{j})}\ldots z(\eta_{n})\\
&-\sum_{m\geq 1, n\geq 0}\int \frac{d\thetav d\etav}{m!n!}\;\cme{m,n}{A}(\thetav,\etav)
\sum_{j=1}^{m}\Big(\prod_{l=j+1}^{m}S(\theta_{j}-\theta_{l})\Big)
\\
&\qquad\qquad\qquad \times z^{\dagger}(\theta_{1})\ldots\widehat{\zd(\theta_{j})}\ldots \zd(\theta_{m})B^{g^{-},\theta_{j}}z^{n}(\etav).
\end{aligned}
\end{equation}
We set $\eta_{j}=:\xi$ in the first sum and $\theta_{j}=:\xi$ in the second sum, and permute the argument of $\cme{m,n}{A}$ so that they read $(\hat\thetav,\xi,\etav)$ and $(\thetav,\xi,\hat\etav)$ respectively, noting that this cancels the $S$-factors in the sums due to $S$-symmetry of the $\cme{m,n}{A}$ \cite[Prop.~3.4]{BostelmannCadamuro:expansion}. This yields:
\begin{multline}
[A,\phi'(g)]=\sum_{m\geq 0, n\geq 1}\int \frac{d\thetav d\hat{\etav}}{m!(n-1)!}\int d\xi\;
 \cme{m+n}{A}(\thetav,\xi,\hat\etav)
z^{\dagger m}(\thetav)(B^{\overline{g^{+}},\xi})^{*}z^{n-1}(\hat{\etav})\\
-\sum_{m\geq 1, n\geq 0}\int \frac{d\hat{\thetav} d\etav}{(m-1)!n!}\int d\xi\;
 \cme{m+n}{A}(\hat{\thetav},\xi,\etav)z^{\dagger m-1}(\hat{\thetav})B^{g^{-},\xi}z^{n}(\etav).
\end{multline}
The result \eqref{maincommutator} now follows by a relabeling of the summation indices.
\cmpqed\end{proof}

We will now prove from conditions (FW) that the operator $A$, as given as in \eqref{eq:afromf-w}, is localized in $\leftwedge$. To that end, we need to show that $[A,\phi'(g)]$ vanishes if $g$ has support in $\rightwedge$. The idea for the proof is based on Prop.~\ref{proposition:aphicomm}, and works as follows. For our specific $A$, we have $\cme{m,n}{A}(\thetav,\etav)=F_{m+n}(\thetav+i\zerov,\etav+i\piv-i\zerov)$; it follows that $\cme{m+1,n}{A}(\thetav,\xi+i\pi,\etav)=\cme{m,n+1}{A}(\thetav,\xi,\etav)$. Also, we have (at least formally) that $B^{g^-,\xi+i \pi} = (B^{\overline{g^+},\xi})\st$, since $g^-(\theta+i\pi)=g^+(\theta)$. Inserting this into Eq.~\eqref{maincommutator}, we see that $[A,\phi'(g)]$ vanishes if it is allowed to shift the integration contour in $\xi$ from $\rbb$ to $\rbb+i\pi$.

Whether it is actually permitted to shift this integration contour is crucially dependent on the growth behavior of the analytic functions involved, and by this means, dependent on the localization regions of $F_k$ and $g$. We will therefore first analyze this growth behavior carefully. To that end, let $m,n,q\in\nbb_0$, $f \in \dcal(\rbb^{m+n})$, and $\nuv \in \rbb^q$ be fixed (until Lemma~\ref{lemma:integralshift} inclusive). We define
\begin{equation}\label{eq:kxi}
   K(\xi) := \Big(\prod_{j=1}^q S(\xi-\nu_j)\Big)\int d\thetav d\etav  \, f(\thetav,\etav)  F_{m+n+1}(\thetav+i\zerov,\xi,\etav+i\piv-i\zerov).
\end{equation}
By our assumptions on $F_k$, this $K$ is analytic for $\xi\in\strip(0,\pi)$, with distributional boundary values. We can derive bounds for $K$ near the boundary.

\begin{lemma} \label{lemma:Kbounds}
If $F_{m+n+1}$ fulfills \ref{it:fwmero} and \ref{it:fwboundsimag} for $r=0$, then there exist $c,c'>0$ such that
\begin{equation}
    \lvert K(\xi+i\lambda) \rvert
  \leq
   \frac{c \, e^{c' \omega(\cosh \xi)}}{(\lambda(\pi-\lambda))^{(m+n)/2}}, \quad 0 < \lambda < \pi.
\end{equation}
\end{lemma}
\begin{proof}
 We set $h:=\min(\lambda,\pi-\lambda)/(m+n+1)$ and $\nuv_L := (1,2,\ldots,m)$, $\nuv_R := (n,\ldots,2,1)$. For fixed $\xi,\lambda,\thetav,\etav$, set
\begin{equation}
   G(z):= F_{m+n+1}(\thetav+zh\nuv_L,\xi+i\lambda,\etav-zh\nuv_R).
\end{equation}
This function is analytic in $z \in \rbb+i(0,1)$, and for the imaginary part of the argument of $F_{m+n+1}$, we have
\begin{equation}
 \operatorname{dist}\big( (h \im z, \ldots, mh\im z, \lambda, \pi-nh\im z, \ldots, \pi-h \im z)   ,\partial\ical^k_+\big) \geq h \im z.
\end{equation}
Then, for any $\rho>0$, condition \ref{it:fwboundsimag} yields constants $c_\rho$ (dependent on $\rho$, but not on $\xi,\lambda$) and $c'$ such that
\begin{multline}\label{eq:gzbound}
   |G(z)| \leq c_\rho  e^{c' \omega(\cosh \xi)} (h \lvert \im z\rvert)^{-k/2}
\\
\quad \text{for all }  z \in (-\rho,\rho)+i(0,1), \; \|\thetav\| \leq \rho, \; \|\etav\| \leq \rho.
\end{multline}
Using standard techniques -- see, e.g., \cite[Prop.~4.2]{BostelmannFewster:2009} -- one obtains from \eqref{eq:gzbound} the following estimate for the boundary distribution:
\begin{equation}
 \Big\lvert
\int G(x+i0) g(x)\,dx
\Big\rvert
\leq c_{g,\rho} h^{-k/2}
   e^{c' \omega(\cosh \xi)}
\leq \frac{ c_{g,\rho} \, e^{c' \omega(\cosh \xi)}} {(3(m+n+1) \lambda(\pi-\lambda) )^{k/2}}
\end{equation}
where the constant $c_{g,\rho}$ may depend on the test function $g\in\dcal(-\rho,\rho)$ and the cutoff $\rho$, but \emph{not} on $G$ (and hence not on $\xi,\lambda,\thetav,\etav$).  In view of the definition \eqref{eq:kxi} of $K$, where the factors $S(\xi-\nu_j)$ are bounded functions on the strip $\strip(0,\pi)$, this yields the proposed result.
\cmpqed\end{proof}

This enables us to describe details of the proposed shifting of integral contours.

\begin{lemma} \label{lemma:integralshift}
If $F_{m+n+1}$ fulfills \ref{it:fwmero} and \ref{it:fwboundsimag} for $r=0$, then there exists an analytic indicatrix $\omega'\geq \omega$ such that for all $g \in \dcal^{\omega'}(\wcal)$,
\begin{equation}
  \int K(\xi+i0) g^-(\xi) d\xi = \int K(\xi+i\pi-i0) g^+(\xi) d\xi.
\end{equation}
\end{lemma}

\begin{proof}
We set $\omega'(p):=(a_\omega+2) (c'\omega(p)+ \frac{m+n+6}{2} \log(1+p))$, with $c'$ as in Lemma~\ref{lemma:Kbounds}. One checks that $\omega'$, as a linear combination of the analytic indicatrices $\omega$ and \eqref{eq:omegalog} with positive coefficients, is itself an analytic indicatrix, with $a_{\omega'}=a_\omega+2$. 
Then, from Lemma~\ref{lemma:Kbounds} and Prop.~\ref{prop:omegapw}, we know that for fixed $g \in \dcal^{\omega'}(\wcal)$ and $\epsilon>0$, there is $c_\epsilon>0$ such that
\begin{equation}
\forall \lambda \in [\epsilon, \pi-\epsilon]: \quad \big\lvert g^-(\xi+i\lambda)K(\xi+i\lambda) \big\rvert \leq \frac{c_\epsilon}{(1+\cosh\xi)^{(m+n+6)/2}}.
\end{equation}
Hence by Cauchy's formula,
\begin{equation}\label{prop}
\forall \epsilon>0: \int d\xi\; K(\xi+i\epsilon)g^-(\xi+i\epsilon)=\int d\xi\; K(\xi+i\pi-i\epsilon)g^+(\xi-i\epsilon).
\end{equation}
We will show below that
\begin{equation}\label{contour}
\lim_{\epsilon \searrow 0}\int K(\xi+i\epsilon) g^-(\xi+i\epsilon)\;d\xi = \int K(\xi+i0)g^-(\xi)\;d\xi,
\end{equation}
which then holds similarly for the upper boundary. The result now follows from \eqref{prop} as $\epsilon \searrow 0$.---%
For Eq.~\eqref{contour}, we need to show that
\begin{equation}\label{eq:klimit}
\lim_{\epsilon \searrow 0} \int d\xi\; K(\xi+i\epsilon) \big( g^-(\xi)-g^-(\xi+i\epsilon) \big)=0.
\end{equation}
Let $K^{(-\ell)}$ be the $\ell$th antiderivative of $K$, normalized to $K^{(-\ell)}(i\frac{\pi}{2})=0$.
Using the bounds of Lemma~\ref{lemma:Kbounds}, we find by integration that for $\ell > (m+n)/2$, with some $c''>0$,
\begin{equation}
  |K^{(-\ell)}(\xi+i\lambda)|\leq c'' (1+|\xi|)^\ell e^{c' \omega(\cosh \xi)}.
\end{equation}
Integrating by parts in \eqref{eq:klimit}, this yields
\begin{multline}
\lim_{\epsilon \searrow 0}\left|  \int d\xi\; K(\xi+i\epsilon) \big( g^-(\xi)-g^-(\xi+i\epsilon) \big) \right|
\\
=\lim_{\epsilon\searrow 0}\left|  \int d\xi\; K^{(-\ell)}(\xi+i\epsilon)\left( \partial_\xi^\ell g^-(\xi)-\partial_\xi^\ell g^-(\xi+i\epsilon)\right)\right|\\
\leq c'' \lim_{\epsilon\searrow 0} \epsilon \int d\xi\; (1+|\xi|)^\ell e^{c' \omega(\cosh \xi)} \sup_{0 < \lambda< \pi } |\partial_\xi^{\ell+1} g^- (\xi+i\lambda)| = 0
\end{multline}
if we can show that the integral in the last line exists. Indeed, using the bounds on $\partial_\xi^{\ell+1} g^-$ from Prop.~\ref{prop:omegapw}, we know that for all $\lambda\in(0,\pi)$,
\begin{equation}
\begin{aligned}
 |\partial_\xi^{\ell+1} g^-(\xi+i\lambda)| &\leq c''' (\cosh \xi)^{\ell+1} e^{-\omega'(\cosh\xi)/a_{\omega'}}
  \\ &\leq  c''' (\cosh \xi)^{\ell-(m+n)/2-2}  e^{-c' \omega(\cosh \xi)},
\end{aligned}
\end{equation}
which makes the integral convergent if we choose $ m+n < 2 \ell \leq m+n+2$.
\cmpqed\end{proof}

This finally allows us to prove that $A$ is local in $\leftwedge$. We summarize:

\begin{theopargself}
\begin{proof}[\ofwhat{Thm.~\ref{theorem:wedgeequiv}\ref{it:wedgeequiv-fa}}]
Let $F_k$ be functions fulfilling (FW) with $r=0$. We define $A \in \qf^\omega$ by \eqref{eq:afromf-w}; we already noted that this is well-defined due to \ref{it:fwboundsreal}. Also, the expansion coefficients $F_{m+n}(\thetav+i\zerov,\etav+i\piv-i\zerov)$ are $S$-symmetric in $\thetav,\etav$ separately; this follows from \ref{it:fwsymm} by analytic continuation. Hence, by the uniqueness result in Thm.~\ref{theorem:expansion}, we have $\cme{m,n}{A}(\thetav,\etav)=F_{m+n}(\thetav+i\zerov,\etav+i\piv-i\zerov)$ as required for \eqref{eq:fmnfromF}. 

We claim that $A$ is $\omega$-local in $\leftwedge$. 
By Lemma~\ref{lemma:localitychar}\ref{it:charomegavar}, it suffices to show that for fixed $\psi,\chi\in\fpno$, one has $\hscalar{\psi}{ [A,\phi'(g)]\chi}=0$ for all $g \in \dcal^{\omega'}(\wcal)$, with an indicatrix $\omega'$ suitably chosen. By density arguments, we can assume that $\psi,\chi$ have fixed particle number and compact support in rapidity space. Using Prop.~\ref{proposition:aphicomm},
and considering a summand with fixed $m,n$ in Eq.~\eqref{maincommutator}, it suffices to show that for $g\in \mathcal{D}^{\omega'}(\wcal)$, for fixed $q \in \nbb_0$, and for fixed $f\in\dcal(\rbb^{m+n})$,
\begin{multline}
\int d\thetav d\etav  \int d\xi \; f(\thetav,\etav) \Big(F_{m+n+1}(\thetav+i\zerov, \xi +i\pi-i0, \etav +i\piv-i\zerov)(B_{q}^{\overline{g^{+}},\xi})^{*}\\
-F_{m+n+1}(\thetav+i\zerov,\xi+i0,\etav +i\piv-i\zerov)B_{q}^{g^{-},\xi}\Big)=0.\label{vanishing}
\end{multline}
With the definitions \eqref{multop} and \eqref{eq:kxi}, this claim rewrites to
\begin{equation}
  \int K(\xi+i0) g^-(\xi) d\xi = \int K(\xi+i\pi-i0) g^+(\xi) d\xi.
\end{equation}
But this is guaranteed by Lemma~\ref{lemma:integralshift}.
\cmpqed\end{proof}
\end{theopargself}

%% file: doublecones.tex
\section{Locality in a double cone}\label{sec:doublecones}

We now extend our analysis in the previous section to a characterization of observables localized in compact regions, more precisely, in a standard double cone $\ocal_r$ of radius $r$ around the origin. Again, we will formulate locality conditions (AD), (TD), and (FD) for quadratic forms, CR distributions, and analytic functions, respectively, and show their equivalence. 

Since $\ocal_r \subset \wcal_r'$, these new conditions need to be stronger than (AW), (TW) and (FW) before; but the way in which these conditions are strengthened involves some entirely new aspects. In particular, the functions $F_k$ will now extend meromorphically to all of $\cbb^k$. Moreover, as indicated in \eqref{eq:residue1k}, the residue of $F_k$ at the kinematic poles $\zeta_n-\zeta_m=i\pi$ ($m < n$) have a prescribed value involving  $F_{k-2}$, thus giving relations between the $F_k$ of different orders. We will refer to these as \emph{recursion relations}.


Let us formulate the strengthened locality conditions. The one on the level of quadratic forms is again easy to state.

\begin{definition}\label{def:conditionAD}
$A \in \qf^\omega$ fulfills condition (AD) if it is $\omega$-local in $\ocal_r$.
\end{definition}


On the other hand, the conditions for CR distributions differ noticeably from those in the wedge local case; in particular, they refer to a larger graph. Besides  $\gcal_+^k$ as introduced in Sec.~\ref{sec:wedges}, we also consider the graph  $\gcal_-^k=\gcal_+^k - \piv$, i.e., with all components of the nodes and edges shifted by $-\pi$. We label the nodes of $\gcal_-^k$ as $\lambdav^{(k,-j)}:=\lambdav^{(k,k-j)}-\piv$ ($j=0,\ldots,k$). Further, let $\gcal_0^k$ be the union of $\gcal_+^k$ and $\gcal_-^k$, noting that the two graphs have $\lambdav^{(k,0)}=0$ as a common node. See Fig.~\ref{fig:gk1} for a sketch of these graphs in $k=2$ and $k=3$ dimensions.
We remark at this point that the graphs can alternatively be written as follows:
\begin{align}
   \label{setgplus}
   \gcal_+^k &= \{  \text{edges} :  0 \leq \lambda_1 \leq \ldots \leq \lambda_k \leq \pi \},
\\
   \label{setgminus}
   \gcal_-^k &= \{  \text{edges} : -\pi \leq \lambda_1 \leq \ldots \leq \lambda_k \leq 0 \},
\\
   \label{setg0}
   \gcal_0^k &= \{  \text{edges} : -\pi \leq \lambda_1 \leq \ldots \leq \lambda_k \leq \lambda_1 +  \pi\leq 2\pi \}.
\end{align}
(Our shorthand notation $\{\text{edges} : C(\lambdav) \}$ denotes the graph containing all those next-neighbor edges on the grid $\pi\zbb^k$ where the condition $C(\lambdav)$ is true for all $\lambdav$ on the edge; the nodes of the graph are the end points of these edges.)

The CR distributions $T_k$ will now be defined on $\gcal_0^k$ rather than $\gcal_+^k$. Besides an appropriate extension of the previous conditions (TW) to this graph, we will also need to add a periodicity condition on $\gcal_0^k$, which will give rise to \eqref{eq:introperio} later, and a form of the recursion relations mentioned above.
This relation involves the factors $S_C$ and $R_C$ as defined in \eqref{eq:sc} and \eqref{eq:rc}.

\begin{definition} \label{def:conditionTD}
A collection $T=(T_{k})_{k=0}^\infty$ of distributions on $\tube(\gcal^k_0)$ fulfills condition (TD) if the following holds for any fixed $k$, and with $\thetav \in \rbb^k$ arbitrary:
\begin{enumerate}
\renewcommand{\theenumi}{(TD\arabic{enumi})}
\renewcommand{\labelenumi}{\theenumi}

\item \label{it:tdmero} \emph{Analyticity:} $T_k$ is a CR distribution on $\tube(\gcal^k_0)$.

\item \label{it:tdsymm} \emph{$S$-symmetry}:
For any $\sigma \in \perms{k}$, we have
$\displaystyle{
T_k(\thetav)= S^\sigma(\thetav) T_k(\thetav^\sigma) .
}$

\item \label{it:tdperiod} \emph{Periodicity:}
$
T_k (\thetav + i \lambdav^{(k,-k)}) =
T_k (\thetav + i \lambdav^{(k,k)}).
$

\item \label{it:tdrecursion} \emph{Recursion relations:}
For any $0 \leq m \leq k$,
\begin{equation*}
T_k(\thetav +i\lambdav^{(k,-m)})
=\sum_{C\in\ccal_{m,k-m}}(-1)^{|C|}\delta_{C}
S_{C} R_{C}(\thetav) T_{k-2|C|}(\check \thetav + i \lambdav^{(k-2|C|,m-|C|)}),
\end{equation*}
where $\check\thetav = (\theta_{m+1},\ldots,\widehat{\theta_{r_{1}}},\ldots,\widehat{\theta_{r_{|C|}}},\ldots,\theta_{k},\theta_{1},\ldots,\widehat{\theta_{l_{1}}},\ldots,\widehat{\theta_{l_{|C|}}},\ldots,\theta_{m}).$

\item \label{it:tdboundsreal}
\emph{Bounds at nodes:}
For any $j \in \{0,\ldots,k\}$,
\begin{equation*}
\onorm{ T_k( \cdotarg + i \lambdav^{(k,j)} )}{(k-j) \times j} < \infty,
 \quad
\onorm{ T_k( \cdotarg + i \lambdav^{(k,-j)} )}{j \times (k-j)} < \infty.
\end{equation*}

\item \label{it:tdboundsimag}
\emph{Bounds at edges:}
There exists $c>0$ such that for any $\lambdav\in \bar\gcal^{k}_{\pm}$,
\begin{equation*}
\gnorm{ e^{\pm i \mu r \sum_j \sinh \zeta_j}e^{-\sum_{j}\oa(\pm \sinh \zeta_{j})}
T_k( \pmb{\zeta})\big\vert_{\zetav=\cdotarg +i\lambdav} }{\times} \leq c.
\end{equation*}

\end{enumerate}

\end{definition}

\begin{figure}
        \centering
        \begin{subfigure}[b]{0.48\textwidth}
                \centering
                \includegraphics[width=\textwidth]{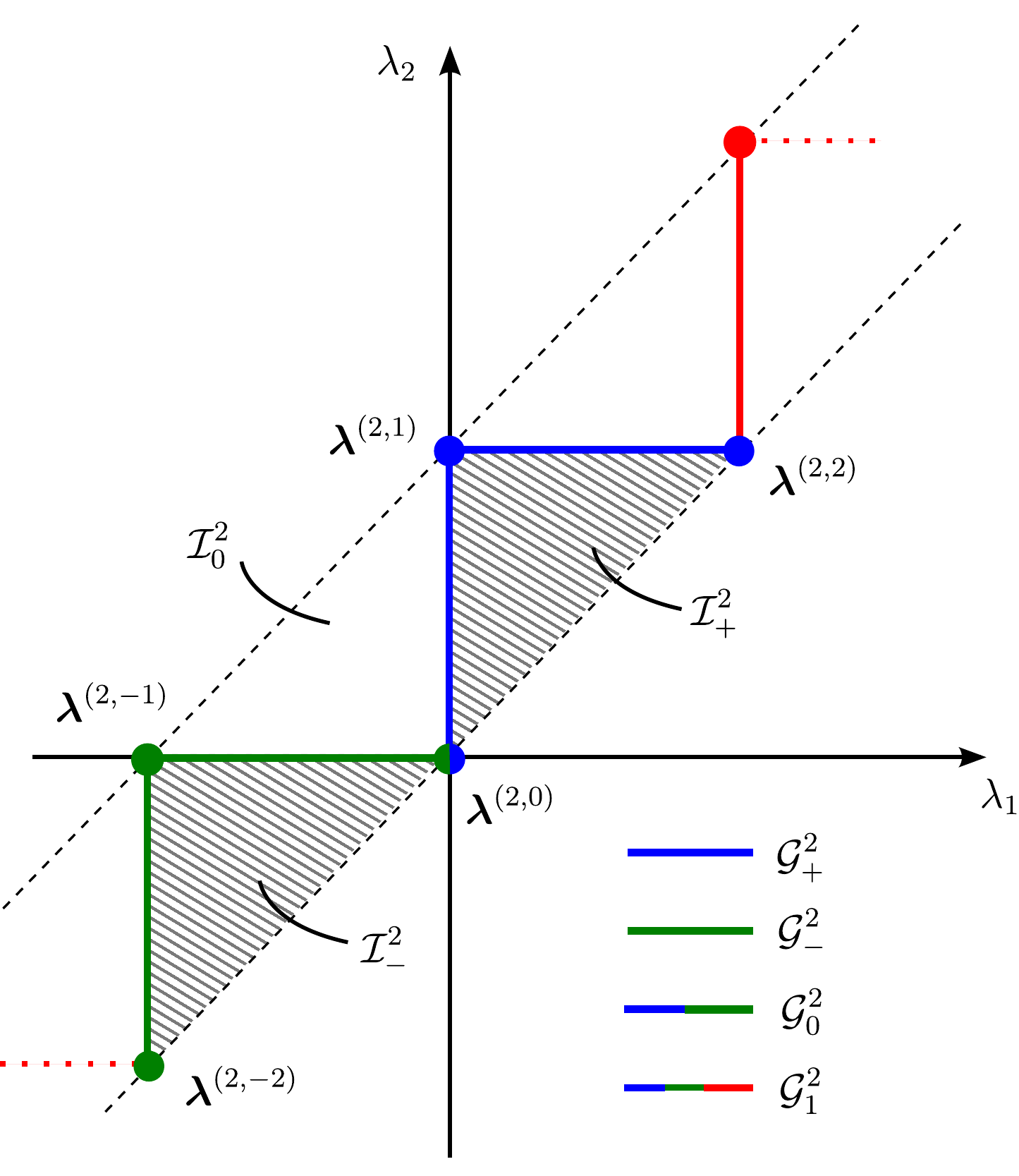}
                \caption{$k=2$}
                \label{fig:g21}
        \end{subfigure}%
        ~ 
        \begin{subfigure}[b]{0.48\textwidth}
                \centering
                \includegraphics[width=\textwidth]{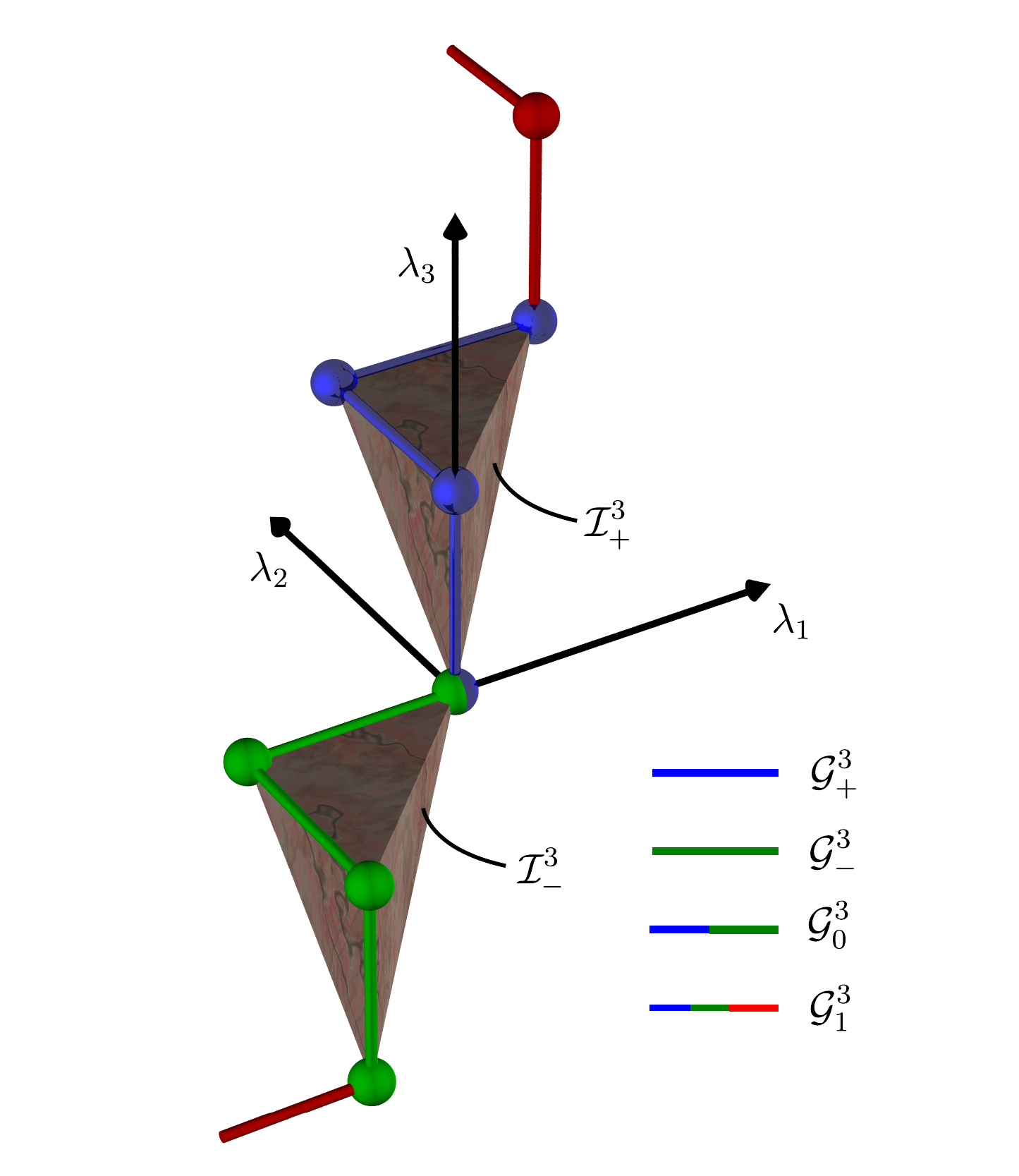}
                \caption{$k=3$}
                \label{fig:g31}
        \end{subfigure}
        \caption{The graph $\gcal^k_0$, composed of $\gcal^k_+$ and its translate $\gcal^k_-$, and the graph $\gcal^k_1$ which arises from $\gcal^k_0$ by periodic continuation.}\label{fig:gk1}
\end{figure}

On the level of analytic functions $F_k$, locality in a double cone implies a \emph{meromorphic} (but in general not analytic) extension to the entire multi-variable complex plane. The analyticity region is $\im \zeta_1 < \ldots < \im \zeta_k < \im \zeta_1 + 2\pi$, except for the kinematic poles mentioned above; outside this region, further singularities of $F_k$ will arise from poles of the scattering function $S$ outside the physical strip. Apart from the pole structure, reflected in the recursion relations, an $S$-periodicity condition as in \eqref{eq:introperio} is needed.

\begin{definition}\label{def:conditionFD}
A collection $F=(F_{k})_{k=0}^\infty$ of functions $\cbb^k \to \bar\cbb$ fulfills conditions (FD) if the following holds for any fixed $k$, and with $\zetav \in \cbb^k$ arbitrary:

\begin{enumerate}
\renewcommand{\theenumi}{(FD\arabic{enumi})}
\renewcommand{\labelenumi}{\theenumi}

\item \label{it:fdmero}
\emph{Analyticity:}
$F_k$ is meromorphic on $\cbb^k$, and analytic where $\im \zeta_1 < \ldots < \im \zeta_k < \im \zeta_1 + \pi$.

\item \label{it:fdsymm} \emph{$S$-symmetry:} For any $\sigma \in \perms{k}$, we have
$
\displaystyle{
F_k(\zetav)
=  S^\sigma(\zetav) F_k(\zetav^\sigma) .
}
$

\item \label{it:fdperiod} \emph{$S$-periodicity:}
$\displaystyle{
F_k (\zetav + 2i\pi \ev^{(j)} ) =
\Big(\prod_{\substack{i=1 \\ i \neq j}}^k S(\zeta_i-\zeta_j)\Big)  F_k (\zetav ).
}$

\item \label{it:fdrecursion}

\emph{Recursion relations:}
The $F_k$ have first order poles at $\zeta_n-\zeta_m = i \pi$, where $1 \leq m < n \leq k$,
and one has with $\hat\zetav = (\zeta_1,\ldots,\widehat{\zeta_m},\ldots, \widehat{\zeta_n}, \ldots, \zeta_k)$,
\begin{equation*}
\res_{\zeta_n-\zeta_m = i \pi} F_{k}(\boldsymbol{\zeta})
= - \frac{1}{2\pi i }
\Big(\prod_{j=m}^{n} S(\zeta_j-\zeta_m) \Big)
\Big(1-\prod_{p=1}^{k} S(\zeta_m-\zeta_p) \Big)
F_{k-2}( \boldsymbol{\hat\zeta} ).
\end{equation*}

\item \label{it:fdboundsreal}
\emph{Bounds at nodes:}
For each $j \in \{0,\ldots,k\}$ and $\ell \in \{-1,0\}$, we have
\begin{equation*}
\onorm{ F_k\big( \cdotarg + i \lambdav^{(k,j+k\ell)} + i \zerov \big) }{(k-j) \times j} < \infty.
\end{equation*}
Here $+i\zerov$ denotes approach from inside the region of analyticity as in \ref{it:fdmero}.

\item \label{it:fdboundsimag}
\emph{Pointwise bounds:}
There exist $c,c'>0$ such that for all $\zetav\in\tube(\ical^k_\pm)$:
\begin{equation*}
  |F_k(\zetav)| \leq c \,{ \operatorname{dist}(\im \zetav,\partial \ical_\pm^k)^{-k/2}} \prod_{j=1}^k \exp \big(\mu r  |\im \sinh \zeta_j|+ c' \omega(\cosh \re \zeta_j)\big).
\end{equation*}
\end{enumerate}
\end{definition}

Analogous to \eqref{eq:ikplus}, we have denoted $\ical^k_-:=\ich \gcal^k_-$.
Note that \ref{it:fdsymm}, \ref{it:fdperiod} and \ref{it:fdrecursion} are strengthened versions of \eqref{eq:ssymmdc}, \eqref{eq:introperio} and \eqref{eq:residue1k}, respectively, but that these stronger conditions may always be obtained from the weaker ones by using \eqref{eq:ssymmdc} repeatedly.


Equivalence of the three conditions is formulated as follows, very similar to Thm.~\ref{theorem:wedgeequiv} in the wedge local case.

\begin{theorem} \label{theorem:doubleconeequiv}
 Let $r>0$ and an analytic indicatrix $\omega$ be fixed.
\begin{enumerate}
\renewcommand{\theenumi}{(\roman{enumi})}
\renewcommand{\labelenumi}{\theenumi}
 \item \label{it:doubleconeequiv-at}
 If $A \in \qf^\omega$ fulfills (AD), then there are distributions $T_k$ fulfilling (TD) such that
\begin{equation} \label{eq:tboundary-d}
\cme{m,n}{A}(\thetav,\etav) = T_{m+n}(\thetav,\etav+i\piv),
\quad
\cme{m,n}{J A^\ast J}(\thetav,\etav) = T_{m+n}(\thetav-i \piv,\etav).
\end{equation}

 \item \label{it:doubleconeequiv-tf}
If $T_k$ fulfill (TD), then there are functions $F_k$ fulfilling (FD) such that for $-k \leq j \leq k$,
\begin{equation}\label{eq:ftboundary-d}
T_{k}\big(\thetav + i \lambdav^{(k,j)}\big)
= F_{k}\big(\thetav + i \lambdav^{(k,j)} + i\zerov \big).
\end{equation}

 \item \label{it:doubleconeequiv-fa} 
  If $F_k$ fulfill (FD), there is a quadratic form $A$ fulfilling (AD) such that
\begin{equation}\label{eq:fmnboundary-d}
 \cme{m,n}{A}(\thetav,\etav) = F_k(\thetav+i\zerov,\etav+i\piv-i\zerov), 
\quad
 \cme{m,n}{J A\st J}(\thetav,\etav) = F_k(\thetav-i \piv +i\zerov,\etav-i\zerov).
\end{equation}

\end{enumerate}
\end{theorem}

Again, the notation $\pm i\zerov$ denotes approach from inside the analyticity region as appropriate.

Comparing with the wedge-local variant in Thm.~\ref{theorem:wedgeequiv}, the most apparent change is that Thm.~\ref{theorem:doubleconeequiv} involves both $A$ and $J A\st J$. In fact, this will be the main idea of the proof: Locality in a double cone consists of two pieces of information, namely, that both $A$ and $JA\st J$ are localized in the wedge $\wcal_r'$. Using Thm.~\ref{theorem:wedgeequiv} for both of these, and putting these two pieces together, we will show the equivalence of the double cone locality conditions. As we shall see, the passage from $A$ to $J A\st J$, i.e., space-time reflection, corresponds to passing from $F_k$ to $F_k^\pi := F_k(\cdotarg + i \piv)$.

We remark at this point that the conditions are in fact invariant under space-time reflection in the following sense: If $A$ fulfills (AD), then so does $J A\st J$. If functions $F_k$ fulfill (FD), then $F_k^\pi$ fulfill (FD) as well. (This follows by using the periodicity condition \ref{it:fdperiod}, noting that the $S$ factors in the conditions depend only on differences of rapidities, and that $|\im \sinh(\zeta_j+i\pi)|=|\im \sinh(\zeta_j)|$.) On the level of the CR distributions $T_k$, a corresponding statement holds, but is more difficult to see directly; it will follow from our results.

We now proceed to the proof of Thm.~\ref{theorem:doubleconeequiv}, again handling each of the three parts in its own subsection. The passage (AD)$\Rightarrow$(TD)$\Rightarrow$(FD) in Secs.~\ref{sec:ad-to-td} and \ref{sec:td-to-fd} will involve an analytic continuation of the coefficients $\cme{m,n}{A}$ to larger and larger graphs and to their interior. Essential features in the geometry of these domains -- in particular, the kinematic poles -- become relevant only for $k\geq 3$, which makes them harder to understand intuitively. While we have sketched some of the regions in Fig.~\ref{fig:gk1} and \ref{fig:gk2}, the reader is invited to review the supplemental animation \videoref{} which gives a better geometric overview of the respective analyticity domains for $k=3$.

\subsection{(AD) \texorpdfstring{$\Rightarrow$}{=>} (TD)}\label{sec:ad-to-td}

For constructing the CR distributions $T_k$ from a quadratic from $A$ that is $\omega$-local in $\ocal_r$, the key technique is to apply Thm.~\ref{theorem:wedgeequiv}\ref{it:wedgeequiv-at} to both $A$ and $J A\st J$.

\begin{theopargself}
\begin{proof}[\ofwhat{Thm.~\ref{theorem:doubleconeequiv}\ref{it:doubleconeequiv-at}}]
Let $A$ fulfill (AD). Since $A$ in particular fulfills (AW), we can apply Thm.~\ref{theorem:wedgeequiv}\ref{it:wedgeequiv-at} which yields CR distributions $T_k$ on $\tube(\gcal_+^k)$, fulfilling (TW), with boundary values as in the first half of \eqref{eq:tboundary-d}.
Now $JA^{*}J$ fulfills (AD) and hence (AW) as well; therefore Thm.~\ref{theorem:wedgeequiv}\ref{it:wedgeequiv-at} yields another collection of CR distributions $T_k'$. We use this to define $T_k$ on $\tube(\gcal^k_-)$ by
\begin{equation}
   T_k(\zetav- i \piv) := T_k' (\zetav).
\end{equation}
This has the boundary distributions
\begin{equation}\label{eq:fkpotherbv}
 T_k(\thetav-i\piv,\etav) 
 = T_k'(\thetav,\etav+i\piv) 
 = \cme{m,n}{JA^{*}J}(\thetav,\etav),
\end{equation}
which shows the second half of \eqref{eq:tboundary-d}.
For $T_k$ being CR on $\tube(\gcal_0^k)$, it remains to show that the two boundary values of $T_k$ at the origin agree. By the above, we know that for real $\thetav$,
\begin{equation}
T_k(\thetav)\big\vert_{\gcal_+^k}=\cme{k,0}{A}(\thetav),
\quad
T_k(\thetav)\big\vert_{\gcal_-^k}=\cme{0,k}{J A\st J}(\thetav).
\end{equation}
However, a short computation from \eqref{eq:fmndef} shows that these are equal:
\begin{equation}
\cme{k,0}{A}(\thetav)
=\hscalar{z^{\dagger}(\theta_{1})\ldots z^{\dagger}(\theta_{k})\Omega}{A\Omega}
=\hscalar{\Omega}{J A\st J z^{\dagger}(\theta_{k})\ldots z^{\dagger}(\theta_{1})\Omega}
=\cme{0,k}{JA^{*}J}(\thetav).
\end{equation}
This proves \ref{it:tdmero}. Similarly, one shows that the boundary values $T_k(\thetav+i\lambdav^{(k,-k)})$ and $T_k(\thetav+i\lambdav^{(k,k)})$ agree, yielding the periodicity condition \ref{it:tdperiod}.

$S$-symmetry \ref{it:tdsymm} of the $T_k$ follows directly from \ref{it:twsymm}. 
For \ref{it:tdrecursion}, we use that $\cme{m,n}{J A\st J}$ can be expressed in terms of the $\cme{m,n}{A}$ as in \cite[Prop.~3.11]{BostelmannCadamuro:expansion}. Replacing the $\cme{m,n}{\cdotarg}$ with boundary values of $T_k$ as in \eqref{eq:tboundary-d}, that expression yields
\begin{equation}
T_{m+n}(\thetav -i\pmb{\pi},\etav)
=\sum_{C\in\ccal_{m,n}}(-1)^{|C|}\delta_{C}
 S_{C} R_C(\thetav,\etav) T_{m+n-2|C|}(\hat \etav, \hat\thetav + i \piv).\label{recursione5}
\end{equation}
This is exactly \ref{it:tdrecursion}.

The bounds \ref{it:tdboundsreal} and \ref{it:tdboundsimag} follow by combining the known bounds \ref{it:twboundsreal} and \ref{it:twboundsimag} for $T_k$ and $T_k'$, noting that a shift of arguments by $i\piv$ yields a minus sign in the exponent in \ref{it:twboundsimag}.
\cmpqed\end{proof}
\end{theopargself}

\subsection{(TD) \texorpdfstring{$\Rightarrow$}{=>} (FD)}\label{sec:td-to-fd}

We consider a collection of CR distributions $T_k$ on $\tube(\gcal^k_{0})$ fulfilling conditions (TD). In order to construct meromorphic functions $F_k$, we start by extending the $T_k$ to certain larger graphs, using the symmetry relations in conditions (TD).

Let us first consider the graph 
\begin{equation} \label{setg1}
 \gcal_1^k :=\gcal_{0}^{k}+2 \piv \mathbb{Z} = \{ \text{edges} : \lambda_1 \leq \ldots \leq \lambda_k \leq \lambda_1 +  \pi \},
\end{equation}
that is, $\gcal_1^k$ is $\gcal^k_0$ with all edges and nodes translated by integer multiples of $2\pi$ in all coordinates simultaneously; cf.~Eq.~\eqref{setg0} and Fig.~\ref{fig:gk1}.                                                                                                                                                                 
We continue $T_k$ to $\tube(\gcal_1^k)$ by defining for $n\in \mathbb{Z}$ and $\zetav\in \tube(\gcal_{0}^{k})$,
\begin{equation}\label{defonestair}
T_k(\zetav+2 i n \piv):=T_k(\zetav).
\end{equation}
This is indeed a CR distribution on the graph, since the boundary values $T_{k}(\thetav+i n \piv +i0\ev^{(k)})$ and $T_{k}(\thetav+in\piv-i0\ev^{(1)})$ agree for all $n \in 2\mathbb{Z}+1$ at real $\thetav$, due to \ref{it:tdperiod}.

\begin{figure}
        \centering
        \begin{subfigure}[b]{0.49\textwidth}
                \centering
                \includegraphics[width=\textwidth]{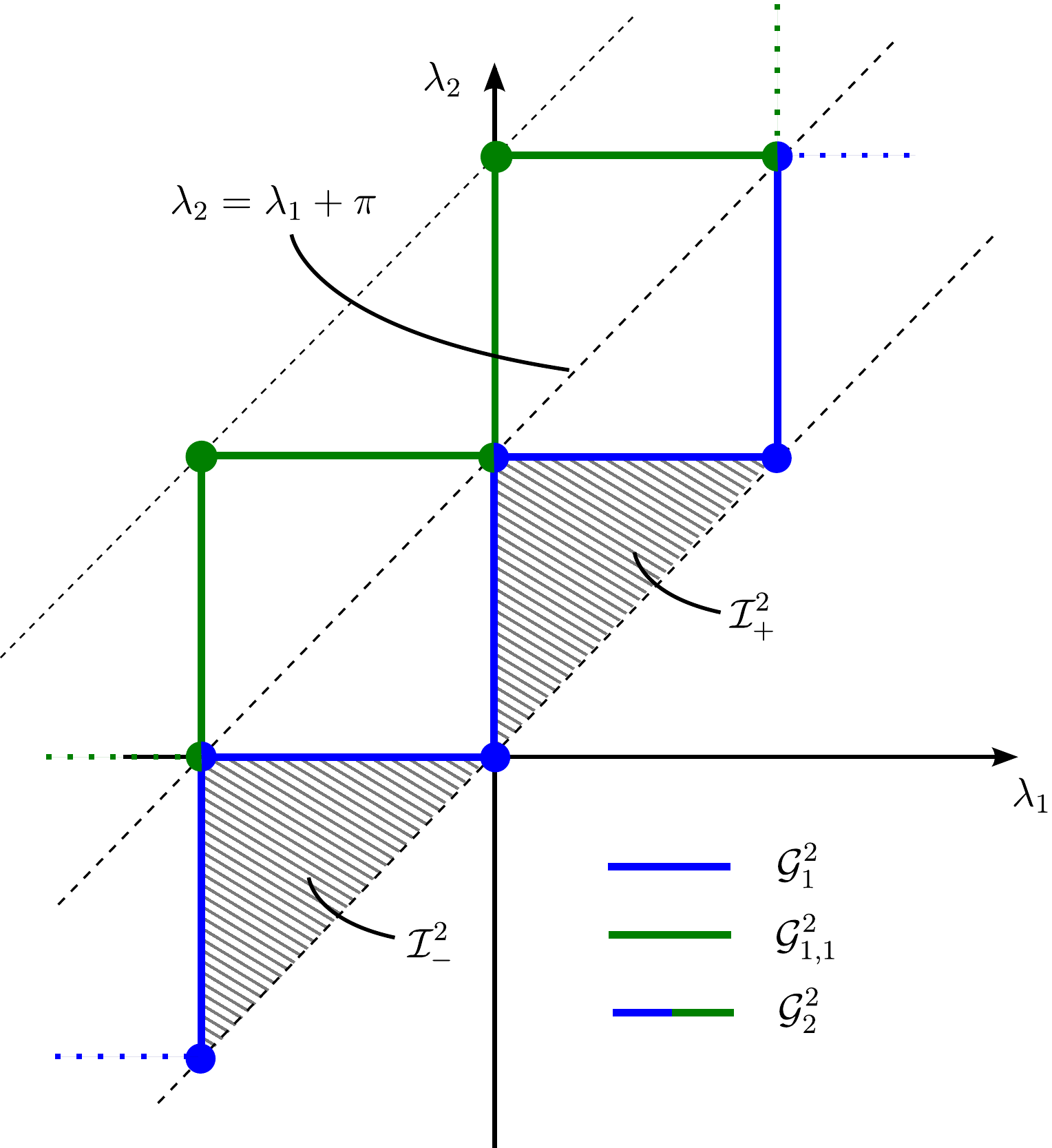}
                \caption{$k=2$}
                \label{fig:g22}
        \end{subfigure}%
        ~ 
        \begin{subfigure}[b]{0.49\textwidth}
                \centering
                \includegraphics[width=\textwidth]{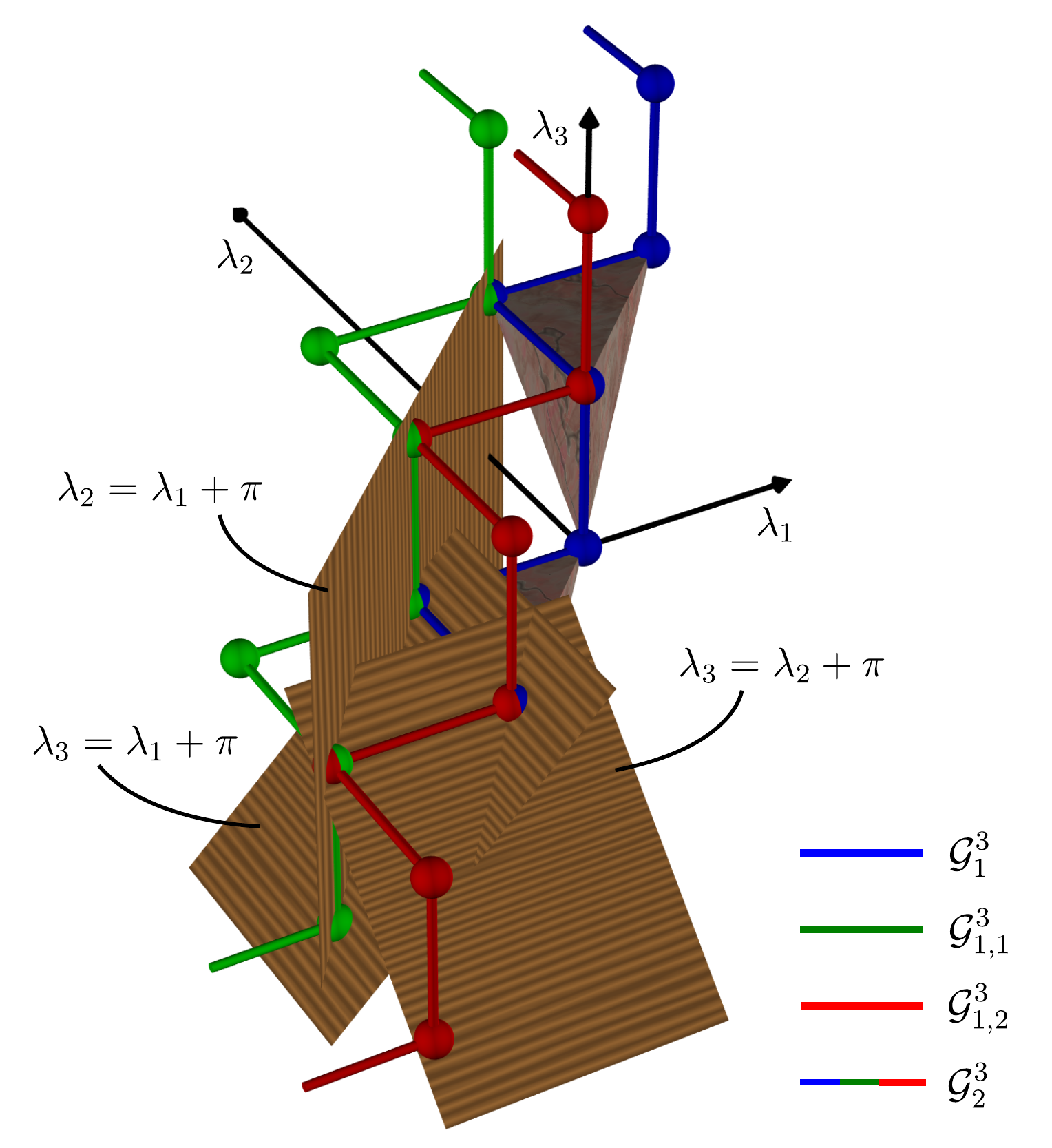}
                \caption{$k=3$}
                \label{fig:g32}
        \end{subfigure}
        \caption{The graph $\gcal^k_2$, composed of $\gcal^k_1$ and its translates $\gcal^k_{1,m}$.%
         The hyperplanes of the kinematic poles, $\lambda_n=\lambda_m+\pi$, are (partially) shown. Note that in the case $k=2$, the function $F_2$ does not actually have a pole on this hyperplane.}
        \label{fig:gk2}
\end{figure}

Next, we consider for $0 \leq m < k$ the graph (cf.~Fig.~\ref{fig:gk2})
\begin{equation}
\gcal^{k}_{1,m}:=\gcal^{k}_{1}+\lambdav^{(k,-m)}=\gcal^{k}_{1}+(\underbrace{-\pi,\ldots,-\pi}_{m},0,\ldots,0).
\end{equation}
One has $\gcal^k_{1,0}=\gcal^k_{1}$.
We define $T_k$ on $\tube(\gcal^{k}_{1,m})$ by
\begin{equation}\label{stairs}
T_k(\zetav):=T_{k}(\zeta_{m+1},\ldots,\zeta_{k},\zeta_{1}+2i\pi,\ldots,\zeta_{m}+2i\pi),
\end{equation}
noting that for $\zetav\in \tube(\gcal^{k}_{1,m})$, the argument on the r.h.s.~is in $\tube(\gcal^{k}_{1})$. Combining this for all $m$, we obtain $T_k$ as a distribution on the tube over the graph (again see Fig.~\ref{fig:gk2})
\begin{equation}\label{eq:g2def}
 \gcal^{k}_2 := \bigcup_{0 \leq m \leq k-1} \gcal^{k}_{1,m} =
\{  \text{edges} : \lambda_1 \leq \ldots \leq \lambda_k \leq \lambda_1 + 2 \pi \}.
\end{equation}
It is important here to note the following: While $T_k$ evidently are CR distributions on all $\tube(\gcal^{k}_{1,m})$, they are \emph{not} necessarily CR distributions on $\tube(\gcal^k_2)$. Namely, the two graphs $\gcal^{k}_{1,m}$ and $\gcal^{k}_{1,m'}$ ($m > m')$ have some nodes in common,
given by
\begin{equation}\label{eq:stairsmeet}
\underbrace{(\overbrace{-\pi \ldots -\pi}^{m'}, \overbrace{0 \ldots 0}^{m-m'}, \overbrace{\pi \ldots \pi}^{k-m})}_{=:\lambdav^{m \cap m'}}+\ell\piv , \quad \ell\in \mathbb{Z}.
\end{equation}
At these common nodes, the boundary values of $T_k$ from different edges need not agree. Indeed, let us compute the difference of the boundary values at the point $\zetav=\thetav+i\lambdav^{m \cap m'}$, setting $\ell=0$. On $\tube(\gcal_{1,m}^{k})$, we have
\begin{equation}\label{eq:fgm}
\begin{aligned}
&T_k(\zetav)\big\vert_{\gcal^{k}_{1,m}}
\\&=T_k(\theta_{m+1}+i\pi,\ldots,\theta_{k}+i\pi,\theta_{1}+i\pi,\ldots,\theta_{m'}+i\pi,\theta_{m'+1}+2i\pi,\ldots,\theta_{m}+2i\pi)
\\&=T_k(\theta_{m+1}-i\pi,\ldots,\theta_{k}-i\pi,\theta_{1}-i\pi,\ldots,\theta_{m'}-i\pi,\theta_{m'+1},\ldots,\theta_{m}),
\end{aligned}
\end{equation}
where we made us of \eqref{stairs} and \eqref{defonestair}. 
Analogously, we find
\begin{equation}\label{eq:fgmp}
T_k(\zetav)\big\vert_{\gcal^{k}_{1,m'}}
=T_k(\theta_{m'+1},\ldots,\theta_{m},\theta_{m+1}+i\pi,\ldots,\theta_{k}+i\pi,\theta_{1}+i\pi,\ldots,\theta_{m'}+i\pi).
\end{equation}
Note that on the r.h.s.~of \eqref{eq:fgm} and \eqref{eq:fgmp}, the distribution $T_{k}$ is evaluated at two different nodes of $\gcal_{0}^{k}$.
The difference of the boundary values \eqref{eq:fgm} and \eqref{eq:fgmp} can now be computed using condition \ref{it:tdrecursion}; it is in general nonzero and quite intricate to describe. 

Nevertheless, we can use the above results in order to construct a continuation of $T_k$ as meromorphic functions on the tube over the open set
\begin{equation}\label{eq:ich2}
\ical_2^k := \ich \gcal_2^k =\{ \lambdav\in\rbb^k : \lambda_{1}<\ldots < \lambda_{k}< \lambda_{1}+2\pi \}.
\end{equation}
This is the content of the next proposition.

\begin{proposition} \label{proposition:extendinterior}
 Let $T_k$ be distributions fulfilling (TD). Then there exist meromorphic functions $F_k$ on $\tube(\ical_2^k)$ which have the boundary values \eqref{eq:ftboundary-d}. They are analytic except for possible first-order poles at $\zeta_n-\zeta_m=i\pi$, $m<n$.
\end{proposition}

\begin{proof}
Using the extended $T_k$ constructed above, we define distributions $G_k$ on $\mathcal{T}(\gcal_{2}^{k})$ by
\begin{equation}\label{gdef}
G_k(\zetav) := T_k(\zetav)  \prod_{j > j'} \big( \zeta_{j}-\zeta_{j'} - i \pi \big).
\end{equation}
We claim that these are CR distributions on the graph. For that, it remains to show that their boundary values agree at the nodes $\im \zetav=\lambdav^{m \cap m'}+\ell\piv$, see Eq.~\eqref{eq:stairsmeet}. We treat only the case $\ell=0$; the case $\ell=1$ can then be treated similarly,
and for general $\ell$ the result can be obtained by periodicity. 

For $\ell=0$, we can compute the boundary values in question from \eqref{eq:fgm}, \eqref{eq:fgmp} and \ref{it:tdrecursion}. In fact, this computation simplifies greatly since the polynomial factor in \eqref{gdef} vanishes on the support of $\delta_{C}$ on the right hand side of \ref{it:tdrecursion}, except for the term corresponding to the contraction $C=(m,n,\emptyset)$. This leads us to
\begin{equation}
G_{k}(\zetav)\big\vert_{\gcal^{k}_{1,m}}=G_{k}(\zetav)\big\vert_{\gcal^{k}_{1,m'}}.
\end{equation}
Hence, the $G_k$ are CR distributions on $\tube(\gcal_2^k)$. 

We can now apply Lemma \ref{lem:graphtube} to the connected graph $\gcal_2^k$, which yields an extension of $G_k$ to an analytic function on $\tube(\ich \gcal_{2}^{k})$, with distributional boundary values on $\tube(\ach \gcal_{2}^{k})$. We then define $F_k$ as
\begin{equation}
F_k(\zetav) := G_k(\zetav) \cdot \prod_{j > j'} \big(\zeta_{j}-\zeta_{j'}-i\pi\big)^{-1},
\end{equation}
which is evidently analytic on the same domain, except for possible poles at $\zeta_{n}-\zeta_{m}=i\pi$, $m<n$. Taking the boundary limit to $\tube(\gcal^k_0)$ from within the convex hull of $\gcal^k_0$, the boundary distribution coincides with $T_k$ by construction, i.e., \eqref{eq:ftboundary-d} holds.
\cmpqed\end{proof}

As meromorphic functions, we can extend $F_{k}$ even further, using $S$-symmetry of the $T_k$.

\begin{proposition}\label{proposition:extendpermuted}
  The functions $F_k$ of Prop.~\ref{proposition:extendinterior} extend meromorphically to $\tube(\ical_3^k)$, where
\begin{equation}\label{eq:i3k}
\ical_3^k :=\{ \lambdav \in \rbb^k: |\lambda_{j}-\lambda_{j'}|< 2\pi \text{ for all $j,j'$}\}.
\end{equation}
They fulfill the $S$-symmetry condition \ref{it:fdsymm}.
\end{proposition}

\begin{proof}
For any fixed permutation $\sigma \in \perms{k}$, we consider the region
\begin{equation}
  \ical_{2,\sigma}^k :=\{ \lambdav  \in \rbb^k: \lambda_{\sigma(1)}<\ldots<\lambda_{\sigma(k)}<\lambda_{\sigma(1)}+2\pi \}.
\end{equation}
We define $F_k$ on $\tube(\ical_{2,\sigma}^k)$ by
\begin{equation}\label{permf}
F_{k}(\zetav):=F_{k}(\zeta_{\sigma(1)},\ldots,\zeta_{\sigma(k)}) \,S^{\sigma}(\zetav)
\end{equation}
with $S^\sigma$ as in Eq.~\eqref{eq:Sperm}. Since $S$, and hence $S^\sigma$, is meromorphic for all arguments, this gives $F_k$ as a meromorphic function on the disjoint regions $\tube(\ical_{2,\sigma}^k)$. Since $S$ has no poles on the real line, we can in fact find a complex neighborhood $\ncal$ of $\rbb^k$ (not necessarily tubular) where all $S^\sigma$ are analytic; hence $F_k$ is analytic in $\ncal \cap \tube(\ical_{2,\sigma}^k)$ for all $\sigma$. Due to \ref{it:tdsymm}, the boundary distributions at $\rbb^k$ from within all these domains agree. An application of the edge-of-the-wedge theorem (e.g., in the form of \cite{Eps:edge_of_wedge}) around each real point yields an extension of $F_k$  to a possibly smaller complex neighborhood $\ncal'\subset \ncal$ of $\rbb^k$. That is, $F_k$ is meromorphic on the connected domain
\begin{equation}
   \rcal := \ncal' \cup \bigcup_{\sigma \in \perms{k}} \tube(\ical_{2,\sigma}^k).
\end{equation}
It follows from the tubular edge-of-the-wedge theorem \cite{Bros:1977} that the envelope of holomorphy of $\rcal$ is $\operatorname{conv}( \rcal )$. But this agrees with its envelope of meromorphy \cite[Theorem~3.6.6]{JarnickiPflug:2000}. Hence $F_k$ continues meromorphically to $\operatorname{conv}( \rcal) = \tube(\ical_{3}^k)$.
\cmpqed\end{proof}

Periodicity and $S$-symmetry of $T_k$ finally allow us to extend $F_k$ to the entire multi-variable complex plane.

\begin{proposition}\label{proposition:extendeverywhere}
  The functions $F_k$ of Prop.~\ref{proposition:extendinterior} extend meromorphically to $\cbb^k$. They fulfill \ref{it:fdmero}, \ref{it:fdsymm} and \ref{it:fdperiod}.
\end{proposition}

\begin{proof}
We define $F_{k}$ on $\cbb^k$ by
\begin{equation}\label{fperiodica}
F_{k}(\zetav) :=
\Bigg( \prod_{\ell=1}^{k}\Big( \prod_{j \neq \ell} S(\zeta_\ell-\zeta_j)
   \Big)^{ n_\ell } \Bigg)
F_{k}(\zetav + 2 i \pi \nv),
\end{equation}
where $\nv\in\zbb^k$ is chosen such that $\zetav + 2 i \pi\nv \in \tube(\ical_3^k)$.
We need to show that this is well-defined: It is certainly possible to choose such $\nv$ for any $\zetav$, but there might be several such choices. 
Suppose that, for fixed $\zetav$, there exist $\nv\neq\nv'$ such that $\im\zetav+2\pi\nv\in \ical_3^k$ and $\im\zetav+2\pi\nv'\in \ical_3^k$.
We need to show that
\begin{equation}\label{periodrel}
\underbrace{\prod_{\ell=1}^{k}\Big( \prod_{j \neq \ell} S(\zeta_\ell-\zeta_j)
   \Big)^{ n_\ell }}_{=:S_{\nv}(\zetav)}
F_{k}(\zetav +2i\pi \nv)
=\underbrace{\prod_{\ell=1}^{k}\Big( \prod_{j \neq \ell} S(\zeta_\ell-\zeta_j)
   \Big)^{ n'_\ell }}_{=S_{\nv'}(\zetav)}
F_{k}(\zetav+2i\pi \nv').
\end{equation}
Dividing by $S_{\nv'}(\zetav)$, and using $2 i\pi$-periodicity of the $S$-factors, we can assume without loss of generality that $\nv'=0$ and $\im\zetav \in \ical_3^k$.
Further, one checks that the factor $S_{\nv}(\zetav)$ defined above fulfills $S_{\nv}(\zetav)=S_{\nv^\rho}(\zetav^\rho)$ for any permutation $\rho$. Taking into account that $F_k$ is known to be $S$-symmetric by Prop.~\ref{proposition:extendpermuted}, we see that the relation \eqref{periodrel} is invariant under permuting the components of $\zetav,\nv$; hence we can assume that $n_1 \leq \ldots \leq n_k$.

Now, with $\lambdav:=\im\zetav$, the conditions $\lambdav\in \ical_3^k$ and $\lambdav+2 \pi \nv \in \ical_3^k$ imply, cf.~\eqref{eq:i3k},
\begin{equation}
\forall j,k: \quad |\lambda_{j}-\lambda_{k}| < 2\pi, \quad |\lambda_{j}-\lambda_{k}+2\pi (n_{j}-n_{k})| < 2\pi. 
\end{equation}
A short computation shows that $n_j\in\{N,N+1\}$ for all $j$ with some fixed $N \in \zbb$. 
In the following, we treat only the case $N=0$; the case $N=-1$ can be handled with similar arguments, and for all other $N$ we employ $2i\piv$-periodicity of $F_k$.

\sloppy
For showing the identity \eqref{periodrel} between meromorphic functions --  where now $\nv=(0,\ldots,0,1,\ldots,1)$ with $m$ entries of $0$, and $\nv'=0$ --, it suffices to check it on a real open set, possibly on the boundary of the domain. Therefore, we can choose $\im \zetav=0$ and $\im\zetav+2\pi \nv\in \bar\ical_3^k$.
Inserting $T_k$ as the boundary value of $F_k$, and using \eqref{eq:Sperm}, it remains to show that for real $\thetav$ and in the sense of distributions,
\begin{equation}\label{neq01form}
T_k(\thetav)=\Big(\prod_{\ell>m}\prod_{j \leq m}S(\theta_{\ell}-\theta_{j})\Big)
T_k(\theta_1,\ldots,\theta_m,\theta_{m+1} + 2 \pi i, \ldots, \theta_k+2\pi i).
\end{equation}
On the right hand side, $T_k$ is evaluated on a point of $\tube(\gcal^{k}_{1,m})$. Using Eq.~\eqref{stairs}, and then $2 i \piv$-periodicity of $T_k$ \eqref{defonestair}, we find
\begin{equation}\label{stairform}
\rhs{neq01form} = \Big(\prod_{\ell>m}\prod_{j \leq m}S(\theta_{\ell}-\theta_{j})\Big) T_k(\theta_{m+1},\ldots,\theta_{k},\theta_{1},\ldots,\theta_{m}) = S^\sigma(\thetav) T_k(\thetav^\sigma)
\end{equation}
with a certain permutation $\sigma$.
Since $T_k$ is $S$-symmetric by \ref{it:tdsymm}, this proves \eqref{neq01form}.

Now that $F_k$ is known to be well-defined on $\cbb^k$, it is clear that it is meromorphic everywhere and analytic on $\tube(\ich \gcal_1^k)$, hence \ref{it:fdmero} is fulfilled. \ref{it:fdsymm} was shown in Prop.~\ref{proposition:extendpermuted}, and \ref{it:fdperiod} is a special case of \eqref{fperiodica}, which extends to all complex arguments.
\cmpqed\end{proof}

Now, we want to compute the residues of $F_{k}$ in order to derive the recursion relations \ref{it:fdrecursion}. They arise as a consequence of the corresponding condition \ref{it:tdrecursion}.

\begin{proposition}\label{proposition:residues}
  The first-order poles of $F_k$ at $\zeta_n-\zeta_m=i\pi$, $m<n$, have residues as given by \ref{it:fdrecursion}.
\end{proposition}

\begin{proof}
 It suffices to prove \ref{it:fdrecursion} for $m=1$, $n=k$; the general case then follows by $S$-symmetry. Since the residues are meromorphic functions on the pole hypersurfaces, it suffices to verify \ref{it:fdrecursion} on a real open set. We therefore compare the boundary values of $F_k$ at the points $\zetav_\pm = \thetav + i (0,\ldots, 0, \pi \pm 0)$, where we can assume $\theta_j \neq \theta_{j'}$ for $j \neq j'$ (unless $j=1,j'=k$).  We note that $\im\zetav_-\in\ich \gcal_1^k $ but $\im\zetav_+\in\ich \gcal_{1,k-1}^k $. Using \eqref{stairs} and the boundary values of $F_k$ as in \eqref{eq:ftboundary-d}, we have
\begin{equation}
\begin{aligned}
   F_{k}(\theta_{1}, \ldots, &\theta_{k-1},\theta_{k}+i\pi -i0)- F_{k}(\theta_{1}, \ldots, \theta_{k-1},\theta_{k}+i\pi+i0)\\
&=T_{k}(\theta_{1}, \ldots, \theta_{k-1},\theta_{k}+i\pi)-T_{k}(\theta_{k}-i\pi,\theta_{1}, \ldots, \theta_{k-1})\\
&=\delta(\theta_k-\theta_1)\Big( 1- \prod_{p=1}^{k}S(\theta_p-\theta_k)\Big)F_{k-2}(\theta_{2}, \ldots, \theta_{k-1}),
\end{aligned}
\end{equation}
where in the second equality we made use of \ref{it:tdrecursion} in the case $m=1$. Referring to Lemma~\ref{lemma:onepole}, we can read off the residue of the pole:
\begin{equation}
\res_{\zeta_{k}-\zeta_{1}=i\pi}F_{k}(\zetav)=\frac{1}{2 \pi i}\Big( 1-\prod_{p=1}^{k}S(\zeta_{p}-\zeta_{1})\Big)F_{k-2}(\hat{\zetav}).
\end{equation}
This is exactly \ref{it:fdrecursion} in the case $m=1,n=k$.
\cmpqed\end{proof}

The only remaining properties to be discussed are the bounds \ref{it:fdboundsreal} and \ref{it:fdboundsimag}, which are easy to obtain from results of Sec.~\ref{sec:wedges}. We summarize:

\begin{theopargself}
\begin{proof}[\ofwhat{Theorem~\ref{theorem:doubleconeequiv}\ref{it:doubleconeequiv-tf}}]
 Let $T_k$ fulfill (TD). We saw in Prop.~\ref{proposition:extendeverywhere} and \ref{proposition:residues} that these distributions have meromorphic extensions $F_k$ which fulfill \ref{it:fdmero}--\ref{it:fdrecursion}. They have the proposed boundary values at nodes (Prop.~\ref{proposition:extendinterior}). The bounds on nodes \ref{it:fdboundsreal} are a direct consequence of \ref{it:tdboundsreal}. Bounds in the interior \ref{it:fdboundsimag} can be obtained by applying Thm.~\ref{theorem:wedgeequiv}\ref{it:wedgeequiv-tf} twice, namely to $T_k$ and $T_k(\cdotarg - i \piv)$, which both fulfill conditions (TW).
\cmpqed\end{proof}
\end{theopargself}

\subsection{(FD) \texorpdfstring{$\Rightarrow$}{=>} (AD)}\label{sec:fd-to-ad}

We now set out from meromorphic functions $F_k$ fulfilling (FD), and construct an associated local observable $A$. Before proceeding to the definition of $A$, let us first compute the higher-order residues of $F_k$ as a consequence of \ref{it:fdrecursion}. We remind the reader of the notion of contractions, see Eqs.~\eqref{eq:lcvector}--\eqref{eq:rc}.

\begin{lemma}\label{lemma:fres}
There holds
\begin{equation}\label{resgenrec}
\res_{\eta_{r_{1}}-\theta_{\ell_{1}}=0} \!\!\ldots\!\! \res_{\eta_{r_{|C|}}-\theta_{\ell_{|C|}}=0}F_{m+n}(\thetav,\etav +i\piv)=\\
\frac{(-1)^{|C|}}{(2i\pi)^{{|C|}}}S_{C} R_{C}(\thetav,\etav) F_{m+n-2{|C|}}(\hat{\thetav},\hat{\etav}+i\piv),
\end{equation}
where $C$ is the contraction $(m,n,\{(\ell_1,r_1+m),\ldots,(\ell_{|C|},r_{|C|}+m)\})$.
\end{lemma}
\begin{proof}
Our proof uses induction on $|C|$.
We first note that \ref{it:fdrecursion} in our specific situation simplifies to
\begin{equation}\label{eq:oneres}
\res_{\eta_{r}-\theta_{\ell}=0}F_{m+n}(\thetav,\etav+i\piv)
=-\frac{1}{2\pi i }S_{C_{1}} R_{C_{1}} (\thetav,\etav)
F_{m+n-2}( \boldsymbol{\hat\theta},\boldsymbol{\hat\eta} +i\piv),
\end{equation}
where $C_1=(m,n,\{(\ell,r+m)\})$.
This is just Eq.~\eqref{resgenrec} in the case $|C|=1$.

Now assume that Eq.~\eqref{resgenrec} holds for $|C|-1$ in place of $|C|$.
We split $C$ into two contractions, namely into $C'=(m,n,\{ (\ell_{2},r_{2}+m),\ldots, (\ell_{|C|},r_{|C|}+m)  \})$ and $C_1 \in \ccal_{m-|C'|,n-|C'|}$, $|C_1|=1$. (Cf.~\cite[Sec.~3]{BostelmannCadamuro:expansion} for details; in notation used there, we have $C=C'\dot\cup C_{1}$.) Employing the induction hypothesis, we obtain
\begin{equation}
\begin{aligned}
\res_{\eta_{r_{1}}-\theta_{\ell_{1}}=0}&\Big(\res_{\eta_{r_{2}}-\theta_{\ell_{2}}=0}\ldots\res_{\eta_{r_{|C|}}-\theta_{\ell_{|C|}}=0} F(\thetav ,\etav +i\piv) \Big)\\
&=\frac{(-1)^{|C'|}}{(2i\pi)^{|C'|}}R_{C'} S_{C'} (\thetav,\etav)
\res_{\eta_{r_{1}}-\theta_{\ell_{1}}=0} F_{m+n-2|C|+2}(\hat{\thetav},\hat{\etav}+i\piv)\\
&=\frac{(-1)^{|C|}}{(2i\pi)^{|C|}} R_{C'} S_{C'} (\thetav,\etav)
 R_{C_{1}} S_{C_{1}} (\hat\thetav,\hat\etav) F_{m+n-2|C|}( \Hat{\Hat{\thetav}},\Hat{\Hat{\etav}} +i\piv),
\end{aligned}
\end{equation}
where Eq.~\eqref{eq:oneres} was used. Note that the argument $(\thetav, \etav)$ on the right-hand side needs to be read on the support of $\delta_C$. By \cite[Lemma~3.2]{BostelmannCadamuro:expansion}, the factor $R_{C'} S_{C'} (\thetav,\etav)
 R_{C_{1}} S_{C_{1}} (\hat\thetav,\hat\etav)$ can then be replaced with $R_{C} S_{C} (\thetav,\etav)$, which gives the desired result.
\cmpqed\end{proof}

Noting that the $F_k$ fulfill in particular (FW), we now take $A\in\qf^\omega$ to be the quadratic form constructed in Thm.~\ref{theorem:wedgeequiv}\ref{it:wedgeequiv-fa}. Its expansion coefficients are 
\begin{equation}\label{eq:fmnbydef}
 \cme{m,n}{A}(\thetav,\etav)=F_{m+n}(\thetav+i\zerov,\etav+i\piv-i\zerov),
\end{equation}
i.e., the first half of \eqref{eq:fmnboundary-d} is fulfilled. The crucial point is now to establish the second half, or in other words, the correspondence between the shifted function $F_k^\pi = F_k(\cdotarg - i \piv)$ and the reflected operator $J A\st J$.

\begin{proposition}\label{proposition:fshifted}
If the $F_{k}$ fulfill condition (FD), then the quadratic form $A$ above  fulfills
\begin{equation}\label{eq:Fshifted}
\cme{m,n}{J A^{*} J}(\thetav,\etav) = F^{\pi}_{m+n}(\thetav+i\zerov,\etav+i\piv-i\zerov).
\end{equation}
\end{proposition}

\begin{proof}
The right hand side of \eqref{eq:Fshifted} can be rewritten as
\begin{multline}\label{eq:fpi}
F^{\pi}_{m+n}(\thetav+i\zerov,\etav+i\piv-i\zerov)
=F_{m+n}(\thetav-i\piv+i\zerov,\etav-i\zerov)\\
=\Big(\prod_{p=1}^{m}\prod_{q=1}^{n}S(\theta_{p}-\eta_{q}) \Big)
F_{m+n}(\thetav+i\piv+i\zerov,\etav-i\zerov)
=F_{m+n}(\etav-i\zerov,\thetav+i\piv+i\zerov),
\end{multline}
where we used \ref{it:fdperiod} and \ref{it:fdsymm} in the second and third equality, respectively.

To the last expression, we apply Prop.~\ref{proposition:multivarres} with the substitution $\zv = (\etav,\thetav+i\piv)$,
with the indices $p$ there labeling pairs $(\ell,r)$, $1 \leq \ell \leq n$, $n+1 \leq r \leq n+m$, with contractions $C \in \ccal_{n,m}$ in place of $M \subset \{1,\ldots,p\}$, and with the following vectors in $\rbb^{n+m}$,
\begin{align} \notag
\av^{(\ell,r)} &:= (0,\ldots,0,\underset{\substack{\uparrow \\ \ell}}{1},0,\ldots,0,\underset{\substack{\uparrow \\ r}}{-1},0,\ldots,0),\\
\notag
\bv^C &:= (1,\ldots,\underset{\substack{\uparrow \\ \ell_j}}{0},\ldots,\underset{\substack{\uparrow \\ n}}{1},\underset{\substack{\uparrow \\ n+1}}{-1},\ldots,\underset{\substack{\uparrow \\ r_j}}{0},\ldots,-1),
  \quad \text{where } C=(n,m,\{(\ell_j,r_j)\}), \\
\cv &:= (\underbrace{-1,\ldots,-1}_{n},\underbrace{1,\ldots,1}_m).
\end{align}
One notes $\av^{(\ell,r)} \cdot \bv^C \geq 0$, and $=0$ exactly when $(\ell,r)$ is a contracted pair of $C$; also $\av^{(\ell,r)} \cdot \cv < 0$, so that Prop.~\ref{proposition:multivarres} is applicable.
We insert the residues of $F_{m+n}$ known from Lemma~\ref{lemma:fres}, observing however that the orientation of the hyperplanes $\zv\cdot\av_{\ell,r}=0$ is opposite to those in \eqref{resgenrec}, yielding a factor $(-1)^{|C|}$. In this way we obtain
\begin{equation}
F_{m+n}(\etav-i\zerov,\thetav+i\piv+i\zerov)
=
\sum_{C\in\ccal_{n,m}}\delta_{C} S_C R_C (\etav,\thetav) F_{m+n-2|C|}(\hat{\etav}+i\zerov,\hat\thetav+i\piv-i\zerov).
\end{equation}
(We note that for those sets of pairs $(\ell,r)$ that do not form a valid contraction, the corresponding residues vanish.) Now we use \cite[Lemma~3.10]{BostelmannCadamuro:expansion} to swap left with right indices in the contraction, replacing $\delta_{C} S_C R_C (\etav,\thetav)$ with $(-1)^{|C|}\delta_{C} S_C R_C(\thetav,\etav)$.
Together with \eqref{eq:fpi} and \eqref{eq:fmnbydef}, we arrive at
\begin{equation}
 F^{\pi}_{m+n}(\thetav+i\zerov,\etav+i\piv-i\zerov)
=
\sum_{C\in\ccal_{m,n}} (-1)^{|C|}\delta_{C} S_C R_C(\thetav,\etav) \cme{n-|C|,m-|C|}{A}(\hat{\etav},\hat\thetav).
\end{equation}
But in view of~\cite[Prop.~3.11]{BostelmannCadamuro:expansion}, the right-hand side is just $\cme{m,n}{J A^\ast J}(\thetav,\etav)$, which concludes the proof.
\cmpqed\end{proof}

This finally allows us to conclude that $A$ is local in a double cone.

\begin{theopargself}
\begin{proof}[\ofwhat{Theorem~\ref{theorem:doubleconeequiv}\ref{it:doubleconeequiv-fa}}]
Let $F_k$ fulfill (FD), and let $A$ be the quadratic form of Thm.~\ref{theorem:wedgeequiv}\ref{it:wedgeequiv-fa}. Then $A$ is $\omega$-local in $\wcal_r'$, and both parts of \eqref{eq:fmnboundary-d} hold, cf.~Prop.~\ref{proposition:fshifted}. Now with $F_k$, also $F_k^\pi$ fulfill (FD) and hence (FW). Theorem~\ref{theorem:wedgeequiv}\ref{it:wedgeequiv-fa} applied to $F_k^\pi$ yields another quadratic form $A^\pi$ which is $\omega$-local in $\wcal_r'$. Due to \eqref{eq:fmnfromF} for $A^\pi$ and Prop.~\ref{proposition:fshifted}, we have $A^\pi=J A\st J$, since they agree in all $\cme{m,n}{\cdotarg}$.
Thus $A$ is $\omega$-local in $\wcal_r' \cap \wcal_{-r} = \ocal_r$. That is, it fulfills (AD).
\cmpqed\end{proof}
\end{theopargself}

%% file: conclusions.tex
\section{Conclusions and outlook}\label{sec:conclusion}

In this article, we have found a one-to-one characterization of local observables (quadratic forms) $A$ in terms of the analyticity properties of their expansion coefficients $\cme{m,n}{A}$. To that end, symmetry properties and pole structure of the coefficients, which are expected from the form factor programme, needed to be complemented with specific bounds on the analytic functions and on their boundary value distributions. 
The characterizing conditions depend on the localization region (we discussed wedges and double cones) and on the high energy behaviour of $A$. They do however not require that the observable is of a specific internal structure, such as being a smeared pointlike field.

Our present results, in the form stated, are valid for integrable models that contain only one scalar species of particle. This restriction was chosen mainly to simplify the discussion. But we certainly expect that an analogous result can be obtained in models with richer particle structure, incorporating any finite number of particle species, allowing also for nonzero spin and for gauge symmetries, as long as the two-particle $S$ matrix is analytic in the physical strip. This would make use of the multi-component wedge-local fields established in \cite{LechnerSchuetzenhofer:2012}.

We would see potential applications of our characterization result mainly in two directions: for constructing local observables, and conversely, for establishing restrictions on their existence.
In that context, we recall that in a number of integrable models, it is currently unclear whether local observables in double cones exist at all. These models have been established in terms of wedge algebras (Borchers triples), but the size of local observable algebras, which are intersections of two wedge algebras, is often unknown. Indeed, even in the present class of S-matrices with a simplified single particle space, the existence of local observables was shown in \cite{Lechner:2008} only under certain extra assumptions on the scattering function, and is in general unclear.
Sufficient criteria have been established for showing \emph{existence} of local observables, e.g., the split property for wedge inclusions or modular nuclearity (cf.~\cite{Lechner:2008}), but few if any criteria in the opposite direction are known.

Now on the one hand, our results might be the basis for the explicit construction of local operators. That is, one would explicitly determine a sequence of analytic functions $F_k$ that fulfill all locality conditions, and investigate the corresponding observable $A$. In fact, in a number of models, candidates for these functions $F_k$ are known; see, e.g., \cite{FringMussardoSimonetti:1993} for the sinh-Gordon model or \cite{SchroerTruong:1978} for the massive Ising model. One strategy for their construction is as follows. One first computes the \emph{minimal solution} of the model, essentially a function $F_\mathrm{min}$ of one variable such that $\zetav \mapsto F_\mathrm{min}(\zeta_2-\zeta_1)$ fulfills \ref{it:fdmero}--\ref{it:fdperiod} for $k=2$. In the simple case of the Ising model ($S=-1$), one can choose $F_\mathrm{min}(\zeta)=-i \sinh (\zeta/2)$, whereas in other situations $F_\mathrm{min}$ can be computed from $S$ by a certain integral expression \cite{Karowski:1979}. One then makes the ansatz
\begin{equation}
    F_k (\zeta) = Q_k(e^{\zeta_1},\ldots,e^{\zeta_k}) \prod_{m<n} F_\mathrm{min}(\zeta_n-\zeta_m)
\end{equation}
with symmetric functions $Q_k$ of suitable analyticity. These $F_k$ automatically satisfy \ref{it:fdmero}--\ref{it:fdperiod} for all $k$. Now one chooses $Q_k$ recursively, starting from given $Q_1$, $Q_2$, so that the relations \ref{it:fdrecursion} are fulfilled. For pointlike localized observables $A$, \cite{SchroerTruong:1978,FringMussardoSimonetti:1993} obtain $Q_k$ as certain rational functions. The bounds \ref{it:fdboundsreal} and \ref{it:fdboundsimag} should then also hold for $F_k$ after smearing $A$ with a test function from a suitable class. The challenge is now to show that the quadratic forms $A$ so constructed extend to closable (unbounded) operators; in other words, one needs to control the 
domain of $A$. We have not investigated the last mentioned aspect here, but sufficient criteria can be found, and they can be verified at least in simple examples from the Ising model \cite{Cadamuro:2012}. We intend to return to this point elsewhere \cite{BostelmannCadamuro:examples-wip}.

This constructive application might also lead to a deeper understanding of concrete, physically important local quantities such as the energy density. With Fewster \cite{BostelmannCadamuroFewster:2013}, the present authors  have recently established lower bounds for the energy density (``quantum energy inequalities'') in the massive Ising model; it will be interesting to see whether similar properties pertain to a larger class of integrable models.

On the other hand, our results might be employed to prove \emph{non-}existence of local observables, in the sense of a ``no go theorem''.  That is, one could aim to show that for certain functions $S$, the conditions \ref{it:fdmero}--\ref{it:fdboundsimag} are incompatible, or more generally that they yield constraints on the size of the local algebras. An example might be provided by the ``exotic'' S-matrix $S(\theta)=\exp (a \sinh \theta)$.  Variants of the present method could also be used to clarify the corresponding question in massless models \cite{BostelmannLechnerMorsella:2011}, where it is open in general, or in higher-dimensional generalizations \cite{BuchholzSummers:2007} where the existence of observables in bounded regions is not expected, but where to the knowledge of the authors it has not rigorously been ruled out either.

%% file: acknowledgements.tex

The authors are indebted to K.~H.~Rehren for valuable suggestions.
They are also grateful to J.~Bros for discussions and to G.~Lechner for comments on a draft version.
D.~C.~would further like to thank M.~Bischoff, W.~Dybalski and Y.~Tanimoto for helpful discussions.